\pdfoutput=1
\documentclass[twoside,leqno,twocolumn]{article}

\usepackage[letterpaper]{geometry}

\usepackage{siamproceedings}
\usepackage{siunitx}

\usepackage[T1]{fontenc}
\usepackage{amsfonts}
\usepackage{graphicx}
\usepackage{epstopdf}
\usepackage{enumitem}
\usepackage{algorithmic}
\ifpdf
  \DeclareGraphicsExtensions{.eps,.pdf,.png,.jpg}
\else
  \DeclareGraphicsExtensions{.eps}
\fi

\newsiamremark{remark}{Remark}
\newsiamremark{hypothesis}{Hypothesis}
\crefname{hypothesis}{Hypothesis}{Hypotheses}
\newsiamthm{claim}{Claim}

\usepackage{amsopn}

\usepackage{wrapfig}
\usepackage{nccmath}
\usepackage{nicematrix}
\usepackage[linewidth=1pt]{mdframed}
\usepackage{mathtools}

\usepackage{pgfplots}
\usepackage{tikz}
\usepackage{subcaption}
\usepackage{stfloats}

\definecolor{my-dark-red}{RGB}{183, 28, 28}
\definecolor{my-red}{RGB}{244,67,54}
\definecolor{my-pink}{RGB}{233,30,99}
\definecolor{my-purple}{RGB}{156,39,176}
\definecolor{my-deep-purple}{RGB}{103,58,183}
\definecolor{my-indigo}{RGB}{63,81,181}
\definecolor{my-blue}{RGB}{33,150,243}
\definecolor{my-light-blue}{RGB}{3,169,244}
\definecolor{my-cyan}{RGB}{0,188,212}
\definecolor{my-teal}{RGB}{0,150,136}
\definecolor{my-green}{RGB}{76,175,80}
\definecolor{my-light-green}{RGB}{139,195,74}
\definecolor{my-lime}{RGB}{205,220,57}
\definecolor{my-yellow}{RGB}{255,235,59}
\definecolor{my-amber}{RGB}{255,193,7}
\definecolor{my-orange}{RGB}{255,152,0}
\definecolor{my-deep-orange}{RGB}{255,87,34}
\definecolor{my-brown}{RGB}{121,85,72} 
\definecolor{my-grey}{RGB}{120,120,120}
\definecolor{my-blue-grey}{RGB}{96,125,139}
\definecolor{my-lipics-grey}{rgb}{0.6,0.6,0.61}

\tikzset{
  empty/.style={draw=none,mark=x,every mark/.append style={mark size=0.0pt}},
  brindex/.style={mark=*,my-yellow,semithick,draw=black,every mark/.append style={mark size=2.5pt}},
  bmove/.style={mark=square*,my-blue,semithick,draw=black,every mark/.append style={mark size=2.5pt}},
  columba/.style={mark=pentagon*,my-purple,semithick,draw=black,every mark/.append style={mark size=2.8pt}},
  moverb/.style={mark=diamond*,my-green,semithick,draw=black,every mark/.append style={mark size=3.2pt}},
  moverbrlzsa/.style={mark=triangle*,my-red,semithick,draw=black,every mark/.append style={mark size=3.2pt}},
  marks/.style={
    every mark/.append style={mark size=3pt},
  },
  legendmarks/.style={
    marks,
    every plot/.append style={only marks}
  },
  chart/.style={
    legend label/.style={font={\footnotesize},anchor=west,align=left},
    legend box/.style={rectangle, draw, minimum size=5pt},
    axis/.style={black,semithick,->},
    axis label/.style={anchor=east,font={\footnotesize}},
  }
}

\pgfplotsset{
  title style={yshift=-5pt, font=\bfseries, inner sep=0pt},
  major grid style={thin,dotted,color=my-blue-grey},
  minor grid style={thin,dotted,color=my-grey!45},
  ymajorgrids,
  yminorgrids,
  xmajorgrids,
  xminorgrids,
  every tick label/.append style={font=\footnotesize},
  every axis label/.append style={font=\footnotesize},
  legend cell align=left,
  legend pos=outer north east,
  scaled x ticks=false,
  scaled y ticks=false,
  max space between ticks=30,
  minor tick num=2,
  every axis/.append style={
    line width=0.5pt,
    tick style={
      line cap=round,
      thin,
      major tick length=4pt,
      minor tick length=2pt,
    },
  },
  xticklabel style={
    /pgf/number format/fixed,
    /pgf/number format/precision=3
  },
  yticklabel style={
    /pgf/number format/fixed,
    /pgf/number format/precision=3
  },
  querystyle_apm/.style={
    scale only axis,
    width=4.6cm,
    height=3.5cm,
    yticklabel style={text width=0.72cm, align=right},
    title style={at={(0.5,1.03)}},
    x label style={at={(0.5,0.02)}},
    y label style={at={(0.1,0.5)}},
    align=center,
    xmin=10,
    xmax=1280,
    log basis x=2,
    log basis y=10,
    ytick={1e-10,1e-9,1e-8,1e-7,1e-6,1e-5,1e-4,1e-3,1e-2,1e-1,1e0,1e1,1e2,1e3,1e4,1e5,1e6,1e7,1e8,1e9,1e10,1e11,1e12},
    minor ytick={2e-10,3e-10,4e-10,5e-10,6e-10,7e-10,8e-10,9e-10,2e-9,3e-9,4e-9,5e-9,6e-9,7e-9,8e-9,9e-9,2e-8,3e-8,4e-8,5e-8,6e-8,7e-8,8e-8,9e-8,2e-7,3e-7,4e-7,5e-7,6e-7,7e-7,8e-7,9e-7,2e-6,3e-6,4e-6,5e-6,6e-6,7e-6,8e-6,9e-6,2e-5,3e-5,4e-5,5e-5,6e-5,7e-5,8e-5,9e-5,2e-4,3e-4,4e-4,5e-4,6e-4,7e-4,8e-4,9e-4,2e-3,3e-3,4e-3,5e-3,6e-3,7e-3,8e-3,9e-3,2e-2,3e-2,4e-2,5e-2,6e-2,7e-2,8e-2,9e-2,2e-1,3e-1,4e-1,5e-1,6e-1,7e-1,8e-1,9e-1,2e0,3e0,4e0,5e0,6e0,7e0,8e0,9e0,2e1,3e1,4e1,5e1,6e1,7e1,8e1,9e1,2e2,3e2,4e2,5e2,6e2,7e2,8e2,9e2,2e3,3e3,4e3,5e3,6e3,7e3,8e3,9e3,2e4,3e4,4e4,5e4,6e4,7e4,8e4,9e4,2e5,3e5,4e5,5e5,6e5,7e5,8e5,9e5,2e6,3e6,4e6,5e6,6e6,7e6,8e6,9e6,2e7,3e7,4e7,5e7,6e7,7e7,8e7,9e7,2e8,3e8,4e8,5e8,6e8,7e8,8e8,9e8,2e9,3e9,4e9,5e9,6e9,7e9,8e9,9e9,2e10,3e10,4e10,5e10,6e10,7e10,8e10,9e10,2e11,3e11,4e11,5e11,6e11,7e11,8e11,9e11},
    enlarge x limits=0.05,
    enlarge y limits=0.08,
    xtick={10,20,40,80,160,320,640,1280},
    xticklabels={10,20,40,80,160,320,640,1280},
    ymajorgrids=true,
    yminorgrids=true,
  },
  querystyle_apm_same_alg/.style={
    querystyle_apm,
    width=4.6cm,
    height=2.75cm,
  },
  querystyle_raw_performance/.style={
    querystyle_apm,
    width=4.7cm,
    height=3.25cm,
  },
  querystyle_kfig/.style={
    querystyle_apm,
    width=4.5cm,
    height=2.75cm,
  },
  constructionstyle/.style={
    align=center,
    title style={at={(0.5,1.03)}},
    x label style={at={(0.5,0.02)}},
    y label style={at={(0.1,0.5)}},
    width=6.35cm,
    height=4.9cm,
    max space between ticks=20,
  },
  legend/.style={
    width=2cm,
    height=2cm,
    xmin=10,xmax=10,
    ymin=10,ymax=10,
    ymajorgrids=false,
    xtick=\empty,
    ytick=\empty,
    x axis line style={-,draw opacity=0},
    y axis line style={-,draw opacity=0},
    xtick style={draw=none},
    ytick style={draw=none},
    legend style={
        font=\footnotesize,
        anchor=north,
        at={(0,0)},
        /tikz/every even column/.append style={column sep=0.1cm}
    },
  },
}

\newcommand{\Oh}{\mathcal{O}}

\newcommand{\LF}{\mathsf{LF}}
\newcommand{\SA}{\mathsf{SA}}
\newcommand{\sa}{\mathsf{sa}}
\newcommand{\p}{\mathsf{p}}

\newcommand{\idx}{\mathsf{idx}}

\newcommand{\peak}{\mathsf{peak}}
\newcommand{\iM}{\mathsf{M}}      
\newcommand{\iB}{\mathsf{B}}      
\newcommand{\iR}{\mathsf{R}}      
\newcommand{\isz}[1]{\idx_{#1}}   
\newcommand{\pk}[1]{\peak_{#1}}   
\newcommand{\pks}[1]{\peak_{#1}^{*}} 

\newcommand{\MDS}{\mathcal{M}}
\newcommand{\MLF}{\MDS^{\LF}}
\newcommand{\MLFbwd}{\MDS^{\bwd{\LF}}}
\newcommand{\MPhi}{\MDS^{\Phi}}
\newcommand{\MPhiinv}{\MDS^{\Phi^{-1}}}
\newcommand{\move}{\mathsf{move}}
\newcommand{\rank}{\mathsf{rank}}
\newcommand{\select}{\mathsf{select}}

\newcommand{\pred}{\mathsf{pred}}
\newcommand{\suc}{\mathsf{succ}}
\newcommand{\pos}{\mathsf{pos}}
\newcommand{\fwd}[1]{\overrightarrow{#1}}
\newcommand{\bwd}[1]{\overleftarrow{#1}}
\newcommand{\ceil}[1]{\lceil #1 \rceil}
\newcommand{\floor}[1]{\lfloor #1 \rfloor}

\newcommand{\PLCP}{\mathsf{PLCP}}
\newcommand{\LIF}[2]{\STATE \textbf{if} #1 \textbf{then} #2;}

\newcommand{\LELSEIF}[2]{\STATE \textbf{else if} #1 \textbf{then} #2;}
\newcommand{\LELSE}[1]{\STATE \textbf{else} #1;}

\newcommand{\edit}{\mathsf{edit}}
\newcommand{\Occ}{\mathsf{Occ}}
\newcommand{\dom}{\mathsf{dom}}

\makeatletter
\newcommand{\ELSEIFONLY}[1]{%
  \STATE \algorithmicelse\ \algorithmicif\ #1\ \algorithmicthen
  \begin{ALC@g}%
}
\newcommand{\ENDELSEIFONLY}{%
  \end{ALC@g}%
}
\makeatother

\begin{document}

\newcommand\relatedversion{}

\title{\Large Move-rb: Faster Bi-Directional r-indexes and Approximate Pattern Matching\relatedversion}
\author{
Johannes Fischer\thanks{TU Dortmund (\email{johannes.fischer@cs.tu-dortmund.de}).}
\and Lukas Nalbach\thanks{TU Dortmund (\email{lukas.nalbach@tu-dortmund.de}).}
}

\date{}

\maketitle

\fancyfoot[C]{\small\thepage}
\pagenumbering{arabic}
\setcounter{page}{1}

\begin{abstract}
Approximate pattern matching (APM) on highly repetitive texts --- such as collections of genomes from the same species --- is a central task in bioinformatics.
Bi-directional r-indexes support both left- and right-extension of a pattern and thereby accelerate APM algorithms based on search schemes, but existing variants --- br-index (\cite{br_index} Arakawa et al., 2022) and b-move (\cite{b_move} Depuydt et al., 2025) --- suffer from two bottlenecks: the $\Oh(\sigma)$ character-predecessor/-successor queries on the run-length-encoded BWT needed for an extension, and the predecessor queries on sparse bit vectors used to maintain a value in the suffix array interval and to access the PLCP array while locating.
We present Move-rb, a bi-directional r-index built on the optimized r-index Move-r.
Although Move-rb is 2$\times$ larger than br-index, it is up to 24\% smaller than b-move, answers APM queries 1--4 orders of magnitude faster than br-index, 1.9--10$\times$ faster than b-move and even up to 5.5$\times$ faster than the state-of-the-art bi-directional (uncompressed) FM-index columba (\cite{columba_dynamic_partitioning, min_u, columba_in_text_verification} Renders et al., 2021--2024), which is 36--42$\times$ larger.
Memory usage (including index size) during APM locate queries is reduced by $1.5\times$ (up to $6.5\times$) for Hamming distance and $2.5\times$ (up to $7.4\times$) for edit distance.
Move-rb can be constructed 3--14$\times$ faster while using 19--141$\times$ less memory than br-index, b-move and columba.
A variant using a relative Lempel--Ziv (RLZ)-encoded suffix array (\cite{rlzsa} Dinklage et al., 2025) locates up to 10$\times$ faster while being 1.2--2.4$\times$ larger.
We achieve these speedups by theoretically and practically optimizing index operations and search scheme APM algorithms:
Without a direction switch, Move-rb computes an all-$k$-character extension in output-optimal $\Oh(k)$ time and a single-character extension in the same time, by avoiding the character-predecessor/-successor queries needed by previous bi-directional r-indexes.
Augmenting Move-rb with $\Oh((r + \bwd{r}) \log(n / r_{\min}) \log \sigma)$ bits reduces a single-character extension to $\Oh(\log \sigma)$ time, where $r$ and $\bwd{r}$ are the numbers of equal-letter runs in the BWT of the text and its reverse, resp., and $r_{\min}=\min(r, \bwd{r})$.
A direction switch incurs only $\Oh(\log\log_\omega (n / r_{\min}))$ additional time.
\end{abstract}

\section{Introduction}
Answering approximate pattern matching queries (APM) on repetitive texts is a common task in bioinformatics, in particular when indexing DNA (assembled or unassembled) from the same species. There, it is important to exploit the repetitiveness of the data instead of storing it uncompressed.
The r-index \cite[Lemma~2.1, Theorem~3.6]{r_index} uses $\Oh(r \log n)$ bits of space, where $r$ is the number of equal-letter runs in the Burrows-Wheeler Transformation (BWT) \cite{bwt}, an accepted measure for compressibility of highly similar texts \cite{repetitiveness_measures}.

\begin{table*}[t]
    \centering
    \resizebox{\linewidth}{!}{
    \setlength{\tabcolsep}{2pt}
    \begin{tabular}{l|c|c|c}
        index & space (bits) & single-character extension & all-$k$-character extension \\
        \hline
        br-index & $\Oh((r + \bwd{r}) \log n)$ &
        $\Oh(\sigma \log\log_\omega \sigma + \log\log_\omega (n / r_{\min}))$ &
        $\Oh(\sigma \log\log_\omega \sigma + k\log\log_\omega (n / r_{\min}))$ \\

        br-index & $\Oh((r + \bwd{r}) \log n)$ &
        $\Oh((1 / \epsilon) \log^{2+\epsilon} r_{\max})$ &
        $\Oh(\sigma \log\log_\omega \sigma + k\log\log_\omega (n / r_{\min}))$ \\

        b-move & $\Oh((r + \bwd{r}) \log n)$ &
        $\Oh(\sigma r_{\max} + \log\log_\omega (n / r_{\min}))$ &
        $\Oh(\sigma r_{\max} + k\log\log_\omega (n / r_{\min}))$ \\

        Move-rb & $\Oh((r + \bwd{r}) \log n)$ &
        $\Oh(k + \underline{\log\log_\omega (n / r_{\min})})$ &
        $\Oh(k + \underline{\log\log_\omega (n / r_{\min})})$ \\

        Move-rb & $\Oh((r + \bwd{r})(\log n + \log(n / r_{\min}) \log \sigma))$ &
        $\Oh(\log \sigma + \underline{\log\log_\omega (n / r_{\min})})$ &
        $\Oh(k + \underline{\log\log_\omega (n / r_{\min})})$ \\
    \end{tabular}
    }
    \caption{
        Running times and space for single- and all-$k$-character extensions, where $k \leq \sigma$ is the number of characters that can extend the pattern, $r_{\min}=\min(r, \bwd{r})$ and $r_{\max}=\max(r, \bwd{r})$.
        Underlined terms are only incurred if the extension performs a direction switch.
    }
    \label{tab:running_times_extensions}
\end{table*}

\begin{table*}[t]
    \centering
    \setlength{\tabcolsep}{3pt}
    \begin{tabular}{l|c|c}
        index & space (bits) & $m$ single-character extensions ($s$ switches) + locate \\
        \hline
        br-index & $\Oh((r + \bwd{r}) \log n)$ &
        $\Oh(m \sigma \log\log_\omega \sigma + m \log\log_\omega (n / r_{\min}) + occ)$ \\
        
        br-index & $\Oh((r + \bwd{r}) \log n)$ &
        $\Oh((m / \epsilon) \log^{2+\epsilon} r_{\max} + occ)$ \\
        
        b-move & $\Oh((r + \bwd{r}) \log n)$ &
        $\Oh(m \sigma r_{\max} + m \log\log_\omega (n / r_{\min}) + occ \log\log_\omega (n / r))$ \\
        
        Move-rb & $\Oh((r + \bwd{r}) \log n)$ &
        $\Oh(\sum\nolimits_{i=1}^m k_i + s \log\log_\omega (n / r_{\min}) + occ)$ \\

        Move-rb & $\Oh((r + \bwd{r})(\log n + \log(n / r_{\min}) \log \sigma))$ &
        $\Oh(m \log \sigma + s \log\log_\omega (n / r_{\min}) + occ)$ \\
    \end{tabular}
    \caption{
        Running times and space for $m$ single-character extensions with $s$ switches and locating, where $k_i \leq \sigma$ is the number of characters that can extend the pattern before extension $i$; remaining notation as in \Cref{tab:running_times_extensions}.
    }
    \label{tab:running_times_locate}
\end{table*}

In the r-index, the bottlenecks are $\rank$-queries on the BWT and consecutive $\Phi$-queries, which are implemented using predecessor ($\pred$) queries on sparse bit vectors.
The r-index can prepend a character to a searched pattern (left-extension) in $\Oh(\log\log_\omega (\sigma + n / r))$ time.
Locating all occurrences of the searched pattern requires additional $\Oh(occ \log\log_\omega (n / r))$ time.

Nishimoto and Tabei \cite[Lemma~2, Theorems~7--9]{move_data_structure} showed that the $\rank$-queries on the BWT can be replaced by a combination of (\textbf{i}) character-predecessor and -successor ($\pred_c$ and $\suc_c$) queries on the string $L'$ storing the BWT runs ($\Oh(\log\log_\omega \sigma)$ time), and (\textbf{ii}) consecutive $\LF$-queries ($\Oh(1)$ time).
They also presented the move data structure, which can answer consecutive $\LF$- and $\Phi$-queries in $\Oh(1)$ time.
Thus, the time for left-extension is reduced to $\Oh(\log\log_\omega \sigma)$, and the additional time for locate is reduced to $\Oh(occ)$.

While the r-index is uni-directional, i.e.,\ it only allows for left-extension, the bi-directional r-index (br-index) \cite[Theorems~1 and~2]{br_index} also allows for right-extension.
This enables us to start by matching an inner part of the pattern and then extend it to the left and right, which significantly accelerates APM algorithms based on search schemes \cite{search_schemes}.
In the br-index, one such extension takes $\Oh(\sigma \log\log_\omega \sigma + \log\log_\omega (n / r_{\min}))$ time, where $\bwd{r}$ is the number of runs in the BWT of the reverse string $\bwd{T}$ and $r_{\min}=\min(r, \bwd{r})$.
The bottlenecks are $\Oh(\sigma)$ $\pred_c$- and $\suc_c$-queries on $L'$, and $\pred$-queries on sparse bit vectors involved in maintaining a value in the $\SA$-interval.
They also present an impractical variant of the br-index that achieves $\Oh((1 / \epsilon) \log^{2 + \epsilon} r_{\max})$ time for any $\epsilon > 0$ (see \Cref{tab:running_times_extensions,tab:running_times_locate}).

Although the theoretical description of the br-index uses move data structures for $\LF$ and $\Phi$, its implementation uses those of the r-index.
Furthermore, instead of their $\Oh(1)$ time solution for accessing the $\PLCP$ array during the locate phase, it employs a generic data structure relying on $\pred$-queries on a sparse bit vector, which incurs $\Oh(\log\log_\omega (n / r))$ additional time per occurrence.
Overall, these drawbacks result in worse running times in practice: an extension takes $\Oh(\sigma \log\log_\omega (\sigma + n / r_{\min}))$ time, and matching a pattern of length $m$ and locating all $occ$ occurrences takes $\Oh(m \sigma \log\log_\omega (\sigma + n / r_{\min}) + occ \log\log_\omega (n / r))$ time.

b-move \cite{b_move} is another bi-directional r-index, which uses move data structures for $\LF$, $\Phi$ and $\Phi^{-1}$ also in practice.
However, $\pred_c$- and $\suc_c$-queries on $L'$ are answered naively using scanning ($\Oh(r)$ time).
Locating still takes $\Oh(occ \log\log_\omega (n / r))$ time, because b-move relies on the same generic data structure for accessing $\PLCP$.
Furthermore, b-move only supports input texts with a DNA alphabet.
Finally, b-move's APM algorithms \cite{columba_dynamic_partitioning, columba_in_text_verification} leave room for improvement, as they enumerate duplicate and nested $\SA$-intervals and often incur high memory usage.

\textbf{Our Contributions.}
We present a text index that we call Move-rb, which builds upon the optimized r-index Move-r \cite{move_r} and improves several theoretical and practical aspects of existing bi-directional r-indexes:

\begin{itemize}
    \item We bypass the $\Oh(\sigma)$ $\pred_c$- and $\suc_c$-queries on $L'$ required by previous bi-directional r-indexes and show that the $\pred$-queries on sparse bit vectors (for maintaining a value in the $\SA$-interval) can be avoided.
    As a result, without a direction switch, Move-rb performs an all-$k$-character extension in output-optimal $\Oh(k)$ time and a single-character extension in the same time, both within $\Oh((r + \bwd{r}) \log n)$ bits of space.
    A direction switch incurs only $\Oh(\log\log_\omega (n / r_{\min}))$ additional time.
    \item Augmenting Move-rb with a specialized wavelet tree \cite{wavelet_tree}, requiring $\Oh((r + \bwd{r}) \log(n / r_{\min}) \log \sigma)$ further bits, reduces a single-character extension without a direction switch to $\Oh(\log \sigma)$ time, matching uncompressed bi-directional FM-indexes \cite{two_bwt, bidirectional_fm_index}, which no bi-directional r-index has attained before.
    \item Move-rb can be constructed 10--14$\times$ faster than br-index, 3--4.6$\times$ faster than b-move and 3.7$\times$ faster than the state-of-the-art bi-directional (uncompressed) FM-index columba \cite{columba_dynamic_partitioning, min_u, columba_in_text_verification}, while requiring 36--141$\times$, 19--24$\times$ and 124--129$\times$ less memory, resp.
    \item Beyond these index operations, we also optimize the state-of-the-art search scheme APM algorithms provided by b-move \cite{columba_dynamic_partitioning, columba_in_text_verification}: We avoid enumerating duplicate and nested $\SA$-intervals and employ further practical optimizations (see \Cref{sct:apm_algorithms}).
    \item As a result, Move-rb answers APM queries 1--4 orders of magnitude faster than br-index, 1.9--10$\times$ faster than b-move and even up to 5.5$\times$ faster than columba, while being 2$\times$ larger than br-index, up to 24\% smaller than b-move and 36--42$\times$ smaller than columba.
    \item We reduce peak memory usage (including index size) during APM locate queries compared to b-move by factors $1.5\times$ (up to $6.5\times$) and $2.5\times$ (up to $7.4\times$) for Hamming distance and edit distance, resp.
    \item We implemented a variant of Move-rb using an RLZ-encoded suffix array (RLZSA), further improving locate performance by up to 10$\times$ while being 1.2--2.4$\times$ larger than Move-rb.
\end{itemize}

In the most efficient currently known search schemes used in APM algorithms, \emph{minU} \cite{min_u} (for $k \leq 13$) and \emph{suffix-filter} \cite{suffix_filter} (for arbitrary $k$), each search only performs $\leq 2$ direction switches, in which case our index performs particularly well, because the factor $s$ before the $\log\log_\omega (n / r_{\min})$ term disappears (see \Cref{tab:running_times_locate}).

\subsection{Strings}
An \emph{interval} $[i, j]$ describes the set $\{i, \ldots, j\}$.
An \emph{alphabet} $\Sigma$ is a finite ordered set of \emph{symbols} of size $|\Sigma| = \sigma$.
A string $T \in \Sigma^*$ with $|T| = n$ is a sequence of symbols in the alphabet $\Sigma$.
If $|T| = 0$, then $T$ is the empty string $\epsilon$, and $T[i]$ denotes the $i$-th symbol of $T$.
A \emph{substring} is $T[i,j] = T[i]T[i+1]\cdots T[j]$, where $T[i,j] = \epsilon$ if $i > j$.
We call $T[1,i]$ the $i$-th \emph{prefix} and $T_i = T[i,n]$ the $i$-th \emph{suffix} of $T$.
We assume an order $c_1 < c_2 < \cdots < c_\sigma$ on the alphabet $\Sigma = \{c_1, c_2, \ldots, c_\sigma\}$.
The \emph{lexicographic order} of strings is then defined by $T < T' \Leftrightarrow T \text{ is a proper prefix of } T' \lor \exists i \in [0, \min(|T|, |T'|) - 1]: T[1,i] = T'[1,i] \land T[i+1] < T'[i+1]$.
For $i \in [1, n]$, let $\rank_c(T, i)$ be the number of occurrences of $c$ in $T[1, i]$. Let $\rank_{<c}(T, i)$ be the number of occurrences of all $\alpha < c$ in $T[1, i]$. We additionally define $\rank_c(T, 0) = \rank_{<c}(T, 0) = 0$.
For an integer $i \geq 1$, let $\select_c(T, i) = j$, where $j$ is the position of the $i$-th occurrence of $c$ in $T$ if it exists, and $j = \infty$, else. For $i = 0$, let $\select_c(T, i) = -\infty$.
For $i \in [1, n]$ and $c \in \Sigma$, we call $\pred_c(T, i) = \select_c(T, \rank_c(T, i))$ and $\suc_c(T, i) = \select_c(T, \rank_c(T, i - 1) + 1)$ the \emph{character-predecessor/-successor} of $c$ before/after position $i$ in $T$, resp.
Given a sorted array $A[1..N]$, we call the positions $\pred(A, i) = \max \{j \in [1, N] \mid A[j] \leq i\}$ and $\suc(A, i) = \min \{j \in [1, N] \mid A[j] \geq i\}$ the \emph{predecessor} and \emph{successor} of $i$ in $A$, resp.
We assume that all described algorithms work in the \emph{word-RAM model} with word size $\omega$ \cite{word_ram}.
In the following, $T \in \Sigma^*$ is terminated by a \emph{sentinel} symbol \$ that is lexicographically smaller than all other symbols in $\Sigma$, which simplifies our algorithms.
We denote with $\bwd{T}$ the reverse string $\bwd{T} = T[n - 1]T[n - 2]\cdots T[1]$\$. The reverse of a regular string $S$ (without the terminator symbol) is defined as $\bwd{S} = S[|S|]S[|S| - 1]\cdots S[1]$.
Let $D$ be a data structure defined w.r.t.\ $T$.
Then $\bwd{D}$ denotes the same data structure defined w.r.t.\ $\bwd{T}$.

The \emph{suffix array} $\SA$ \cite{suffix_array} of a string $T$ consists of the starting positions of all suffixes of $T$ in their lexicographic order, i.e.,\ $T_{\SA[1]} < T_{\SA[2]} < \cdots < T_{\SA[n]}$. For a pattern $P \in \Sigma^*$, there is a maximum interval in the suffix array (the \emph{suffix array interval} of $P$) containing exactly the positions of all occurrences of $P$ in $T$.

\emph{Indexed pattern matching} is the approach of using a text index over $T$ to search for $P$.
A bi-directional text index supports bi-directional pattern search, i.e.,\ given the $\SA$-interval of a pattern $P$ and the $\bwd{\SA}$-interval of $\bwd{P}$, we can efficiently compute the $\SA$- and $\bwd{\SA}$-interval of $cP$ and $Pc$ (left- and right-extension).
In order to locate the occurrences of a searched pattern, the index must also support enumeration of the maintained $\SA$-interval.

\subsection{Burrows-Wheeler Transform and Backward Search}\label{sct:bwt_bwsearch}
Bi-directional text indexes are based on the \emph{Burrows-Wheeler Transform} (BWT), defined as follows.
The $i$-th \emph{rotation} of a string $T$ is $T[i,n]T[1,i-1]$, and the \emph{rotation matrix} (BWM) is the $n \times n$ matrix whose $i$-th row is the $i$-th rotation of $T$.
If we sort the BWM lexicographically, its last column $L$ is the BWT of $T$ and its first column is denoted by $F$.

If we store the frequency array $C[1..\sigma]$ with $C[c] = |\{i \in [1,n]\, |\ T[i] < c\}|$ and a data structure to compute $\rank_c(L, i)$ \cite{wavelet_tree}, then we can implement left-extension:
Given $\SA$-interval $[b, e]$ of a pattern $P$, we can compute the $\SA$-interval $[b', e']$ of $cP$ for any $c \in \Sigma$ by $b' = C[c] + \rank_c(L, b - 1) + 1$ and $e' = C[c] + \rank_c(L, e)$.
The process of matching a pattern backwards using left-extensions is called backward search \cite{indexing_compressed_text}.

\subsection{Bi-Directional r-index}
In the BWT, equal symbols are often grouped into \emph{runs}, which we can exploit for compression by applying a \emph{run-length encoding} to it.
Let $L_1, L_2, \ldots, L_r$ denote those runs, i.e.,\ $r$ is the number of runs in $L$ and $L = L_1L_2 \cdots L_r$.
Then we can compress the BWT into the \emph{run-length encoded BWT} (RLBWT) $(L_1[1], |L_1|), (L_2[1], |L_2|), \ldots, (L_r[1], |L_r|)$.
The r-index \cite[Lemma~2.1]{r_index} uses this to implement left-extension in $\Oh(r)$ space and $\Oh(\log\log_\omega (\sigma + n / r))$ time.

\begin{figure*}
    \centering
    \includegraphics[width=0.77\linewidth]{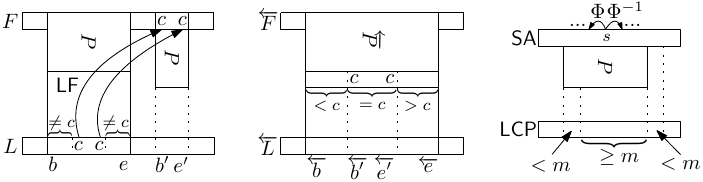}
    \caption{Illustration of left-extension in a bi-directional index (left and middle), and the locate algorithm of br-index (right).}
    \label{fig:br-index}
\end{figure*}

The bi-directional r-index (br-index) \cite{br_index} additionally stores an r-index over $\bwd{T}$, which enables us to switch between left- and right-extension if we also maintain the $\bwd{\SA}$-interval $[\bwd{b}, \bwd{e}]$ of $\bwd{P}$ in $\bwd{T}$.
The three main challenges here are maintaining $[b, e]$ and $[\bwd{b}, \bwd{e}]$, maintaining a value $s \in \SA[b, e]$, and enumerating $\SA[b, e]$ using $s$.

In order to discuss how the br-index achieves this, we first define the variables that are needed during bi-directional search.

\begin{definition}
    Let $P$ be a pattern of length $m$.
    We call $C = \langle m, b, e, s, \bwd{b}, \bwd{e}, \bwd{s} \rangle$ a valid context for $P$ iff $[b, e]$ is the $\SA$-interval of $P$, $[\bwd{b}, \bwd{e}]$ is the $\bwd{\SA}$-interval of $\bwd{P}$, $s \in \SA[b, e]$ and $\bwd{s} \in \bwd{\SA}[\bwd{b}, \bwd{e}]$.
\end{definition}

With this, we can now discuss the algorithm for left-extension (see \Cref{alg:bri-prepend}).
Right-extension works like left-extension w.r.t.\ $\bwd{T}$.

\subsubsection{Maintaining the SA- and $\overleftarrow{\text{SA}}$-Interval}
$[b', e']$ is computed in lines 1--2, as described in \Cref{sct:bwt_bwsearch}.
When computing $\bwd{\SA}$-interval $[\bwd{b'}, \bwd{e'}]$ of $\bwd{cP}$ w.r.t.\ $\bwd{T}$, the main observation is that $[\bwd{b'}, \bwd{e'}]$ is fully contained within $[\bwd{b}, \bwd{e}]$.
More precisely, we have $\bwd{b'} - \bwd{b} = |\{i \in [b, e] \mid L[i] < c\}|$, because the same symbols that prepend occurrences of $P$ in $T$, append occurrences of $\bwd{P}$ in $\bwd{T}$, and occur sorted in the range $[\bwd{b}, \bwd{e}]$ in the $(m + 1)$-st column of the BWM of $\bwd{T}$ (see middle in \Cref{fig:br-index}).
Thus, we can set $\bwd{b'} = \bwd{b} + \rank_{<c}(L, e) - \rank_{<c}(L, b - 1)$ and $\bwd{e'} = \bwd{b'} + e' - b'$.
Since we cannot directly compute $\rank_{<c}$-queries with the data structures of the r-index, we have to compute two $\rank_\alpha$-queries for each $\alpha \in \Sigma$ with $\alpha < c$ (see lines 7--11).
This amounts to $\Oh(\sigma \log\log_\omega \sigma)$ time for computing $[\bwd{b'}, \bwd{e'}]$.
The authors also propose another (purely theoretical) method, which achieves $\Oh((1 / \epsilon) \log^{2+\epsilon} r)$ time, still with $\Oh(r \log n)$ bits of space \cite[Theorem~2]{br_index}.

\subsubsection{Maintaining One Value in the SA-Interval}
If we only perform left-extensions, we can just maintain one $\SA$-value $s = \SA[e]$ like with the r-index.
This does not work if we allow right-extensions: if we have $x \leq b < b' \leq e' < e \leq y$ during a right-extension, where $L[x, y]$ is a run, then there is no stored sample in $\SA[b', e']$, hence we cannot maintain a value in $\SA[b', e']$ using only $\SA$-samples.
We can instead obtain such a value using $\bwd{\SA}$-samples: due to $e' - b' < e - b$, there is a run $\bwd{L}[x', y'] \neq c$ with $[x', y'] \cap [\bwd{b}, \bwd{e}] \neq \emptyset$.
Thus, there is at least one run in $\bwd{L}$ with character $c$ whose start- or end position lies in $[\bwd{b}, \bwd{e}]$.
Hence, we can instead obtain a value in $\bwd{\SA}[\bwd{b'}, \bwd{e'}]$ by storing $\SA$- and $\bwd{\SA}$-samples at run start- and end-positions using the following data structures (w.r.t.\ $T$ and $\bwd{T}$):
For each $c \in \Sigma$, we store the start and end positions of the $r_c$ runs with character $c$ in $L$ in a sorted array $R_c[1..2r_c]$ with $\pred$-support (e.g., a sparse bit vector).
Additionally, we store for each $c \in \Sigma$ an $\SA$-sample array $S_c[1..2r_c]$ with $S_c[i] = \SA[R_c[i]]$.
This enables us to compute $\bwd{s'} = \bwd{S_c}[\bwd{R_c}.\pred(\bwd{e})] - 1$ with $\bwd{s'} \in \bwd{\SA}[\bwd{b'}, \bwd{e'}]$ using the $\bwd{\LF}$-property.
Finally, we translate $\bwd{s'}$ back to a value $s' = n - \bwd{s'} - m \in \SA[b', e']$.
This corresponds to lines 12--13 in \Cref{alg:bri-prepend}, but with the directions flipped.
Finally, lines 3--5 handle the case where the $\bwd{\SA}$-interval remains the same.

Overall, we get the following running time for extensions, where the $\sigma$ part stems from the $\rank$-queries on $L'$ and the $n / r_{\min}$ term bounds the $\pred$-query on $R_c$, resp.\ $\bwd{R_c}$.

\begin{theorem}[{\cite[Theorem~1]{br_index}}]
    Let $P$ be a pattern of length $m$ and let $C$ be a valid context for $P$.
    Then there is a data structure of size $\Oh((r + \bwd{r}) \log n)$ bits, s.t.\ we can compute a valid context of $cP$ (left-extension) and $Pc$ (right-extension) from $C$ in $\Oh(\sigma \log\log_\omega \sigma + \log\log_\omega (n / r_{\min}))$ time.
\end{theorem}

\begin{algorithm}[!h]
\caption{$\mathsf{Left\text{-}Extension}(c, \langle m, b, e, s, \bwd{b}, \bwd{e}, \bwd{s} \rangle)$}\label{alg:bri-prepend}
\begin{algorithmic}[1]
\STATE $b' \gets C[c] + \rank_c(L, b - 1) + 1$;
\STATE $e' \gets C[c] + \rank_c(L, e)$;
\IF{$e' - b' = e - b$}
    \STATE $\bwd{b'} \gets \bwd{b};\ \bwd{e'} \gets \bwd{e}$;
    \STATE $s' \gets s - 1;\ \bwd{s'} \gets \bwd{s}$;
\ELSE
    \STATE $\bwd{b'} \gets \bwd{b}$;
    \FOR{$\alpha \gets 1$ \textbf{to} $c - 1$}
        \STATE $\bwd{b'} \gets \bwd{b'} + \rank_\alpha(L, e) - \rank_\alpha(L, b - 1)$;
    \ENDFOR
    \STATE $\bwd{e'} \gets \bwd{b'} + e' - b'$;
    \STATE $s' \gets S_c[R_c.\pred(e)] - 1$;
    \STATE $\bwd{s'} \gets n - s' - m$;
\ENDIF
\RETURN $\langle m + 1, b', e', s', \bwd{b'}, \bwd{e'}, \bwd{s'} \rangle$;
\end{algorithmic}
\end{algorithm}

\subsubsection{SA-Interval Enumeration}\label{sct:br_sa_enumeration}
In contrast to the r-index, where we have $s = \SA[e]$, we generally have $s \in \SA[b, e]$.
Hence, in order to enumerate $\SA[b, e]$, we need data structures to compute $\Phi$ and $\Phi^{-1}$ and either
(\textbf{i}) maintain $\SA^{-1}[s]$ or
(\textbf{ii}) find a way to detect when $\Phi^i(s) \notin \SA[b, e]$ for an integer $i$.

(\textbf{i}) can be done using $\Oh((r + \bwd{r}) \log n)$ bits of space, as we will see in \Cref{sct:sa_enumeration}.
(\textbf{ii}) can be done with the array $\PLCP[1..n]$, where $\PLCP[i]$ is the length of the longest matching prefix of $T_i$ and $T_{\Phi(i)}$ (see right in \Cref{fig:br-index}).
As soon as $\PLCP[\Phi^i(s)] < m$ holds, we output $\Phi^i(s)$ and stop performing $\Phi$-queries, because this implies $P$ does not occur at position $\Phi^{i + 1}(s)$ in $T$.
Similarly, we stop performing $\Phi^{-1}$-queries as soon as $\PLCP[\Phi^{-i}(s)] < m$ (and do not output $\Phi^{-i}(s)$), because this implies $P$ does not occur at position $\Phi^{-i}(s)$ in $T$.

We need $\Oh(1)$ time per occurrence if we implement $\Phi$ and $\Phi^{-1}$ with move data structures ($\MPhi$ and $\MPhiinv$).
By storing the irreducible $\PLCP$ values interleaved with the arrays of $\MPhi$ and $\MPhiinv$, $\PLCP[\Phi^i(s)]$ can be computed in $\Oh(1)$ time \cite[Section~3.2]{br_index}.
However, finding the intervals of $\MPhi$ and $\MPhiinv$ that contain $s$ (which we need for performing $\move$ queries) requires two sparse-bit-vector $\pred$-queries ($\Oh(\log\log_\omega (n / r))$ time).
Overall, we get the following running time.

\begin{theorem}[{\cite[Theorem~1]{br_index}}]\label{thm:br_locate}
    Let $P$ be a pattern of length $m$ and let $C$ be a valid context for $P$.
    Then there is a data structure of size $\Oh((r + \bwd{r}) \log n)$ bits, s.t.\ we can compute all occurrences of $P$ in $T$ from $C$ in $\Oh(\log\log_\omega (n / r) + occ)$ time.
\end{theorem}

\section{Algorithmic Optimizations}\label{sect:opt}
We first describe the fundamental data structures of our index and the variables that we maintain during extensions (\Cref{sct:fundamental_definitions}), and then show how to compute a value in $\SA[b, e]$ from them (\Cref{sct:finding_the_initial_sa_value}).
After discussing character extensions (\Cref{sct:all_k_extensions,sct:single_extensions}), we describe our $\SA$-interval enumeration algorithm (\Cref{sct:sa_enumeration}), practical extensions (\Cref{sct:practical_extensions}), construction (\Cref{sct:construction}) and further practical optimizations (\Cref{sct:practical_optimizations}).
Our optimized search scheme APM algorithms are deferred to the appendix (\Cref{sct:apm_algorithms}).

\subsection{Fundamental Definitions}\label{sct:fundamental_definitions}

\begin{definition}
    Let $\MDS$ be a move data structure with $N$ intervals representing a function $f : [1, n] \rightarrow [1, n]$.
    Then, we denote with $\MDS_\p[1..N]$ the array storing its input interval starting positions in ascending order.
    Given $(i, x)$ with $i \in [1, n]$ and $x = \pred(\MDS_\p, i)$, the move query asks for $(i', x')$, where $i' = f(i)$ and $x' = \pred(\MDS_\p, i')$.
    Let $\MLF$ be a balanced move data structure for $\LF$, and let $r' = \Oh(r)$ be the number of its intervals.
    Let $L'[1..r']$ be the string containing the BWT sub-runs w.r.t.\ $\MLF$, i.e.,\ $L'[i] = L[\MLF_\p[i]]$.
    Let $\MPhi$ and $\MPhiinv$ with interval numbers $r'' = \Oh(r)$ and $r''' = \Oh(r)$ be move data structures for $\Phi$ and $\Phi^{-1}$, resp.
    Finally, let $S[1..2r']$ be an array, where $S[2i - 1] = \SA[\MLF_\p[i]]$ and $S[2i] = \SA[\MLF_\p[i + 1] - 1]$.
\end{definition}

Recall that the br-index keeps a value $s \in \SA[b, e]$ at \emph{every} extension, re-establishing it through an expensive $\pred$-query on a sparse bit vector whenever the interval shrinks.
This can be avoided due to the following observation:
After the last shrinking extension (the critical extension, see \Cref{def:extension-path}), the $\SA$- and $\bwd{\SA}$-intervals only relocate by $\LF$, resp.\ $\bwd{\LF}$.
Hence, a sample inside the interval at the critical extension stays inside the interval.
Thus, it suffices to record which sample was caught, its offset in the interval, and how many extensions in that direction have occurred since the critical extension (see \Cref{def:ctx}).
This is lightweight bookkeeping and incurs no cache misses.

\begin{definition}\label{def:extension-path}
    Let $P$ be a pattern of length $m \geq 1$.
    We call any $D = [d_1, \ldots, d_m]$ an extension path for $P$, where $d_i \in \{\mathsf{L}, \mathsf{R}\}$ is the direction of the $i$-th extension.
    For $i \in [0, m]$, let $S_i$ be the substring of $P$ matched after performing the first $i$ extensions of $D$, and let $[b_i, e_i]$ and $[\bwd{b_i}, \bwd{e_i}]$ be the $\SA$- and $\bwd{\SA}$-intervals of $S_i$ and $\bwd{S_i}$, resp.
    We call the last extension shrinking the $\SA$- and $\bwd{\SA}$-intervals the critical extension of this extension path, i.e.,\ its index $X \in [1, m]$ is maximal, s.t.\ $e_X - b_X < e_{X - 1} - b_{X - 1}$.
\end{definition}

Since each shrinking ($i$-th) left-extension (with character $c$) has a BWT sub-run of $c$ starting or ending within its $\SA$-interval $[b_{i - 1}, e_{i - 1}]$, we get the following.

\begin{lemma}\label{lem:exists_sample}
    Let $P$ be a pattern of length $m$ and let $D$ be an extension path for $P$ as defined in \Cref{def:extension-path}.
    Let $j \in [1, m]$ s.t.\ $e_j - b_j < e_{j - 1} - b_{j - 1}$.
    Then $\mathbf{(i)}$ if $d_j = \mathsf{L}$, there is an $x \in [1, 2r']$ s.t.\ $S[x] - 1 \in \SA[b_j, e_j]$, and
    $\mathbf{(ii)}$ if $d_j = \mathsf{R}$, there is an $x \in [1, 2\bwd{r'}]$ s.t.\ $\bwd{S}[x] - 1 \in \bwd{\SA}[\bwd{b_j}, \bwd{e_j}]$.
\end{lemma}

With this, we can now define the variables that we maintain during character extensions.

\begin{definition}\label{def:ctx}
    Let $P$ be a pattern of length $m \geq 1$.
    Let $C = \langle m, d, \sa, \bwd{\sa}, \chi \rangle$, where
    $\sa = \langle b, e, \mu, \nu \rangle$,
    $\bwd{\sa} = \langle \bwd{b}, \bwd{e}, \bwd{\mu}, \bwd{\nu} \rangle$, and
    $\chi = \langle t, x, i, o \rangle$.
    We call $C$ a valid context for $P$ iff
    \begin{enumerate}
        \item[$\mathbf{(i)}$] $[b, e]$ and $[\bwd{b}, \bwd{e}]$ are the $\SA$- and $\bwd{\SA}$-intervals of $P$ and $\bwd{P}$, resp.,
        \item[$\mathbf{(ii)}$] $[\mu, \nu] = [\pred(\MLF_\p, b), \pred(\MLF_\p, e)]$ if $d = \mathsf{L}$, and $[\bwd{\mu}, \bwd{\nu}] = [\pred(\MLFbwd_\p, \bwd{b}), \pred(\MLFbwd_\p, \bwd{e})]$, else, and
        \item[$\mathbf{(iii)}$] there is an extension path $D$ for $P$ as defined in \Cref{def:extension-path}, s.t.\ $d = d_m$, $t = d_X$, $i = |\{j \in [X, m] \mid d_j = d_X\}|$, and
        \begin{enumerate}
            \item if $d_X = \mathsf{L}$, then $S[x] - 1 \in \SA[b_X, e_X]$ and $o = \SA^{-1}[S[x] - 1] - b_X$, and
            \item if $d_X = \mathsf{R}$, then $\bwd{S}[x] - 1 \in \bwd{\SA}[\bwd{b_X}, \bwd{e_X}]$.
        \end{enumerate}
    \end{enumerate}
\end{definition}

Since the first extension of any extension path is shrinking, one can easily see that for any choice of $P$ (that occurs in $T$) and $D$, there is a valid context for $P$.

$\mu$ and $\nu$ are the BWT sub-run indices of $b$ and $e$ and are necessary for performing $\LF$-$\move$-queries with $b$ and $e$, and finding the first and last occurrences of $c$ in $L[b, e]$ during left-extensions.
By (\textbf{ii}), we maintain $\mu$ and $\nu$ if the last performed extension was a left-extension ($d = \mathsf{L}$), and $\bwd{\mu}$ and $\bwd{\nu}$, else.
This has two reasons: we only need $\mu$ and $\nu$ to perform left-extension (and only $\bwd{\mu}$ and $\bwd{\nu}$ to perform right-extension), and we cannot efficiently maintain the other pair (in $\Oh(k)$ time) after a left-, resp.\ right-extension, as this requires $2$ $\pred$-queries on $\MLFbwd_\p$, resp.\ $\MLF_\p$.
Thus, we only perform these $\pred$-queries when the direction switches.
Finally, (\textbf{iii}) ensures that, depending on the direction $t$ of the critical extension, there is either a sample in $\SA[b_X, e_X]$ and $o$ is the offset of its position in $[b_X, e_X]$, if $t = \mathsf{L}$, or a sample in $\bwd{\SA}[\bwd{b_X}, \bwd{e_X}]$, else.
This sample can be used to compute a value in $\SA[b, e]$ in the following way.

\subsection{Computing One Value in the SA-Interval}\label{sct:finding_the_initial_sa_value}
Here, we have two cases, depending on the direction of the critical extension.
If $d_X = \mathsf{L}$, we can compute an $\SA$-value $v$ and its position $p$ from $\chi$ in $\Oh(1)$ time, because $v$ is derived from an $\SA$-sample obtained at the critical extension representing the start of the then matched substring of $P$, so its position only moves under $\LF$ afterwards, i.e.,\ its offset remains unchanged.
Else ($d_X = \mathsf{R}$), $v$ is derived from a $\bwd{\SA}$-sample and represents the end of the then matched substring of $P$, so we have to perform $m - i$ $\LF$ steps to obtain the position in $\SA[b, e]$ of its corresponding $\SA$-value.

\begin{theorem}\label{thm:initial_sa_value}
    Let $P$ be a pattern of length $m$ and let $C$ be a valid context for $P$ as defined in \Cref{def:ctx}.
    Then there is a data structure of size $\Oh((r + \bwd{r}) \log n)$ bits, s.t.\ we can compute a value $v \in \SA[b, e]$ and its position $p = \SA^{-1}[v]$ from $C$ in $\Oh(m)$ time.
\end{theorem}
\begin{proof}
    Let $D$ be an extension path satisfying (\textbf{iii}) in \Cref{def:ctx} w.r.t.\ $P$.
    We consider two cases.

    \textbf{Case 1}: $d_X = \mathsf{L}$.
    Then we have $S[x] - 1 = \SA[b_X + o] \in \SA[b_X, e_X]$ due to (\textbf{iii}(a)) in \Cref{def:ctx}.
    Let $f(Y) = |\{j \in (X, Y] \mid d_j = \mathsf{L}\}|$ count the left-extensions strictly after the critical extension; since $d_X = \mathsf{L}$, we have $f(m) = i - 1$.
    Since each remaining extension does not shrink the $\SA$- and $\bwd{\SA}$-intervals, it follows by induction on $Y \in [X, m]$ that
    \begin{align}
    \begin{split}
        \SA^{-1}[S[x] - 1 - f(Y)] = \LF^{f(Y)}(b_X + o) \\
        = \LF^{f(Y)}(b_X) + o = b_Y + o \in [b_Y, e_Y],\label{eq:induction_l}
    \end{split}
    \end{align}
    where the base case $Y = X$ holds by (\textbf{iii}(a)) in \Cref{def:ctx}, as $f(X) = 0$ and $\SA^{-1}[S[x] - 1] = b_X + o$.
    Setting $Y = m$ in \Cref{eq:induction_l} then yields $\SA^{-1}[S[x] - i] = b + o \in [b, e]$.
    Hence, we can set $v = S[x] - i$ and $p = b + o$ in $\Oh(1)$ time and are done.

    \textbf{Case 2}: $d_X = \mathsf{R}$.
    Then we have $\bwd{S}[x] - 1 \in \bwd{\SA}[\bwd{b_X}, \bwd{e_X}]$ due to (\textbf{iii}(b)) in \Cref{def:ctx}.
    Let $g(Y) = |\{j \in [X, Y] \mid d_j = \mathsf{R}\}|$ count the right-extensions among the first $Y$ extensions; since $d_X = \mathsf{R}$, we have $g(m) = i$.
    Then, for $Y \in [X, m]$, it follows by induction on $Y$ that
    \begin{align}
    \begin{split}
        \SA^{-1}[v' - (Y - g(Y))] = \\
        \LF^{Y - g(Y)}(p') \in [b_Y, e_Y],\label{eq:induction_r}
    \end{split}
    \end{align}
    where $v' = n - \bwd{S}[x] + 1$ and $p' = \SA^{-1}[v']$.
    Setting $Y = m$ yields $\SA^{-1}[v' - (m - i)] = \LF^{m - i}(p') \in [b, e]$.
    Hence we can set $v = v' - (m - i) = n - \bwd{S}[x] + 1 - (m - i)$ and set $p = \LF^{m - i}(p')$.

    It remains to show how we can compute $v$ and $p$ in this case.
    Although obtaining $v$ is easy and takes $\Oh(1)$ time, computing $p$ is more difficult:
    First, we have to find $p'$, which takes $\Oh(1)$ time if we augment our index with the array $\bwd{S}_\pos[1..2\bwd{r'}]$ containing the positions in $\SA$ of the $\bwd{\SA}$-samples, i.e.,\ $\bwd{S}_\pos[i] = \SA^{-1}[n - \bwd{S}[i] + 1]$, and set $p' = \bwd{S}_\pos[x]$.
    Now, it remains to compute $p = \LF^{m - i}(p')$ from $p'$.
    To this end, we augment our index with the array $\bwd{S}_{\p\idx}[1..2\bwd{r'}]$, where $\bwd{S}_{\p\idx}[i] = \pred(\MLF_\p, \bwd{S}_\pos[i])$.
    Then, we can initialize $(p, q) = (p', \bwd{S}_{\p\idx}[x])$ and $m - i$ times set $(p, q) = \MLF.\move(p, q)$ to get $p$.
    This takes $\Oh(m)$ time.
\end{proof}

\subsection{All-$k$-Character Extensions}\label{sct:all_k_extensions}
We now describe how to perform left-extension with Move-rb: we first give an output-sensitive algorithm for all-$k$-character extensions, and then show how to compute a single-character extension with an augmented wavelet tree.

\begin{algorithm}[!h]
\caption{$\mathsf{Left\text{-}Extension}(\langle m, d, \sa, \bwd{\sa}, \chi \rangle)$}\label{alg:prepend_all}
\begin{algorithmic}[1]
\LIF{$d = \mathsf{R}$}{$[\mu, \nu] \mkern-1mu\gets\mkern-1mu [\MLF_\p.\pred(b), \MLF_\p.\pred(e)]$}
\STATE $((c_1,B_1,E_1), \ldots, (c_k,B_k,E_k)) \gets L'.\mathsf{first\text{-}last}(\mu,\nu)$;
\STATE $\beta \gets \bwd{b}$;
\FOR{$j \gets 1$ \textbf{to} $k$}
    \STATE $\overline{b} \gets \max(b, \MLF_\p[B_j])$;
    \STATE $\overline{e} \gets \min(e, \MLF_\p[E_j + 1] - 1)$;
    \STATE $(b', \mu') \gets \MLF.\move(\overline{b}, B_j)$;
    \STATE $(e', \nu') \gets \MLF.\move(\overline{e}, E_j)$;
    \STATE $[\bwd{b'}, \bwd{e'}] \gets [\beta, \beta + e' - b']$;
    \STATE $\beta \gets \bwd{e'} + 1$;
    \LIF{$B_j \neq \mu$}{$\chi' \gets \langle \mathsf{L}, 2 B_j - 1, 1, 0 \rangle$}
    \LELSEIF{$E_j \neq \nu$}{$\chi' \gets \langle \mathsf{L}, 2 E_j, 1, e' - b' \rangle$}
    \ELSEIFONLY{$e' - b' < e - b$}
        \STATE $\chi' \gets \langle \mathsf{L}, 2 B_j, 1, \MLF_\p[B_j + 1] - 1 - \overline{b} \rangle$;
    \ENDELSEIFONLY
    \LELSE{$\chi' \gets \langle t, x, (t = \mathsf{L} \ ? \ i + 1 : i), o \rangle$}
    \STATE $\mathbf{report}(c_j, \langle m + 1, \mathsf{L}, \sa', \bwd{\sa'}, \chi' \rangle)$;
\ENDFOR
\end{algorithmic}
\end{algorithm}

\begin{lemma}\label{lem:first-last-occurrences}
    We can augment $L'$ with $\Oh(r \log n)$ bits, s.t.\ for any interval $[\mu, \nu]$ in $L'$, we can report $((c_1, B_1, E_1), \ldots, (c_k, B_k, E_k))$ in $\Oh(k)$ time, where $c_1 < \cdots < c_k$ are the distinct characters in $L'[\mu, \nu]$, and $B_i = \suc_{c_i}(L', \mu)$ and $E_i = \pred_{c_i}(L', \nu)$.
\end{lemma}
\begin{proof}
    We store the prev-occ array $A^-[i] = \pred_{L'[i]}(L',i - 1)$ and the next-occ array $A^+[i] = \suc_{L'[i]}(L',i + 1)$.
    A position $p \in [\mu, \nu]$ is the first occurrence of $L'[p]$ in $L'[\mu, \nu]$ iff $A^-[p] < \mu$, and $q \in [\mu, \nu]$ is its last occurrence iff $A^+[q] > \nu$.
    We therefore build the range-minimum-query based range-color reporting data structure of \cite[Section~3, Step~3]{document_listing} over $A^-$, which reports all positions $p \in [\mu, \nu]$ with $A^-[p] < \mu$ in $\Oh(1)$ time each, and a range-maximum-query based one over $A^+$, which analogously reports all positions $q \in [\mu, \nu]$ with $A^+[q] > \nu$.
    The two queries thus yield the pairs $(L'[p], p)$ and $(L'[q], q)$ of the first, resp.\ last, occurrence of every distinct character of $L'[\mu, \nu]$ in $\Oh(k)$ time.

    To combine the two reported sets by their characters, a hash table would suffice, but would yield only expected $\Oh(k)$ query time.
    We therefore augment our index with an array $V[1..\sigma]$ and set $V[L'[p]] \gets (p, \bot)$ for every pair $(L'[p], p)$ reported by the first range-color query, and then $V[L'[q]] \gets (p, q)$ for every pair $(L'[q], q)$ reported by the second one, where $(p, \bot) = V[L'[q]]$.

    It remains to report the triples in alphabet order.
    We augment $L'$ with the range top-$k$ color data structure of \cite[Theorem~2]{top_k_colors}, assigning position $i$ the color $L'[i]$ and color $c$ the priority $\sigma-c$.
    Given a range and a number $k$, it reports the $k$ colors of highest priority in that range in decreasing priority and $\Oh(k)$ time.
    So we issue a range top-$k$ color query on $[\mu, \nu]$, which reports the characters in lexicographic order.
    For every character $c_i$, we output $(c_i, B_i, E_i)$, where $(B_i, E_i) = V[c_i]$.

    The two range-color reporting data structures take $\Oh(r' \log n) = \Oh(r \log n)$ bits, and the range top-$k$ color data structure takes $\Oh(r' \log \sigma) = \Oh(r \log \sigma)$ bits.
    $V$ takes $\Oh(\sigma \log r') = \Oh(r \log r')$ bits.
\end{proof}

Two parallel queries would overwrite each other's entries of $V$, so every thread needs its own copy ($\Oh(\sigma \log r')$ bits per thread); where this is undesirable, a query-local hash table replaces $V$ and needs $\Oh(k)$ expected time and optimal $\Oh(k)$ words per query.

\begin{figure*}
    \centering
    \includegraphics[width=0.93\linewidth]{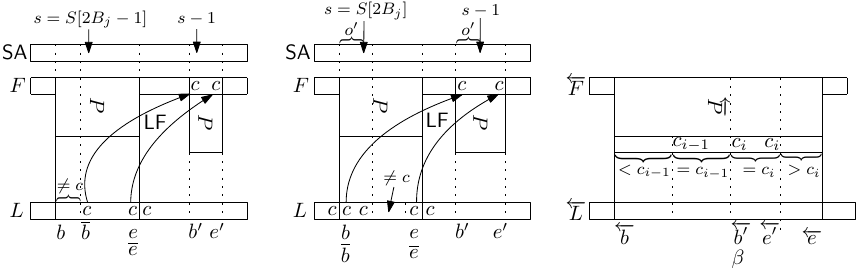}
    \caption{Case 1 (left), Case 3 (middle), and the invariant $\beta = \bwd{b'}$ (right) in the proof of \Cref{thm:left-extension}.}
    \label{fig:move-rb-extension}
\end{figure*}

\begin{theorem}\label{thm:left-extension}
    Let $P$ be a pattern of length $m$ and let $C$ be a valid context for $P$ as defined in \Cref{def:ctx}.
    Let $\bwd{\mathcal{C}}$ (resp.\ $\fwd{\mathcal{C}}$) be the set containing one valid context for every $c \in \Sigma$ such that $cP$ (resp.\ $Pc$) occurs in $T$, and let $\bwd{k} = |\bwd{\mathcal{C}}|$ and $\fwd{k} = |\fwd{\mathcal{C}}|$.
    Then there is a data structure of size $\Oh((r + \bwd{r}) \log n)$ bits that computes $\bwd{\mathcal{C}}$ and $\fwd{\mathcal{C}}$ from $C$.
    If $d = \mathsf{L}$, the running times are $\Oh(\bwd{k})$ and $\Oh(\fwd{k} + \log\log_\omega (n / \bwd{r}))$.
    If $d = \mathsf{R}$, the running times are $\Oh(\bwd{k} + \log\log_\omega (n / r))$ and $\Oh(\fwd{k})$.
\end{theorem}
\begin{proof}
    We will only show how to obtain $\bwd{\mathcal{C}}$, because computing $\fwd{\mathcal{C}}$ is symmetric.

    First, we ensure that $\mu$ and $\nu$ are the input interval indices of $b$ and $e$ in $\MLF$, which is only necessary if $d = \mathsf{R}$, because else, (\textbf{ii}) in \Cref{def:ctx} already ensures this.
    This can be done in $\Oh(\log\log_\omega (n / r))$ time if we additionally store $\MLF_\p$ and $\MLFbwd_\p$ Elias--Fano (EF) encoded~\cite{succinct_indexable_dictionary}, augmented for predecessor queries.

    Recall that computing the $\SA$-intervals of each context in $\bwd{\mathcal{C}}$ requires the first and last occurrences of each $c \in \Sigma$ in $L'[\mu, \nu]$, which the $\mathsf{first\text{-}last}$ query in \Cref{alg:prepend_all} reports in $\Oh(k)$ time by \Cref{lem:first-last-occurrences}.
    These are exactly the characters $c_j$ for which $c_jP$ occurs in $T$, and their contexts are $C_1, \ldots, C_k$.
    In the $j$-th iteration of the for-loop in lines 4--17, we compute $C_j = \langle m', d', \sa', \bwd{\sa'}, \chi' \rangle$, defined as in \Cref{def:ctx} with primed variables.

    The main insight of our algorithm is that the $\bwd{\SA}$-intervals of the contexts in $\bwd{\mathcal{C}}$ exactly cover $[\bwd{b}, \bwd{e}]$ and appear sorted in $[\bwd{b}, \bwd{e}]$ (see right in \Cref{fig:move-rb-extension}).
    Thus, if we maintain the invariant $\beta = \bwd{b'}$, then we can compute each $C_j$ in $\Oh(1)$ time:
    First, we obtain $b', e', \mu'$ and $\nu'$ as with Move-r.
    Then, we can set $[\bwd{b'}, \bwd{e'}] = [\beta, \beta + e' - b']$ due to the invariant, and set $\beta = \bwd{e'} + 1$ to maintain it.

    Finally, it remains to compute $\chi'$.
    Here, we consider four cases:

    \textbf{Case 1}: $B_j \neq \mu$ (see left in \Cref{fig:move-rb-extension}).
    In this case, the performed extension is the new critical extension (because $e' - b' < e - b$).
    Hence, we set $t' = \mathsf{L}$ and $i' = 1$.
    Since we can use the $\SA$-sample $S[2 B_j - 1]$ at the start of the $B_j$-th sub-run in $L$, we set $x' = 2 B_j - 1$ and $o' = 0$.

    \textbf{Case 2}: $B_j = \mu \land E_j \neq \nu$.
    This is symmetric to Case 1, but we use the $\SA$-sample $S[2 E_j]$ at the end of the $E_j$-th sub-run.

    \textbf{Case 3}: $B_j = \mu \land E_j = \nu \land e' - b' < e - b$ (see middle in \Cref{fig:move-rb-extension}).
    Similarly, the performed extension is the new critical extension.
    In this case, we can use the $\SA$-sample $S[2 B_j]$ at the end of the $B_j$-th sub-run, i.e.,\ set $x' = 2 B_j$ and $o' = \MLF_\p[B_j + 1] - 1 - \overline{b}$.

    \textbf{Case 4}: $e' - b' = e - b$.
    This case is simple, as the performed extension is not the new critical extension, hence $t' = t$, $x' = x$ and $o' = o$.
    We only have to increment $i$ if $t = \mathsf{L}$, i.e.,\ set $i' = i + 1$ if $t = \mathsf{L}$ and $i' = i$, else.
    
    Finally, we report the tuple $(c_j, C_j)$ to indicate that $c_jP$ with context $C_j$ occurs in $T$.
\end{proof}

\subsection{Single-Character Extensions}\label{sct:single_extensions}
To answer a single-character extension query in $\Oh(k)$ time without a direction switch, we can use the same algorithm as in \Cref{thm:left-extension}, but only report the context for the character $c$; if we increase space slightly, an augmented wavelet tree reduces this to $\Oh(\log \sigma)$ time.

\begin{theorem}\label{thm:single-character-extension}
    Let $P$ be a pattern of length $m$, let $C$ be a valid context for $P$ as defined in \Cref{def:ctx}, and let $c \in \Sigma$.
    Move-rb can be augmented to take $\Oh((r + \bwd{r})(\log n + \log(n / r_{\min}) \log \sigma))$ bits, s.t.\ valid contexts of $cP$ and $Pc$ can then be computed from $C$.
    If $d = \mathsf{L}$, the running times are $\Oh(\log \sigma)$ and $\Oh(\log \sigma + \log\log_\omega (n / \bwd{r}))$.
    If $d = \mathsf{R}$, the running times are $\Oh(\log \sigma + \log\log_\omega (n / r))$ and $\Oh(\log \sigma)$.
\end{theorem}
\begin{proof}
    We only consider left-extension, because right-extension is symmetric.
    If $d = \mathsf{R}$, we first obtain $[\mu, \nu]$ as in \Cref{alg:prepend_all}, which takes $\Oh(\log\log_\omega (n / r))$ time.
    Let $\lambda[i]=\MLF_\p[i + 1]-\MLF_\p[i]$ be the length of the $i$-th sub-run, and let $W(i,j,c) = \sum_{x=i}^{j} \lambda[x]\mathbf{1}_{L'[x] < c}$ be the total length of the sub-runs in $L'[i,j]$ whose character is smaller than $c$.
    Assume for now that we know $W(\mu,\nu,c)$ and the first and last occurrences $B_c$ and $E_c$ of $c$ in $L'[\mu, \nu]$.
    The number of characters smaller than $c$ in $L[b, e]$ is then
    \begin{equation}
        \begin{aligned}
            \delta_c ={}& W(\mu,\nu,c)
            - (b - \MLF_\p[\mu])\mathbf{1}_{L'[\mu] < c} \\
            &- (\MLF_\p[\nu + 1] - 1 - e)\mathbf{1}_{L'[\nu] < c},
        \end{aligned}
        \label{eq:weighted_wavelet_shift}
    \end{equation}
    where the last two terms remove the parts of the boundary sub-runs outside $[b, e]$.
    As in \Cref{alg:prepend_all}, $B_c$ and $E_c$ give $b',e',\mu'$ and $\nu'$ with two move queries in $\Oh(1)$ time, and we set $\bwd{b'} = \bwd{b} + \delta_c$ and $\bwd{e'} = \bwd{b'} + e' - b'$.
    The remaining context variables follow in $\Oh(1)$ time by the same cases as in the main algorithm, so it remains to compute $W(\mu,\nu,c)$, $B_c$ and $E_c$.

    To this end, we augment $L'$ with a balanced wavelet tree \cite{wavelet_tree}.
    For a node $v$, let $L'_v$ and $D_v$ be its standard sequence and routing bit vector, resp.
    If $L'_v[j]$ originates from $L'[i]$, we set $\lambda_v[j]=\lambda[i]$.
    For every level $\ell$, we concatenate the sequences $D_v$ and $\lambda_v$ of its nodes from left to right into $D_\ell$ and $\lambda_\ell$, resp.
    We store $D_\ell$, but instead of $\lambda_\ell$ we store the monotone sequence $Z_\ell[j]=\sum_{i=1}^j\lambda_\ell[i]\mathbf{1}_{D_\ell[i]=0}$ using EF encoding \cite{succinct_indexable_dictionary}.
    If $D_v$ starts at position $a_v$ in $D_\ell$, then $Z_v[j]=Z_\ell[a_v+j-1]-Z_\ell[a_v-1]$, where $Z_\ell[0]=0$, is the total length among the first $j$ entries at $v$ passed to the left child.
    The bit vector $D_\ell$ supports rank and select in $\Oh(1)$ time, and the EF encoding provides access to every $Z_\ell[j]$ in $\Oh(1)$ time.
    Finally, let $h=\ceil{\log_2\sigma}$ be the height of the tree, let $c[\ell]$ be the routing bit of $c$ at level $\ell$, i.e.,\ the $\ell$-th bit of $c$ from the most significant one, and let $v_0$ and $v_1$ denote the left and right child of a node $v$.

\begin{algorithm}[!h]
\caption{$\mathsf{WT\text{-}Query}(\mu, \nu, c)$}\label{alg:wt_query}
\begin{algorithmic}[1]
\STATE $v \gets \mathrm{root}$; $[i, j] \gets [\mu, \nu]$; $\Lambda \gets 0$;
\FOR{$\ell \gets 1$ \textbf{to} $h$}
    \LIF{$c[\ell] = 1$}{$\Lambda \gets \Lambda + Z_v[j] - Z_v[i - 1]$}
    \STATE $[i, j] \gets [\rank_{c[\ell]}(D_v, i - 1) + 1, \rank_{c[\ell]}(D_v, j)]$;
    \STATE $v \gets v_{c[\ell]}$;
\ENDFOR
\LIF{$i > j$}{\textbf{report} $\bot$ \textbf{and stop}}
\FOR{$\ell \gets h$ \textbf{downto} $1$}
    \STATE $v \gets \mathrm{parent}(v)$;
    \STATE $[i, j] \gets [\select_{c[\ell]}(D_v, i), \select_{c[\ell]}(D_v, j)]$;
\ENDFOR
\STATE $\mathbf{report}(\Lambda, i, j)$;
\end{algorithmic}
\end{algorithm}

    We obtain $W(\mu,\nu,c)$, $B_c$ and $E_c$ with \Cref{alg:wt_query}.
    We descend from the root towards the leaf of $c$, mapping $[i,j]$ to the corresponding child interval with rank queries on the segment $D_v$ of $D_\ell$ (line 4).
    Whenever we continue to the right child, all characters passed to the left child are smaller than $c$, so we add $Z_v[j]-Z_v[i-1]$ to the accumulator $\Lambda$ (line 3), which therefore holds $W(\mu,\nu,c)$ once we reach the leaf.
    The interval we reach at the leaf of $c$ is empty iff $c$ does not occur in $L'[\mu, \nu]$ (line 7).
    Otherwise, we map its endpoints back to the root with select queries (lines 8--11), where they are $B_c$ and $E_c$.
    For readability, \Cref{alg:wt_query} treats the nodes as explicit objects, although we do not store them: we represent a node $v$ at level $\ell$ by the boundaries of its segment in $D_\ell$, from which two rank queries on $D_\ell$ give those of both children of $v$ in $\Oh(1)$ time, and we remember the $h$ pairs of boundaries of the descent for the ascent.
    We spend $\Oh(1)$ time per level in both loops, so the whole extension takes $\Oh(h) = \Oh(\log \sigma)$ time.

    At every wavelet-tree level, $D_\ell$ and the EF encoding of $Z_\ell$ for $L'$ take $\Oh(r)$ bits plus $\Oh(r \log(n/r))$ bits, and the analogous data structures for $\bwd{L'}$ take $\Oh(\bwd{r})$ bits plus $\Oh(\bwd{r} \log(n/\bwd{r}))$ bits.
    Since $r_{\min} \leq r, \bwd{r}$, summing over the $\Oh(\log \sigma)$ levels yields $\Oh((r + \bwd{r}) \log(n / r_{\min}) \log \sigma)$ bits for the EF encodings, plus $\Oh((r + \bwd{r}) \log \sigma)$ bits for the bit vectors.
    The latter is $\Oh((r + \bwd{r}) \log n)$ bits because $\log \sigma \leq \log n$, which proves the bound.
    The augmentation therefore leaves the index size of \Cref{thm:left-extension} unchanged up to a constant factor whenever $\log(n / r_{\min}) \log \sigma = \Oh(\log n)$, i.e.,\ unless both the alphabet and the compression rate $n / r_{\min}$ are large.
\end{proof}

\subsection{SA-Interval Enumeration}\label{sct:sa_enumeration}
Now it remains to discuss our $\SA$-interval enumeration algorithm (see \Cref{alg:locate}).

\begin{theorem}\label{thm:locate}
    Let $P$ be a pattern of length $m$ and let $C$ be a valid context for $P$ as defined in \Cref{def:ctx}.
    There is a data structure of size $\Oh((r + \bwd{r}) \log n)$ bits s.t.\ we can compute all $occ$ occurrences of $P$ in $T$ from $C$ in $\Oh(m + occ)$ time.
\end{theorem}
\begin{proof}
    We start by using \Cref{thm:initial_sa_value} to compute a value $v \in \SA[b, e]$ and its position $p = \SA^{-1}[v]$ in $\Oh(m)$ time, and then compute $\SA[b, p - 1]$ using $\Phi$-$\move$-queries and $\SA[p + 1, e]$ using $\Phi^{-1}$-$\move$-queries.
    To this end, we need $w = \pred(\MPhi_\p, v)$ and $w' = \pred(\MPhiinv_\p, v)$, for which we augment our index with four arrays:
    $S_\idx[1..2r']$, $S'_\idx[1..2r']$, $\bwd{S}_\idx[1..2\bwd{r'}]$ and $\bwd{S'}_\idx[1..2\bwd{r'}]$, where $S_\idx[i] = \pred(\MPhi_\p, S[i])$, $S'_\idx[i] = \pred(\MPhiinv_\p, S[i])$, $\bwd{S}_\idx[i] = \pred(\MPhi_\p, n - \bwd{S}[i] + 1)$ and $\bwd{S'}_\idx[i] = \pred(\MPhiinv_\p, n - \bwd{S}[i] + 1)$.

    Recall from the proof of \Cref{thm:initial_sa_value} that we obtain $v$ from a stored sample $v'$ by subtracting at most $m$ from $v'$, i.e.,\ $v' - v \leq m$.
    Thus, we can compute $w$ by initializing $w$ as in line 3 or 10 (depending on $t$), and then decrementing $w$ until $\MPhi_\p[w] \leq v$.
    This takes $\Oh(\log m)$ time with an exponential search.
    Obtaining $w'$ is analogous.
    Finally, we recover and report $\SA[b, e]$ in $\Oh(occ)$ time (lines 14--21).
\end{proof}

\begin{theorem}\label{thm:locate_plcp}
    Let $P$ be a pattern of length $m$ and let $C$ be a valid context for $P$ as defined in \Cref{def:ctx}.
    There is a data structure of size $\Oh((r + \bwd{r}) \log n)$ bits s.t.\ we can compute all $occ$ occurrences of $P$ in $T$ from $C$ in $\Oh(\log\log_\omega (n / r) + occ)$ time.
\end{theorem}
\begin{proof}
    We compute $v$ as in \Cref{thm:initial_sa_value} in $\Oh(1)$ time and proceed as the br-index with interleaved irreducible $\PLCP$ values (see \Cref{thm:br_locate} and \Cref{sct:br_sa_enumeration}).
\end{proof}

In practice, this is not faster, as the $\Oh(m)$ term in \Cref{thm:locate} is absorbed by the $m$ extensions performed before; we therefore use \Cref{alg:locate}.

\subsection{Practical Character Extensions}\label{sct:practical_extensions}
The data structures in \Cref{lem:first-last-occurrences,thm:single-character-extension} achieve strong worst-case bounds, but in practice, we use a simpler approach that is faster and more memory-efficient to obtain the first and last occurrences $B_c$ and $E_c$ of every character $c$ in $L'[\mu,\nu]$, which \Cref{alg:prepend_all} needs.
We precompute previous and next occurrences at regularly spaced block boundaries in $L'$ (see \Cref{sct:practical_optimizations} for the block size) and inspect only the two blocks containing $\mu$ and $\nu$.
More precisely, we store the arrays $L'_\pred[0..\ceil{r' / \sigma}][1..\sigma]$ and $L'_\suc[0..\ceil{r' / \sigma}][1..\sigma]$, where $L'_\suc[b][c] = \suc_c(L', b \sigma)$ and $L'_\pred[b][c] = \pred_c(L', b \sigma)$.
Analogously, we store $\bwd{L'}_\pred$ and $\bwd{L'}_\suc$.

To compute $(c_1,B_1,E_1), \ldots, (c_k,B_k,E_k)$, we initialize $B \gets L'_\suc[\ceil{\mu / \sigma}]$ and $E \gets L'_\pred[\floor{\nu / \sigma}]$ ($\Oh(\sigma)$ time).
Then, we scan with $i$ over $L'[\mu, \ceil{\mu / \sigma} \sigma]$ from right to left and set $B[L'[i]] \gets i$ at each position.
Similarly, we scan with $i$ over $L'[\floor{\nu / \sigma} \sigma + 1, \nu]$ from left to right and set $E[L'[i]] \gets i$ at each position.
Finally, we iterate with $c$ from $1$ to $\sigma$ and report $(c,B[c],E[c])$ whenever $B[c] \leq E[c]$, preserving lexicographic order.

\begin{algorithm}[!h]
\caption{$\mathsf{Locate}(\langle m, d, \sa, \bwd{\sa}, \chi \rangle)$}\label{alg:locate}
\begin{algorithmic}[1]
\IF{$t = \mathsf{L}$}
    \STATE $v \gets S[x] - i;\ p \gets b + o$;
    \STATE $w \gets S_\idx[x];\ w' \gets S'_\idx[x]$;
\ELSE
    \STATE $v \gets n - \bwd{S}[x] + 1 - (m - i)$;
    \STATE $p \gets \bwd{S}_\pos[x];\ q \gets \bwd{S}_{\p\idx}[x]$;
    \FOR{$k$ \textbf{from} $1$ \textbf{to} $m - i$}
        \STATE $(p, q) \gets \MLF.\move(p, q)$;
    \ENDFOR
    \STATE $w \gets \bwd{S}_\idx[x];\ w' \gets \bwd{S'}_\idx[x]$;
\ENDIF
\STATE $w \gets \max \{j \in [1, r''] \mid \MPhi_\p[j] \leq v\}$;
\STATE $w' \gets \max \{j \in [1, r'''] \mid \MPhiinv_\p[j] \leq v\}$;
\STATE $s \gets v;\ \mathbf{report}(s)$;
\FOR{$j$ \textbf{from} $p - 1$ \textbf{downto} $b$}
    \STATE $(s, w) \gets \MPhi.\move(s, w);\ \mathbf{report}(s)$;
\ENDFOR
\STATE $s \gets v$;
\FOR{$j$ \textbf{from} $p + 1$ \textbf{to} $e$}
    \STATE $(s, w') \gets \MPhiinv.\move(s, w');\ \mathbf{report}(s)$;
\ENDFOR
\end{algorithmic}
\end{algorithm}

A query reads at most $2\sigma$ table entries and at most $2\sigma$ characters from $L'$, and therefore takes $\Oh(\sigma)$ time.
It incurs only two additional random memory accesses (at $L'_\pred$ and $L'_\suc$), as $L'[\mu]$ and $L'[\nu]$ are accessed afterwards in the main algorithm.
$L'_\pred$ and $L'_\suc$ require $\Oh(r' \log r')$ bits and can be constructed in $\Oh(r')$ time.

\subsection{Construction}\label{sct:construction}
In this section, we describe how to construct the implemented data structures that are not already present in Move-r \cite{move_r}.
To construct $\bwd{S}_\pos$, we first build a hash map $H$ that maps each value $v$ to the set $\{x \mid n - \bwd{S}[x] + 1 = v\}$, for which we use the variant emhash6 of the hash set and hash map library \texttt{emhash}.
Since $\bwd{S}[2i - 1] = \bwd{S}[2i]$ holds iff the $i$-th input interval of $\MLFbwd$ has length $1$, and $\bwd{\SA}$ is a permutation, each key is mapped to at most $2$ indices.
Then, we iterate over $\SA$ and check at position $i$ whether there is a mapping $\SA[i] \mapsto X$ in $H$ and write $\bwd{S}_\pos[x] = i$ for each $x \in X$, if it exists.
This takes $\Oh(n)$ expected time.
Finally, $S_\idx$, $S'_\idx$, $\bwd{S}_\idx$, $\bwd{S'}_\idx$ and $\bwd{S}_{\p\idx}$ can be constructed in $\Oh((r + \bwd{r}) \log r)$ time using binary searches over the input interval starting positions of the respective move data structures.

\subsection{Practical Optimizations}\label{sct:practical_optimizations}
In practice, we do not use the arrays $S_\idx$, $S'_\idx$, $\bwd{S}_\idx$, $\bwd{S'}_\idx$ and $\bwd{S}_{\p\idx}$.
Instead, we use the EF encoded copies of $\MLF_\p$ and $\MLFbwd_\p$ and additionally store $\MPhi_\p$ and $\MPhiinv_\p$ EF encoded.
Furthermore, these four EF encodings only contain every $\kappa$-th entry of $\MLF_\p$, $\MLFbwd_\p$, $\MPhi_\p$ and $\MPhiinv_\p$, resp., where $\kappa \geq 1$ is an integer parameter.
Preliminary testing revealed that setting $\kappa = 4$ results in all four EF encodings accounting for a negligibly small fraction of the overall index size, while query times remain unchanged.
We store every array in the index bit-packed.

The description in \Cref{sct:practical_extensions} uses a fixed block size $\sigma$; in practice, we make it tunable with a parameter $\gamma$, so that a block has size $\gamma \sigma$, a query takes $\Oh(\gamma \sigma)$ time, and the tables require $\Oh((r' / \gamma) \log r')$ bits.
Preliminary experiments showed that $\gamma = 4$ offers the best trade-off.
In contrast to $\kappa$, the choice of $\gamma$ changes the speed by $0$--$30\%$, depending on the text.

We also optimize performing an extension with a single character $c$:
Here, we only need the triples of the characters $c_j \leq c$, so it suffices to initialize and report the first $c$ entries of $B$ and $E$ in \Cref{sct:practical_extensions} and to perform only the corresponding iterations of the for-loop in lines 4--17 of \Cref{alg:prepend_all}.
We exploit this by rearranging the alphabet of $T$ before constructing Move-rb, s.t.\ the more frequent a character is, the smaller the character it is remapped to.
Assuming a queried pattern has a similar character distribution as the input, this minimizes the expected running time.
Furthermore, the iterations of the for-loop with $c_j < c$ can be simplified for this case, s.t.\ only $\beta$ is incremented.

\texttt{b-move}'s index being 1.3$\times$ larger on chr19 (see \Cref{tab:texts}) comes from the balancing parameter $a$, which bounds the number of intervals of a move data structure by $(a/(a-1))r$ \cite[Theorem~6]{move_r}.
\texttt{b-move} hardcodes $a = 2$ for $\Phi$ and $\Phi^{-1}$, allowing $2r$ intervals each, and omits balancing for $\LF$ and $\bwd{\LF}$, while we use $a = 8$ (as recommended in \cite{move_r}), allowing $(8/7)r$.
On chr19, $\Phi$ and $\Phi^{-1}$ are maximally unbalanced, so $a = 2$ doubles the intervals; on sars2, $a = 2$ adds only $6\%$ \cite[Table~1]{move_r}.

\section{Experimental Evaluation}
We implemented the algorithms from \Cref{sect:opt} in \texttt{C++20}, using SDSL's sd-array \cite{sd_array,sdsl} for sparse bit vectors and \texttt{malloc\_count} to measure peak memory.
We also implemented a variant of \texttt{move-rb}, called \texttt{move-rb-rlzsa}, that uses the RLZ-encoded suffix array from \cite{rlzsa} instead of $\MPhi$ and $\MPhiinv$.
We performed all measurements on a system with two AMD EPYC 7452 CPUs (32 cores/64 threads, 2.35--3.35 GHz, 2/16/128 MB L1/L2/L3 cache) and 1 TB of 3200 MT/s DDR4 RAM using the GCC 13.3.0 compiler and the compile flags \texttt{"-march=native -DNDEBUG -Ofast"} on Ubuntu 24.04.3.
All indexes are constructed and queried with a single thread.
\Cref{tab:texts} shows three 50 GB texts: German Wikipedia entries (dewiki), chromosome 19 haplotypes (chr19), and SARS-CoV-2 genomes (sars2); we restrict chr19 and sars2 to A, C, G and T as required (hardcoded) by \texttt{b-move} and \texttt{columba}.
Our implementation and all data sets used are available on Zenodo (\url{https://zenodo.org/records/22879117}) and GitHub (\url{https://github.com/LukasNalbach/Move-r}).

\subsection{Competitors}
We compare \texttt{move-rb} against the state-of-the-art compressed bi-directional r-indexes \texttt{br-index} \cite{br_index} and \texttt{b-move} \cite{b_move}, and the bi-directional FM-index \texttt{columba} \cite{columba_dynamic_partitioning, columba_in_text_verification}.
\texttt{br-index} and \texttt{columba} use divsufsort (for $n \geq 2^{31}$) to construct $\SA$ and the BWT, while \texttt{b-move} uses Big-BWT \cite{pfp}, like \texttt{move-rb}.
As \texttt{b-move} is the run-length-compressed version of \texttt{columba}, it uses the same APM algorithms.
\Cref{tab:texts} also shows the index sizes: \texttt{move-rb} is 2.0--2.2$\times$ larger than \texttt{br-index}, has the same size as \texttt{b-move} on sars2, is 24\% smaller on chr19 (see \Cref{sct:practical_optimizations}), and is 36--42$\times$ smaller than the uncompressed \texttt{columba}; \texttt{move-rb-rlzsa} is 1.2--2.4$\times$ larger than \texttt{move-rb}.
Since \texttt{b-move} and \texttt{columba} only support DNA alphabets, they cannot be tested on dewiki.

\begin{figure*}[p]
\centering
\vspace*{\fill}

\begin{tabular}{l|r|r|r|r|r|r|r|r|r}
\hline
text & size & $\sigma$ & $n/r$ & $n/\bwd{r}$ & br-index & b-move & columba & Move-rb & Move-rb-rlzsa \\
\hline
 sars2 & \num{50} GB &   \num{4} & \num{1000.93} & \num{1000.39} & \num{1801} MB & \num{3875} MB & \num{139} GB &  \num{3879} MB &  \num{9268} MB \\
 chr19 & \num{50} GB &   \num{4} & \num{1119.46} & \num{1119.13} & \num{1624} MB & \num{4342} MB & \num{139} GB &  \num{3315} MB &  \num{4975} MB \\
dewiki & \num{50} GB & \num{209} &  \num{348.85} &  \num{349.49} & \num{5273} MB &            -- &           -- & \num{11120} MB & \num{13403} MB \\
\hline
\end{tabular}

{\captionsetup{type=table}\caption{Sizes and compression rates ($n / r$ and $n / \bwd{r}$) of the tested texts, and sizes of the tested indexes. An en dash denotes an unavailable measurement.}\label{tab:texts}}

\vspace{0.2cm}

\begin{subfigure}[h]{.31\textwidth}
\hspace*{-0.4cm}
\begin{tikzpicture}[marks]
\begin{axis}[
    xmode=log,
    constructionstyle,
    title={sars2},
    ylabel={construction thr.\ [MB/s]},
    xlabel={peak memory usage [bytes/$n$]},
    y label style={at={(0.05,0.5)}},
]

\addplot[bmove] coordinates { (2.96715,0.483009) };
\addlegendentry{algo=build\_bmove};
\addplot[brindex] coordinates { (16.5343,0.164353) };
\addlegendentry{algo=build\_br\_index};
\addplot[columba] coordinates { (15.1399,0.598632) };
\addlegendentry{algo=build\_columba};
\addplot[moverb] coordinates { (0.121992,2.18966) };
\addlegendentry{algo=build\_move\_rb\_move};
\addplot[moverbrlzsa] coordinates { (0.347659,1.85704) };
\addlegendentry{algo=build\_move\_rb\_rlzsa};

\legend{};
\end{axis}
\end{tikzpicture}
\end{subfigure}
\begin{subfigure}[h]{.31\textwidth}
\begin{tikzpicture}[marks]
\begin{axis}[
    xmode=log,
    constructionstyle,
    title={chr19},
    xlabel={peak memory usage [bytes/$n$]},
    ylabel={\phantom{p}},
]

\addplot[bmove] coordinates { (2.25366,0.677471) };
\addlegendentry{algo=build\_bmove};
\addplot[brindex] coordinates { (16.5308,0.154512) };
\addlegendentry{algo=build\_br\_index};
\addplot[columba] coordinates { (15.1399,0.566445) };
\addlegendentry{algo=build\_columba};
\addplot[moverb] coordinates { (0.116973,2.12531) };
\addlegendentry{algo=build\_move\_rb\_move};
\addplot[moverbrlzsa] coordinates { (0.116966,2.04772) };
\addlegendentry{algo=build\_move\_rb\_rlzsa};

\legend{};
\end{axis}
\end{tikzpicture}
\end{subfigure}
\begin{subfigure}[h]{.31\textwidth}
\begin{tikzpicture}[marks]
\begin{axis}[
    xmode=log,
    constructionstyle,
    title={dewiki},
    xlabel={peak memory usage [bytes/$n$]},
    ylabel={\phantom{p}},
]

\addplot[brindex] coordinates { (16.5977,0.187268) };
\addlegendentry{algo=build\_br\_index};
\addplot[moverb] coordinates { (0.454988,1.94918) };
\addlegendentry{algo=build\_move\_rb\_move};
\addplot[moverbrlzsa] coordinates { (0.449633,1.74366) };
\addlegendentry{algo=build\_move\_rb\_rlzsa};

\legend{};
\end{axis}
\end{tikzpicture}
\end{subfigure}

\vspace*{-0.35cm}

\centering
\begin{tikzpicture}[legendmarks]
\begin{axis}[legend,legend columns=5]
\addplot[brindex] coordinates { (0,0) };
\addlegendentry{\texttt{br-index}};
\addplot[bmove] coordinates { (0,0) };
\addlegendentry{\texttt{b-move}};
\addplot[columba] coordinates { (0,0) };
\addlegendentry{\texttt{columba}};
\addplot[moverb] coordinates { (0,0) };
\addlegendentry{\texttt{move-rb}};
\addplot[moverbrlzsa] coordinates { (0,0) };
\addlegendentry{\texttt{move-rb-rlzsa}};
\end{axis}
\end{tikzpicture}

\vspace*{-0.1cm}
\caption{Construction throughput versus peak memory usage during index construction.}
\label{fig:construction-memory-usage}

\vspace{0.2cm}

\begin{subfigure}[h]{.31\textwidth}
\hspace*{-0.4cm}
\begin{tikzpicture}[marks]
\begin{axis}[ymin=1e0,ymax=1e4,
    querystyle_apm,
    xticklabels={},
    title={sars2},
    ylabel={count thr. [$1/s$] (ham.)},
    y label style={at={(0.05,0.5)}},
    xmode=log,
    ymode=log,
    xlabel={\vphantom{pattern length}}
]

\addplot[brindex] coordinates { (10,4.78526) (20,0.774743) (40,2.66542) (80,28.0467) (160,34.8465) (320,23.1973) (640,10.8606) (1280,7.6142) };
\addlegendentry{algo=count\_br\_index\_native};
\addplot[moverb] coordinates { (10,94.4704) (20,102.286) (40,1516.64) (80,5233.77) (160,3875.03) (320,2364.14) (640,1028.44) (1280,680.715) };
\addlegendentry{algo=count\_move\_rb\_move};

\legend{};
\end{axis}
\end{tikzpicture}
\end{subfigure}
\begin{subfigure}[h]{.31\textwidth}
\begin{tikzpicture}[marks]
\begin{axis}[ymin=1e0,ymax=1e6,querystyle_apm,
    xticklabels={},
    title={chr19},
    ylabel={\phantom{p}},
    xmode=log,
    ymode=log,
    xlabel={\vphantom{pattern length}}
]

\addplot[brindex] coordinates { (10,7.10974) (20,3.1641) (40,2.4521) (80,67.521) (160,920.746) (320,4942.02) (640,5639.36) (1280,3648.79) };
\addlegendentry{algo=count\_br\_index\_native};
\addplot[moverb] coordinates { (10,86.7893) (20,303.283) (40,1062.78) (80,13976.2) (160,160543) (320,460172) (640,396108) (1280,275210) };
\addlegendentry{algo=count\_move\_rb\_move};

\legend{};
\end{axis}
\end{tikzpicture}
\end{subfigure}
\begin{subfigure}[h]{.31\textwidth}
\begin{tikzpicture}[marks]
\begin{axis}[ymin=1e-1,ymax=1e6,querystyle_apm,
    xticklabels={},
    title={dewiki},
    ylabel={\phantom{p}},
    xmode=log,
    ymode=log,
    xlabel={\vphantom{pattern length}}
]

\addplot[brindex] coordinates { (10,0.0613695) (20,0.0806703) (40,0.846152) (80,2.48722) (160,20.5811) (320,201.331) (640,390.441) (1280,412.474) };
\addlegendentry{algo=count\_br\_index\_native};
\addplot[moverb] coordinates { (10,26.9207) (20,138.994) (40,5422.5) (80,5650.86) (160,37997.8) (320,215060) (640,365195) (1280,414752) };
\addlegendentry{algo=count\_move\_rb\_move};

\legend{};
\end{axis}
\end{tikzpicture}
\end{subfigure}

\vspace{-1.35cm}
\begin{subfigure}[h]{.31\textwidth}
\hspace*{-0.4cm}
\begin{tikzpicture}[marks]
\begin{axis}[ymin=1e3,ymax=1e8,
    querystyle_apm,
    xticklabels={},
    ylabel={locate thr. [$1/s$] (ham.)},
    y label style={at={(0.05,0.5)}},
    xmode=log,
    ymode=log,
    title={\vphantom{Ig}},
    xlabel={\vphantom{pattern length}}
]

\addplot[brindex] coordinates { (20,452389) (40,57031.8) (80,284056) (160,238553) (320,98866.8) (640,22586) (1280,5671.37) };
\addlegendentry{algo=locate\_br\_index\_native};
\addplot[columba] coordinates { (20,1.87585e+06) (40,1.70734e+06) (80,1.96437e+06) (160,1.94149e+06) (320,1.83237e+06) (640,1.37301e+06) (1280,681204) };
\addlegendentry{algo=locate\_columba};
\addplot[bmove] coordinates { (20,909400) (40,951979) (80,1.10325e+06) (160,1.024e+06) (320,867474) (640,444293) (1280,152998) };
\addlegendentry{algo=locate\_columba\_rlc};
\addplot[moverb] coordinates { (20,6.42541e+06) (40,4.8229e+06) (80,6.34363e+06) (160,5.59317e+06) (320,3.90036e+06) (640,1.34493e+06) (1280,374297) };
\addlegendentry{algo=locate\_move\_rb\_move};
\addplot[moverbrlzsa] coordinates { (20,6.39836e+07) (40,2.55151e+07) (80,4.76206e+07) (160,2.49084e+07) (320,8.97464e+06) (640,1.71732e+06) (1280,403913) };
\addlegendentry{algo=locate\_move\_rb\_rlzsa};

\legend{};
\end{axis}
\end{tikzpicture}
\end{subfigure}
\begin{subfigure}[h]{.31\textwidth}
\begin{tikzpicture}[marks]
\begin{axis}[ymin=1e3,ymax=1e8,querystyle_apm,
    xticklabels={},
    ylabel={\phantom{p}},
    xmode=log,
    ymode=log,
    title={\vphantom{Ig}},
    xlabel={\vphantom{pattern length}}
]

\addplot[brindex] coordinates { (20,372612) (40,56073.2) (80,45957.4) (160,8846.4) (320,18246.6) (640,12594.1) (1280,5729.2) };
\addlegendentry{algo=locate\_br\_index\_native};
\addplot[columba] coordinates { (20,1.09834e+06) (40,1.00289e+06) (80,1.17828e+06) (160,601061) (320,771606) (640,592002) (1280,397340) };
\addlegendentry{algo=locate\_columba};
\addplot[bmove] coordinates { (20,848885) (40,600551) (80,685872) (160,463447) (320,552001) (640,352228) (1280,192018) };
\addlegendentry{algo=locate\_columba\_rlc};
\addplot[moverb] coordinates { (20,5.79088e+06) (40,4.52647e+06) (80,3.74344e+06) (160,1.14393e+06) (320,1.28509e+06) (640,716054) (1280,363177) };
\addlegendentry{algo=locate\_move\_rb\_move};
\addplot[moverbrlzsa] coordinates { (20,4.58147e+07) (40,2.04906e+07) (80,8.69383e+06) (160,1.48744e+06) (320,1.57242e+06) (640,782034) (1280,375528) };
\addlegendentry{algo=locate\_move\_rb\_rlzsa};

\legend{};
\end{axis}
\end{tikzpicture}
\end{subfigure}
\begin{subfigure}[h]{.31\textwidth}
\begin{tikzpicture}[marks]
\begin{axis}[ymin=1e2,ymax=1e8,querystyle_apm,
    xticklabels={},
    ylabel={\phantom{p}},
    xmode=log,
    ymode=log,
    title={\vphantom{Ig}},
    xlabel={\vphantom{pattern length}}
]

\addplot[brindex] coordinates { (10,318378) (20,5118.52) (40,745.944) (80,568.671) (160,146.804) (320,598.418) (640,627.528) (1280,516.473) };
\addlegendentry{algo=locate\_br\_index\_native};
\addplot[moverb] coordinates { (10,7.31258e+06) (20,1.8323e+06) (40,1.19248e+06) (80,581436) (160,230011) (320,538879) (640,529707) (1280,499884) };
\addlegendentry{algo=locate\_move\_rb\_move};
\addplot[moverbrlzsa] coordinates { (10,5.81153e+07) (20,2.95186e+06) (40,1.64334e+06) (80,663769) (160,242821) (320,562836) (640,538545) (1280,495446) };
\addlegendentry{algo=locate\_move\_rb\_rlzsa};

\legend{};
\end{axis}
\end{tikzpicture}
\end{subfigure}

\vspace{-1.35cm}
\begin{subfigure}[h]{.31\textwidth}
\hspace*{-0.4cm}
\begin{tikzpicture}[marks]
\begin{axis}[ymin=1e5,ymax=1e8,
    querystyle_apm,
    xlabel={pattern length},
    ylabel={locate thr. [$1/s$] (edit)},
    y label style={at={(0.05,0.5)}},
    xmode=log,
    ymode=log,,
    title={\vphantom{Ig}}
]

\addplot[columba] coordinates { (20,1.33696e+06) (40,1.42917e+06) (80,1.34394e+06) (160,1.22813e+06) (320,1.55483e+06) (640,1.36112e+06) (1280,828384) };
\addlegendentry{algo=locate\_columba};
\addplot[bmove] coordinates { (20,523839) (40,489525) (80,470758) (160,414654) (320,462121) (640,386244) (1280,229694) };
\addlegendentry{algo=locate\_columba\_rlc};
\addplot[moverb] coordinates { (20,5.27563e+06) (40,4.34548e+06) (80,4.732e+06) (160,4.06949e+06) (320,3.50736e+06) (640,2.01783e+06) (1280,625990) };
\addlegendentry{algo=locate\_move\_rb\_move};
\addplot[moverbrlzsa] coordinates { (20,2.09949e+07) (40,1.33959e+07) (80,1.54486e+07) (160,1.0804e+07) (320,7.61368e+06) (640,2.88973e+06) (1280,686311) };
\addlegendentry{algo=locate\_move\_rb\_rlzsa};

\legend{};
\end{axis}
\end{tikzpicture}
\end{subfigure}
\begin{subfigure}[h]{.31\textwidth}
\begin{tikzpicture}[marks]
\begin{axis}[ymin=1e5,ymax=1e7,querystyle_apm,
    xlabel={pattern length},
    ylabel={\phantom{p}},
    xmode=log,
    ymode=log,,
    title={\vphantom{Ig}}
]

\addplot[columba] coordinates { (20,714669) (40,720644) (80,896454) (160,658656) (320,825919) (640,487781) (1280,282250) };
\addlegendentry{algo=locate\_columba};
\addplot[bmove] coordinates { (20,503202) (40,491082) (80,480209) (160,446471) (320,487548) (640,307223) (1280,156122) };
\addlegendentry{algo=locate\_columba\_rlc};
\addplot[moverb] coordinates { (20,3.96068e+06) (40,2.00234e+06) (80,2.20246e+06) (160,849777) (320,1.01968e+06) (640,600377) (1280,301809) };
\addlegendentry{algo=locate\_move\_rb\_move};
\addplot[moverbrlzsa] coordinates { (20,1.09534e+07) (40,4.41251e+06) (80,3.62966e+06) (160,1.03771e+06) (320,1.20102e+06) (640,652162) (1280,311544) };
\addlegendentry{algo=locate\_move\_rb\_rlzsa};

\legend{};
\end{axis}
\end{tikzpicture}
\end{subfigure}
\begin{subfigure}[h]{.31\textwidth}
\begin{tikzpicture}[marks]
\begin{axis}[ymin=1e5,ymax=1e8,querystyle_apm,
    xlabel={pattern length},
    ylabel={\phantom{p}},
    xmode=log,
    ymode=log,,
    title={\vphantom{Ig}}
]

\addplot[moverb] coordinates { (10,5.16966e+06) (20,1.17815e+06) (40,140602) (80,235716) (160,282849) (320,277565) (640,300730) (1280,403425) };
\addlegendentry{algo=locate\_move\_rb\_move};
\addplot[moverbrlzsa] coordinates { (10,1.34878e+07) (20,1.39526e+06) (40,151395) (80,249623) (160,299757) (320,293016) (640,305030) (1280,406089) };
\addlegendentry{algo=locate\_move\_rb\_rlzsa};

\legend{};
\end{axis}
\end{tikzpicture}
\end{subfigure}

\vspace*{-0.35cm}

\centering
\begin{tikzpicture}[legendmarks]
\begin{axis}[legend,legend columns=5]
\addplot[brindex] coordinates { (0,0) };
\addlegendentry{\texttt{br-index}};
\addplot[bmove] coordinates { (0,0) };
\addlegendentry{\texttt{b-move}};
\addplot[columba] coordinates { (0,0) };
\addlegendentry{\texttt{columba}};
\addplot[moverb] coordinates { (0,0) };
\addlegendentry{\texttt{move-rb}};
\addplot[moverbrlzsa] coordinates { (0,0) };
\addlegendentry{\texttt{move-rb-rlzsa}};
\end{axis}
\end{tikzpicture}

\vspace*{-0.1cm}
\caption{APM query throughput as a geometric mean over the results for $k \in \{4, 7, 10, 13\}$. All indexes use their native APM algorithms.}
\label{fig:apm}

\vspace*{\fill}
\end{figure*}

\subsection{Construction}
We measured construction throughput and peak memory usage of each index (see \Cref{fig:construction-memory-usage}).
\texttt{move-rb} is constructed 3.7$\times$ faster than \texttt{columba}, 10--14$\times$ faster than \texttt{br-index} and 3.1--4.6$\times$ faster than \texttt{b-move}, while requiring 124--129$\times$, 36--141$\times$ and 19--24$\times$ less memory, resp.
The gains over \texttt{b-move} stem from omitting the slow $\PLCP$ construction, from never keeping an $\Oh(n)$-space data structure in RAM, and from further practical improvements (see \Cref{sct:construction}).
The much larger gains over \texttt{br-index} and \texttt{columba} follow from them building $\SA$ and BWT in RAM, which takes $\Oh(n)$ words, whereas \texttt{Big-BWT} stores both on disk.

\subsection{Extension and SA-Interval Enumeration}\label{sct:extension_performance}
\begin{figure*}[t]
\centering
\begin{subfigure}[b]{.31\textwidth}
\hspace*{-0.4cm}
\begin{tikzpicture}[marks]
\begin{axis}[
    querystyle_raw_performance,
    xticklabels={},
    title={sars2},
    ylabel={extension thr. [$1/s$]},
    y label style={at={(0.05,0.5)}},
    xmode=log,
    ymode=log,
    xlabel={\vphantom{pattern length}}
]

\addplot[brindex] coordinates { (10,106347) (20,103544) (40,105505) (80,111336) (160,116834) (320,122866) (640,131476) (1280,145805) };
\addlegendentry{algo=count\_br\_index};
\addplot[columba] coordinates { (10,5.4612e+06) (20,2.51915e+06) (40,2.16882e+06) (80,1.8174e+06) (160,1.56849e+06) (320,1.34705e+06) (640,1.16334e+06) (1280,1.09218e+06) };
\addlegendentry{algo=count\_columba};
\addplot[bmove] coordinates { (10,716370) (20,667508) (40,651113) (80,640639) (160,630945) (320,630673) (640,657086) (1280,748013) };
\addlegendentry{algo=count\_columba\_rlc};
\addplot[moverb] coordinates { (10,1.50716e+06) (20,1.16072e+06) (40,1.13085e+06) (80,1.14412e+06) (160,1.186e+06) (320,1.24731e+06) (640,1.33569e+06) (1280,1.49866e+06) };
\addlegendentry{algo=count\_move\_rb\_move};

\legend{};
\end{axis}
\end{tikzpicture}
\end{subfigure}
\begin{subfigure}[b]{.31\textwidth}
\begin{tikzpicture}[marks]
\begin{axis}[ymin=1e5,ymax=1e7,querystyle_raw_performance,
    xticklabels={},
    title={chr19},
    ylabel={\phantom{p}},
    xmode=log,
    ymode=log,
    xlabel={\vphantom{pattern length}}
]

\addplot[brindex] coordinates { (10,164515) (20,180177) (40,243478) (80,318880) (160,384117) (320,427749) (640,455063) (1280,472666) };
\addlegendentry{algo=count\_br\_index};
\addplot[columba] coordinates { (10,4.42761e+06) (20,2.15597e+06) (40,1.84222e+06) (80,1.76652e+06) (160,1.73629e+06) (320,1.73873e+06) (640,1.74286e+06) (1280,1.75064e+06) };
\addlegendentry{algo=count\_columba};
\addplot[bmove] coordinates { (10,591119) (20,791046) (40,1.22849e+06) (80,1.80732e+06) (160,2.42636e+06) (320,2.94103e+06) (640,3.31331e+06) (1280,3.53677e+06) };
\addlegendentry{algo=count\_columba\_rlc};
\addplot[moverb] coordinates { (10,1.71764e+06) (20,1.64447e+06) (40,2.17938e+06) (80,2.87845e+06) (160,3.48916e+06) (320,3.93049e+06) (640,4.22318e+06) (1280,4.40903e+06) };
\addlegendentry{algo=count\_move\_rb\_move};

\legend{};
\end{axis}
\end{tikzpicture}
\end{subfigure}
\begin{subfigure}[b]{.31\textwidth}
\begin{tikzpicture}[marks]
\begin{axis}[ymin=1e4,ymax=1e7,querystyle_raw_performance,
    xticklabels={},
    title={dewiki},
    ylabel={\phantom{p}},
    xmode=log,
    ymode=log,
    xlabel={\vphantom{pattern length}}
]

\addplot[brindex] coordinates { (10,14857.5) (20,22309.5) (40,38635.5) (80,64560.6) (160,101199) (320,145609) (640,191995) (1280,238160) };
\addlegendentry{algo=count\_br\_index};
\addplot[moverb] coordinates { (10,310647) (20,505017) (40,863676) (80,1.37774e+06) (160,1.98237e+06) (320,2.57758e+06) (640,3.12342e+06) (1280,3.54133e+06) };
\addlegendentry{algo=count\_move\_rb\_move};

\legend{};
\end{axis}
\end{tikzpicture}
\end{subfigure}

\vspace{-1.35cm}
\begin{subfigure}[b]{.31\textwidth}
\hspace*{-0.4cm}
\begin{tikzpicture}[marks]
\begin{axis}[
    querystyle_raw_performance,
    xticklabels={},
    title={},
    ylabel={ext. + enum. thr. [$1/s$]},
    y label style={at={(0.05,0.5)}},
    xlabel={pattern length},
    xmode=log,
    ymode=log,
    title={\vphantom{Ig}}
]

\addplot[brindex] coordinates { (10,915315) (20,917937) (40,918297) (80,919292) (160,921903) (320,926433) (640,933304) (1280,927246) };
\addlegendentry{algo=locate\_br\_index};
\addplot[columba] coordinates { (10,3.14556e+07) (20,3.28118e+07) (40,3.28584e+07) (80,3.26389e+07) (160,3.21158e+07) (320,3.117e+07) (640,2.94662e+07) (1280,2.47202e+07) };
\addlegendentry{algo=locate\_columba};
\addplot[bmove] coordinates { (10,2.12638e+06) (20,2.1212e+06) (40,2.1206e+06) (80,2.13082e+06) (160,2.13777e+06) (320,2.15507e+06) (640,2.19682e+06) (1280,2.24603e+06) };
\addlegendentry{algo=locate\_columba\_rlc};
\addplot[moverb] coordinates { (10,7.14356e+06) (20,7.15541e+06) (40,7.14728e+06) (80,7.14705e+06) (160,7.12838e+06) (320,7.12797e+06) (640,7.08969e+06) (1280,6.93538e+06) };
\addlegendentry{algo=locate\_move\_rb\_move};
\addplot[moverbrlzsa] coordinates { (10,2.13429e+08) (20,2.12514e+08) (40,2.10573e+08) (80,2.09242e+08) (160,2.0407e+08) (320,1.95488e+08) (640,1.70454e+08) (1280,1.12239e+08) };
\addlegendentry{algo=locate\_move\_rb\_rlzsa};

\legend{};
\end{axis}
\end{tikzpicture}
\end{subfigure}
\begin{subfigure}[b]{.31\textwidth}
\begin{tikzpicture}[marks]
\begin{axis}[ymin=1e5,ymax=1e9,querystyle_raw_performance,
    title={},
    ylabel={\phantom{p}},
    xlabel={pattern length},
    xmode=log,
    ymode=log,
    title={\vphantom{Ig}}
]

\addplot[brindex] coordinates { (10,812380) (20,813890) (40,748472) (80,629913) (160,585132) (320,533508) (640,488867) (1280,470556) };
\addlegendentry{algo=locate\_br\_index};
\addplot[columba] coordinates { (10,2.08708e+07) (20,3.09203e+07) (40,2.58802e+07) (80,1.22141e+07) (160,7.65102e+06) (320,4.64134e+06) (640,2.95852e+06) (1280,1.95136e+06) };
\addlegendentry{algo=locate\_columba};
\addplot[bmove] coordinates { (10,1.92883e+06) (20,1.95496e+06) (40,1.89335e+06) (80,1.77426e+06) (160,1.83069e+06) (320,1.91574e+06) (640,2.11213e+06) (1280,2.56042e+06) };
\addlegendentry{algo=locate\_columba\_rlc};
\addplot[moverb] coordinates { (10,7.20505e+06) (20,7.33994e+06) (40,6.5964e+06) (80,5.59795e+06) (160,5.16698e+06) (320,4.64906e+06) (640,4.17185e+06) (1280,3.81075e+06) };
\addlegendentry{algo=locate\_move\_rb\_move};
\addplot[moverbrlzsa] coordinates { (10,2.07692e+08) (20,2.22262e+08) (40,7.49496e+07) (80,2.52268e+07) (160,1.59431e+07) (320,9.69483e+06) (640,6.25099e+06) (1280,4.60775e+06) };
\addlegendentry{algo=locate\_move\_rb\_rlzsa};

\legend{};
\end{axis}
\end{tikzpicture}
\end{subfigure}
\begin{subfigure}[b]{.31\textwidth}
\begin{tikzpicture}[marks]
\begin{axis}[ymin=1e5,ymax=1e9,querystyle_raw_performance,
    title={},
    ylabel={\phantom{p}},
    xlabel={pattern length},
    xmode=log,
    ymode=log,
    title={\vphantom{Ig}}
]

\addplot[brindex] coordinates { (10,1.62629e+06) (20,1.32196e+06) (40,1.04129e+06) (80,909264) (160,673630) (320,646709) (640,347041) (1280,267153) };
\addlegendentry{algo=locate\_br\_index};
\addplot[moverb] coordinates { (10,1.06087e+07) (20,1.20776e+07) (40,9.80729e+06) (80,9.37695e+06) (160,8.16238e+06) (320,8.76773e+06) (640,4.56583e+06) (1280,3.34999e+06) };
\addlegendentry{algo=locate\_move\_rb\_move};
\addplot[moverbrlzsa] coordinates { (10,2.30363e+08) (20,1.56146e+08) (40,6.46538e+07) (80,3.544e+07) (160,1.7199e+07) (320,1.35478e+07) (640,5.48432e+06) (1280,3.55035e+06) };
\addlegendentry{algo=locate\_move\_rb\_rlzsa};

\legend{};
\end{axis}
\end{tikzpicture}
\end{subfigure}

\vspace*{-0.35cm}

\centering
\begin{tikzpicture}[legendmarks]
\begin{axis}[legend,legend columns=5]
\addplot[brindex] coordinates { (0,0) };
\addlegendentry{\texttt{br-index}};
\addplot[bmove] coordinates { (0,0) };
\addlegendentry{\texttt{b-move}};
\addplot[columba] coordinates { (0,0) };
\addlegendentry{\texttt{columba}};
\addplot[moverb] coordinates { (0,0) };
\addlegendentry{\texttt{move-rb}};
\addplot[moverbrlzsa] coordinates { (0,0) };
\addlegendentry{\texttt{move-rb-rlzsa}};
\end{axis}
\end{tikzpicture}

\vspace*{-0.1cm}
\caption{Raw extension and $\SA$-interval enumeration throughput versus pattern length.}
\label{fig:raw_performance}
\end{figure*}
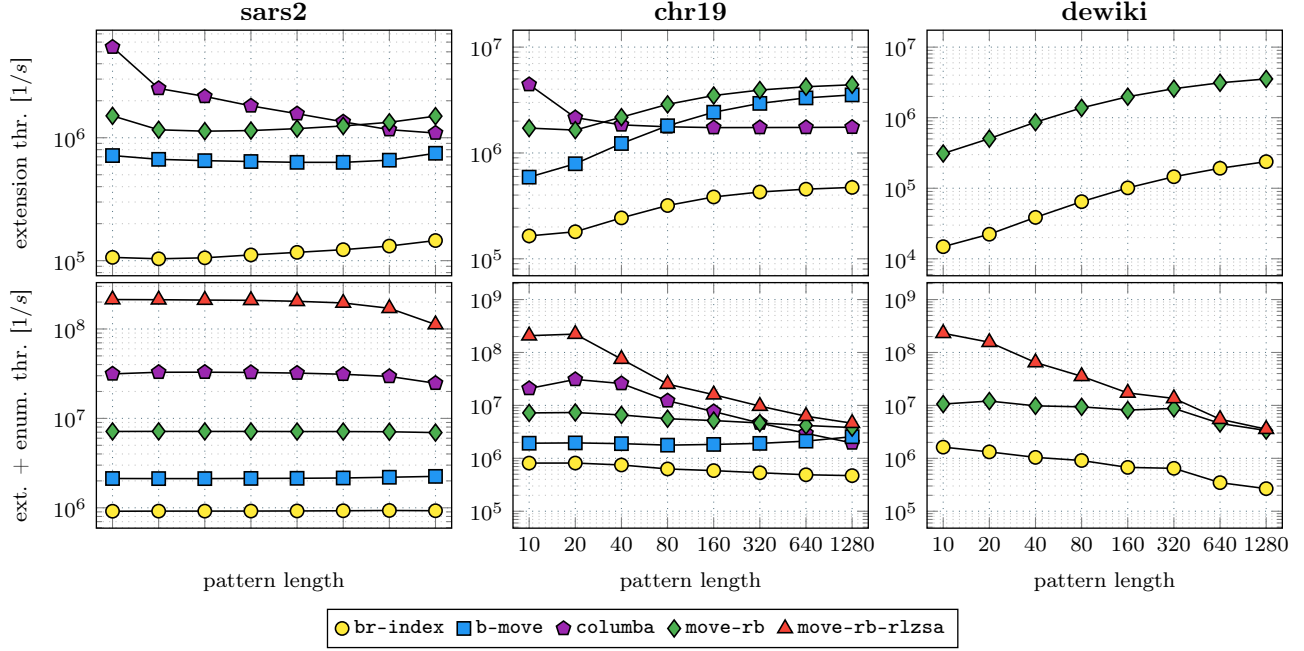

Since all tested indexes use their own APM algorithms, we also evaluate raw index performance in isolation.
For each text and $m \in \{10 \cdot 2^x \mid x \in [0, 7]\}$, we generated two sets of $N$ random substring patterns of length $m$: the first is matched through a fixed, random series of left- and right-extensions, the second additionally enumerates the resulting $\SA$-intervals.
We report the throughputs as $m N / t_{\mathsf{ext}}$ and $(m N + occ) / t_{\mathsf{ext+enum}}$, where $t_{\mathsf{ext}}$ and $t_{\mathsf{ext+enum}}$ are the times to extend, resp.\ to extend and enumerate all $occ$ occurrences of all $N$ patterns in a set (see \Cref{fig:raw_performance}; pattern statistics in \Cref{tab:patterns_count,tab:patterns_locate}).

\texttt{move-rb} extends 9--23$\times$ (typ.\ 12$\times$) faster than \texttt{br-index} and 1.2--2.9$\times$ (typ.\ 1.8$\times$) faster than \texttt{b-move}.
The much larger \texttt{columba} extends faster than \texttt{move-rb} only for short patterns ($m \leq 320$); for longer ones, \texttt{move-rb} is up to 2.5$\times$ faster.
Regarding $\SA$-interval enumeration, \texttt{move-rb} is 6.5--14$\times$ (typ.\ 9$\times$) faster than \texttt{br-index} and 1.5--3.8$\times$ (typ.\ 3$\times$) faster than \texttt{b-move}, while \texttt{columba} enumerates up to 4.6$\times$ faster than \texttt{move-rb}, except for long patterns on chr19 ($m \geq 640$), where \texttt{move-rb} is up to 2.0$\times$ faster.
Finally, \texttt{move-rb-rlzsa} enumerates up to 30$\times$ faster than \texttt{move-rb} when $occ \gg m N$, and 2--10$\times$ (typ.\ 4.4$\times$) faster than \texttt{columba}, while never being slower than \texttt{columba}.
\Cref{sct:raw_index_performance} discusses the trends behind these numbers.

\subsection{Approximate Pattern Matching}
We did not measure edit count queries, because obtaining their number requires filtering the occurrences, hence locating.
\texttt{b-move}'s and \texttt{columba}'s own APM algorithms do not offer count queries, so \Cref{fig:apm} compares counting only against \texttt{br-index}.
While \texttt{br-index}'s APM algorithm is hardcoded to the pigeonhole principle and supports Hamming distance only, \texttt{b-move} and \texttt{move-rb} support arbitrary search schemes and both distance metrics.
We used the currently best-known search schemes, called minU \cite{min_u,search_schemes}.
We generated patterns analogously to \Cref{sct:extension_performance} for each $k \in \{4, 7, 10, 13\}$ (see \Cref{tab:patterns_count_apm,tab:patterns_locate_apm_ham,tab:patterns_locate_apm_edit}); throughput is defined analogously to \Cref{sct:extension_performance}, but with $t_{\mathsf{count}}$ and $t_{\mathsf{locate}}$: the times to count/locate all $N$ patterns in a set.

Here, $occ$ differs between the indexes: w.r.t.\ edit distance, \texttt{b-move} and \texttt{columba} report up to $9.5\times$ fewer occurrences than \texttt{move-rb} (typ.\ $2.7\times$), because their filter violates the coverage condition (see \Cref{sct:coverage}).
We therefore always use $occ_{\mathsf{move-rb}}$ in the numerator, so that every reported speedup is a ratio of running times.

\Cref{fig:apm} shows the results as a geometric mean over $k$, as they span orders of magnitude (OoM) over $k$; the results for specific $k$ are shown in \Cref{sct:apm_per_k}.
At every individual $k$, \texttt{move-rb} and \texttt{move-rb-rlzsa} outperform \texttt{b-move} and \texttt{br-index}, and \texttt{b-move} outperforms \texttt{br-index}; only \texttt{columba}'s rank relative to \texttt{move-rb} varies with $k$.

\subsubsection{Results}
Compared with \texttt{br-index}, \texttt{move-rb} counts patterns w.r.t.\ Hamming distance 1--4 (typ.\ 2.4) OoM faster and locates them 1--3 (typ.\ 2) OoM faster.
Compared with \texttt{b-move}, it locates 1.9--7.5$\times$ (typ.\ 4$\times$) faster w.r.t.\ Hamming distance and 1.9--10$\times$ (typ.\ 4.6$\times$) faster w.r.t.\ edit distance.
For edit distance, the gap is wider on sars2 (typ.\ 7.1$\times$) than on chr19 (typ.\ 3.0$\times$), because sars2 has 3 OoM more occurrences per pattern (see \Cref{tab:patterns_locate_apm_edit}) and \texttt{move-rb}'s $\SA$-interval deduplication saves the most work there (see \Cref{sct:duplicate_nested_sa_intervals}).
Compared with \texttt{columba}, \texttt{move-rb} locates up to 5.3$\times$ faster w.r.t.\ Hamming distance and up to 5.5$\times$ faster w.r.t.\ edit distance (typ.\ 2.1$\times$); \texttt{columba} takes the lead only at the longest patterns, by at most 1.8$\times$ (Hamming, $m \geq 640$) and 1.3$\times$ (edit, $m = 1280$).
\texttt{move-rb-rlzsa} locates up to 10$\times$ (Hamming) and 4$\times$ (edit) faster than \texttt{move-rb}, again with the largest gains where $occ \gg m N$.
Finally, \texttt{move-rb}'s peak memory usage during locating (including the index size and the reported occurrences with their CIGARs) is lower than \texttt{b-move}'s by a factor of 1.5 (at most 6.5) for Hamming distance and 2.5 (at most 7.4) for edit distance, in the geometric mean (see \Cref{tab:mem}).
We store each search context's matched string RLZ-compressed w.r.t.\ $P$ in $\Oh(k)$ space, whereas \texttt{b-move} stores it uncompressed in $\Oh(m)$ space (see \Cref{sct:matched_strings}).

\begin{table*}[t!]
\centering
\scriptsize
\setlength{\tabcolsep}{2.1pt}
\renewcommand{\arraystretch}{1.0}
\begin{tabular}{|rr|rrr|rr|rrr|rr|rrr||rr|r|rr|r|r|}
\hline
\multicolumn{2}{|c|}{} & \multicolumn{13}{c||}{Hamming distance} & \multicolumn{7}{c|}{edit distance} \\
\cline{3-15}\cline{16-22}
\multicolumn{2}{|c|}{} & \multicolumn{5}{c|}{sars2} & \multicolumn{5}{c|}{chr19} & \multicolumn{3}{c||}{dewiki} & \multicolumn{3}{c|}{sars2} & \multicolumn{3}{c|}{chr19} & \multicolumn{1}{c|}{dewiki} \\
\hline
$k$ & $m$ & $\frac{\pk{\iM}}{\isz{\iM}}$ & $\frac{\pk{\iB}}{\isz{\iB}}$ & $\frac{\pks{\iR}}{\isz{\iR}}$ & $\frac{\pk{\iB}}{\pk{\iM}}$ & $\frac{\pks{\iR}}{\pks{\iM}}$ & $\frac{\pk{\iM}}{\isz{\iM}}$ & $\frac{\pk{\iB}}{\isz{\iB}}$ & $\frac{\pks{\iR}}{\isz{\iR}}$ & $\frac{\pk{\iB}}{\pk{\iM}}$ & $\frac{\pks{\iR}}{\pks{\iM}}$ & $\frac{\pk{\iM}}{\isz{\iM}}$ & $\frac{\pks{\iR}}{\isz{\iR}}$ & $\frac{\pks{\iR}}{\pks{\iM}}$ & $\frac{\pk{\iM}}{\isz{\iM}}$ & $\frac{\pk{\iB}}{\isz{\iB}}$ & $\frac{\pk{\iB}}{\pk{\iM}}$ & $\frac{\pk{\iM}}{\isz{\iM}}$ & $\frac{\pk{\iB}}{\isz{\iB}}$ & $\frac{\pk{\iB}}{\pk{\iM}}$ & $\frac{\pk{\iM}}{\isz{\iM}}$ \\
\hline
 \num{4} &   \num{10} & -- & -- & -- & -- & -- & -- & -- & -- & -- & -- & \num{1.04} & \num{1.04} & \num{0.48} & -- & -- & -- & -- & -- & -- & \num{1.04} \\
 \num{4} &   \num{20} & \num{1.03} & \textbf{1.37} &    \num{1.04} & \textbf{1.33} & \num{0.48} & \num{1.24} & \textbf{3.61} & \textbf{1.25} & \textbf{3.80} & \num{0.55} & \num{1.02} & \num{1.03} & \num{0.48} & \num{1.10} & \textbf{1.99} & \textbf{1.80} & \num{1.97} & \textbf{7.93} & \textbf{5.26} & \num{1.02} \\
 \num{4} &   \num{40} & \num{1.01} &    \num{1.19} &    \num{1.03} &    \num{1.17} & \num{0.47} & \num{1.12} & \textbf{1.65} &    \num{1.13} & \textbf{1.93} & \num{0.52} & \num{1.00} & \num{1.01} & \num{0.48} & \num{1.10} & \textbf{2.04} & \textbf{1.85} & \num{1.24} & \textbf{2.90} & \textbf{3.06} & \num{1.00} \\
 \num{4} &   \num{80} & \num{1.01} &    \num{1.19} &    \num{1.02} &    \num{1.17} & \num{0.47} & \num{1.01} &    \num{1.08} &    \num{1.01} & \textbf{1.41} & \num{0.49} & \num{1.00} & \num{1.00} & \num{0.47} & \num{1.10} & \textbf{2.16} & \textbf{1.96} & \num{1.02} &    \num{1.13} & \textbf{1.46} & \num{1.00} \\
 \num{4} &  \num{160} & \num{1.01} &    \num{1.19} &    \num{1.02} &    \num{1.17} & \num{0.47} & \num{1.00} &    \num{1.00} &    \num{1.00} & \textbf{1.31} & \num{0.49} & \num{1.00} & \num{1.00} & \num{0.47} & \num{1.05} & \textbf{2.38} & \textbf{2.26} & \num{1.00} &    \num{1.01} & \textbf{1.32} & \num{1.00} \\
 \num{4} &  \num{320} & \num{1.01} &    \num{1.18} &    \num{1.02} &    \num{1.17} & \num{0.47} & \num{1.00} &    \num{1.00} &    \num{1.00} & \textbf{1.31} & \num{0.49} & \num{1.00} & \num{1.00} & \num{0.47} & \num{1.05} & \textbf{2.77} & \textbf{2.63} & \num{1.00} &    \num{1.00} & \textbf{1.31} & \num{1.00} \\
 \num{4} &  \num{640} & \num{1.01} &    \num{1.10} &    \num{1.02} &    \num{1.08} & \num{0.47} & \num{1.00} &    \num{1.00} &    \num{1.00} & \textbf{1.31} & \num{0.49} & \num{1.00} & \num{1.00} & \num{0.47} & \num{1.05} & \textbf{3.57} & \textbf{3.39} & \num{1.00} &    \num{1.00} & \textbf{1.31} & \num{1.00} \\
 \num{4} & \num{1280} & \num{1.01} &    \num{1.19} &    \num{1.02} &    \num{1.17} & \num{0.47} & \num{1.00} &    \num{1.00} &    \num{1.00} & \textbf{1.31} & \num{0.49} & \num{1.00} & \num{1.00} & \num{0.47} & \num{1.05} & \textbf{3.20} & \textbf{3.03} & \num{1.00} &    \num{1.00} & \textbf{1.31} & \num{1.00} \\
 \num{7} &   \num{20} & \num{1.21} & \textbf{3.93} & \textbf{1.23} & \textbf{3.25} & \num{0.52} & \num{1.24} & \textbf{3.61} & \textbf{1.25} & \textbf{3.80} & \num{0.55} & \num{1.00} & \num{1.00} & \num{0.48} & -- & -- & -- & -- & -- & -- & \num{1.00} \\
 \num{7} &   \num{40} & \num{1.01} &    \num{1.19} &    \num{1.02} &    \num{1.17} & \num{0.47} & \num{1.12} & \textbf{2.31} &    \num{1.13} & \textbf{2.69} & \num{0.52} & \num{1.00} & \num{1.01} & \num{0.48} & \num{1.11} & \textbf{2.91} & \textbf{2.63} & \num{1.50} & \textbf{6.95} & \textbf{6.07} & \num{1.00} \\
 \num{7} &   \num{80} & \num{1.01} &    \num{1.19} &    \num{1.02} &    \num{1.17} & \num{0.47} & \num{1.03} &    \num{1.16} &    \num{1.03} & \textbf{1.48} & \num{0.50} & \num{1.00} & \num{1.00} & \num{0.48} & \num{1.10} & \textbf{2.97} & \textbf{2.69} & \num{1.06} & \textbf{1.93} & \textbf{2.38} & \num{1.00} \\
 \num{7} &  \num{160} & \num{1.01} &    \num{1.19} &    \num{1.02} &    \num{1.17} & \num{0.47} & \num{1.00} &    \num{1.01} &    \num{1.00} & \textbf{1.32} & \num{0.49} & \num{1.00} & \num{1.00} & \num{0.47} & \num{1.11} & \textbf{3.49} & \textbf{3.15} & \num{1.00} &    \num{1.02} & \textbf{1.33} & \num{1.00} \\
 \num{7} &  \num{320} & \num{1.01} &    \num{1.18} &    \num{1.02} &    \num{1.17} & \num{0.47} & \num{1.00} &    \num{1.00} &    \num{1.00} & \textbf{1.31} & \num{0.49} & \num{1.00} & \num{1.00} & \num{0.47} & \num{1.11} & \textbf{4.07} & \textbf{3.68} & \num{1.00} &    \num{1.00} & \textbf{1.31} & \num{1.00} \\
 \num{7} &  \num{640} & \num{1.01} &    \num{1.19} &    \num{1.02} &    \num{1.17} & \num{0.47} & \num{1.00} &    \num{1.00} &    \num{1.00} & \textbf{1.31} & \num{0.49} & \num{1.00} & \num{1.00} & \num{0.47} & \num{1.11} & \textbf{5.24} & \textbf{4.73} & \num{1.00} &    \num{1.00} & \textbf{1.31} & \num{1.00} \\
 \num{7} & \num{1280} & \num{1.01} &    \num{1.10} &    \num{1.02} &    \num{1.08} & \num{0.47} & \num{1.00} &    \num{1.00} &    \num{1.00} & \textbf{1.31} & \num{0.49} & \num{1.00} & \num{1.00} & \num{0.47} & \num{1.11} & \textbf{5.33} & \textbf{4.78} & \num{1.00} &    \num{1.00} & \textbf{1.31} & \num{1.00} \\
\num{10} &   \num{20} & -- & -- & -- & -- & -- & -- & -- & -- & -- & -- & \num{1.00} & \num{1.00} & \num{0.48} & -- & -- & -- & -- & -- & -- & -- \\
\num{10} &   \num{40} & \num{1.01} & \textbf{1.37} &    \num{1.02} & \textbf{1.35} & \num{0.47} & \num{1.12} & \textbf{3.61} &    \num{1.13} & \textbf{4.21} & \num{0.52} & \num{1.00} & \num{1.00} & \num{0.48} & \num{1.21} & \textbf{4.56} & \textbf{3.76} & \num{1.07} & \textbf{1.72} & \textbf{2.11} & \num{1.00} \\
\num{10} &   \num{80} & \num{1.01} & \textbf{1.37} &    \num{1.02} & \textbf{1.35} & \num{0.47} & \num{1.06} & \textbf{1.65} &    \num{1.06} & \textbf{2.04} & \num{0.51} & \num{1.01} & \num{1.01} & \num{0.48} & \num{1.21} & \textbf{4.92} & \textbf{4.06} & \num{1.25} & \textbf{2.55} & \textbf{2.66} & \num{1.00} \\
\num{10} &  \num{160} & \num{1.01} & \textbf{1.37} &    \num{1.02} & \textbf{1.35} & \num{0.47} & \num{1.00} &    \num{1.01} &    \num{1.00} & \textbf{1.32} & \num{0.49} & \num{1.00} & \num{1.00} & \num{0.47} & \num{1.11} & \textbf{3.47} & \textbf{3.14} & \num{1.01} &    \num{1.11} & \textbf{1.44} & \num{1.00} \\
\num{10} &  \num{320} & \num{1.01} &    \num{1.19} &    \num{1.02} &    \num{1.17} & \num{0.47} & \num{1.00} &    \num{1.00} &    \num{1.00} & \textbf{1.31} & \num{0.49} & \num{1.00} & \num{1.00} & \num{0.47} & \num{1.21} & \textbf{6.71} & \textbf{5.54} & \num{1.00} &    \num{1.00} & \textbf{1.32} & \num{1.00} \\
\num{10} &  \num{640} & \num{1.01} & \textbf{1.37} &    \num{1.02} & \textbf{1.35} & \num{0.47} & \num{1.00} &    \num{1.00} &    \num{1.00} & \textbf{1.31} & \num{0.49} & \num{1.00} & \num{1.00} & \num{0.47} & \num{1.11} & \textbf{8.24} & \textbf{7.42} & \num{1.00} &    \num{1.00} & \textbf{1.32} & \num{1.00} \\
\num{10} & \num{1280} & \num{1.01} & \textbf{1.21} &    \num{1.02} &    \num{1.19} & \num{0.47} & \num{1.00} &    \num{1.00} &    \num{1.00} & \textbf{1.31} & \num{0.49} & \num{1.00} & \num{1.00} & \num{0.47} & \num{1.06} & \textbf{7.29} & \textbf{6.88} & \num{1.00} &    \num{1.01} & \textbf{1.32} & \num{1.00} \\
\num{13} &   \num{20} & -- & -- & -- & -- & -- & -- & -- & -- & -- & -- & \num{1.07} & \num{1.08} & \num{0.49} & -- & -- & -- & -- & -- & -- & -- \\
\num{13} &   \num{40} & \num{1.01} &    \num{1.19} &    \num{1.02} &    \num{1.17} & \num{0.47} & \num{1.25} & \textbf{6.22} & \textbf{1.25} & \textbf{6.54} & \num{0.55} & \num{1.00} & \num{1.00} & \num{0.47} & -- & -- & -- & -- & -- & -- & \num{1.00} \\
\num{13} &   \num{80} & \num{1.01} & \textbf{1.37} &    \num{1.02} & \textbf{1.35} & \num{0.47} & \num{1.12} & \textbf{2.31} &    \num{1.13} & \textbf{2.69} & \num{0.52} & \num{1.00} & \num{1.00} & \num{0.47} & \num{1.21} & \textbf{4.91} & \textbf{4.05} & \num{1.07} & \textbf{2.58} & \textbf{3.17} & \num{1.00} \\
\num{13} &  \num{160} & \num{1.01} & \textbf{1.37} &    \num{1.02} & \textbf{1.35} & \num{0.47} & \num{1.00} &    \num{1.08} &    \num{1.00} & \textbf{1.41} & \num{0.49} & \num{1.00} & \num{1.00} & \num{0.47} & \num{1.21} & \textbf{5.43} & \textbf{4.47} & \num{1.00} &    \num{1.06} & \textbf{1.38} & \num{1.00} \\
\num{13} &  \num{320} & \num{1.01} & \textbf{1.37} &    \num{1.02} & \textbf{1.35} & \num{0.47} & \num{1.00} &    \num{1.01} &    \num{1.00} & \textbf{1.32} & \num{0.49} & \num{1.00} & \num{1.00} & \num{0.47} & \num{1.11} & \textbf{4.91} & \textbf{4.42} & \num{1.00} &    \num{1.04} & \textbf{1.36} & \num{1.00} \\
\num{13} &  \num{640} & \num{1.01} & \textbf{1.38} &    \num{1.02} & \textbf{1.36} & \num{0.47} & \num{1.00} &    \num{1.00} &    \num{1.00} & \textbf{1.31} & \num{0.49} & \num{1.00} & \num{1.00} & \num{0.47} & \num{1.21} & \textbf{8.23} & \textbf{6.78} & \num{1.00} &    \num{1.01} & \textbf{1.32} & \num{1.00} \\
\num{13} & \num{1280} & \num{1.01} &    \num{1.20} &    \num{1.02} &    \num{1.18} & \num{0.47} & \num{1.00} &    \num{1.00} &    \num{1.00} & \textbf{1.31} & \num{0.49} & \num{1.00} & \num{1.00} & \num{0.47} & -- & -- & -- & \num{1.00} &    \num{1.01} & \textbf{1.32} & \num{1.00} \\
\hline
\end{tabular}
\caption{
Peak memory usage during locating, per $k$ and $m$, for Hamming and edit distance.
$\iM=$ Move-rb, $\iB=$ b-move and $\iR=$ br-index.
For an index $X$, $\isz{X}$ is its size and $\pk{X}$ its peak memory usage while locating with CIGAR output, \emph{including} the reported occurrences and their CIGARs.
$\pks{X}$ is the peak without CIGAR output; br-index cannot compute CIGARs, so it is compared against Move-rb without CIGAR output as well. An en dash denotes an unavailable measurement.
}
\label{tab:mem}
\end{table*}

To compare the index data structures directly, we also run our optimized APM algorithms on every index (see \Cref{sct:raw_index_amp_performance}).
\texttt{move-rb} then leads \texttt{br-index} by 8--650$\times$ (typ.\ 80$\times$) for counting and 6.6--630$\times$ (typ.\ 46$\times$) for locating, and \texttt{b-move} by 2.2--10.5$\times$ (typ.\ 6$\times$) and 3.2--8.2$\times$ (typ.\ 4.6$\times$), resp.
\texttt{columba} is the only index that ever outperforms \texttt{move-rb}, and only for short patterns: at $m = 10$ for counting (by up to 1.9$\times$) and for $m \leq 320$ for locating (by up to 2.2$\times$); for all other $m$, \texttt{move-rb} leads by 1.5--4.2$\times$ (counting) and 1.3--2.4$\times$ (locating).

\section{Conclusion}
We have shown that existing bi-directional r-indexes can be significantly sped up, both in theory and in practice.
Without a direction switch, Move-rb performs an all-$k$-character extension in output-optimal $\Oh(k)$ time, where $k \leq \sigma$ is the number of characters that can extend the pattern, and a single-character extension in the same time, both within $\Oh((r + \bwd{r}) \log n)$ bits of space.
Augmenting the index with $\Oh((r + \bwd{r}) \log(n / r_{\min}) \log \sigma)$ further bits reduces a single-character extension to $\Oh(\log \sigma)$ time.
A direction switch incurs only $\Oh(\log\log_\omega (n / r_{\min}))$ additional time.
Matching a pattern through $m$ single-character extensions with $s$ direction switches and locating all occurrences thus takes $\Oh(\sum_{i=1}^m k_i + s \log\log_\omega(n / r_{\min}) + occ)$ time with the smaller index, resp.\ $\Oh(m \log \sigma + s \log\log_\omega(n / r_{\min}) + occ)$ time with the augmented one, where $k_i \leq \sigma$ is the number of characters that can extend the pattern before extension $i$.

In practice, Move-rb can be constructed 10--14$\times$ faster than br-index \cite{br_index}, 3--4.6$\times$ faster than b-move \cite{b_move} and 3.7$\times$ faster than columba, while requiring 36--141$\times$, 19--24$\times$ and 124--129$\times$ less memory, resp.
It answers APM queries 1--4 OoM faster than br-index, 1.9--10$\times$ faster than b-move and even up to 5.5$\times$ faster than the uncompressed FM-index columba, while being 2$\times$ larger than br-index, up to 24\% smaller than b-move and 36--42$\times$ smaller than columba.
Its peak memory usage during locating, including the index size, is $1.5\times$ (up to $6.5\times$) and $2.5\times$ (up to $7.4\times$) lower than b-move's for Hamming distance and edit distance, resp.

\bibliographystyle{siamplain}
\bibliography{paper}

\begin{thebibliography}{10}

\bibitem{br_index}
{\sc Y.~Arakawa, G.~Navarro, and K.~Sadakane}, {\em Bi-directional r-indexes}, in 33rd Annual Symposium on Combinatorial Pattern Matching {CPM}, 2022, pp.~11:1--11:14.

\bibitem{wavelet_tree}
{\sc D.~Belazzougui and G.~Navarro}, {\em Optimal lower and upper bounds for representing sequences}, {ACM} Transactions on Algorithms, 11 (2015), pp.~31:1--31:21.

\bibitem{move_r}
{\sc N.~Bertram, J.~Fischer, and L.~Nalbach}, {\em Move-r: Optimizing the r-index}, in 22nd International Symposium on Experimental Algorithms {SEA}, 2024, pp.~1:1--1:19.

\bibitem{pfp}
{\sc C.~Boucher, T.~Gagie, A.~Kuhnle, and G.~Manzini}, {\em Prefix-free parsing for building big {BWTs}}, in 18th International Workshop on Algorithms in Bioinformatics {WABI}, 2018, pp.~2:1--2:16.

\bibitem{bwt}
{\sc M.~Burrows and D.~J. Wheeler}, {\em A block-sorting lossless data compression algorithm}, Tech. Report 124, Digital Equipment Corporation Systems Research Center, Palo Alto, CA, 1994.

\bibitem{search_schemes}
{\sc L.~Depuydt, J.~Fostier, S.~Gottlieb, G.~Kucherov, K.~Reinert, and L.~Renders}, {\em Search schemes for approximate pattern matching: An overview}, in The Expanding World of Compressed Data: {A} Festschrift for Giovanni Manzini's 60th Birthday, 2025, pp.~9:1--9:16.

\bibitem{b_move}
{\sc L.~Depuydt, L.~Renders, S.~V. de~Vyver, L.~Veys, T.~Gagie, and J.~Fostier}, {\em b-move: faster lossless approximate pattern matching in a run-length compressed index}, Algorithms for Molecular Biology, 20 (2025), p.~15.

\bibitem{rlzsa}
{\sc P.~Dinklage, J.~Fischer, L.~Nalbach, and J.~Zumbrink}, {\em {RLZ-r} and {LZ-End-r}: Enhancing {Move-r}}, Computing Research Repository {CoRR}, abs/2507.17300 (2025).

\bibitem{indexing_compressed_text}
{\sc P.~Ferragina and G.~Manzini}, {\em Indexing compressed text}, Journal of the {ACM}, 52 (2005), pp.~552--581.

\bibitem{r_index}
{\sc T.~Gagie, G.~Navarro, and N.~Prezza}, {\em Fully functional suffix trees and optimal text searching in {BWT}-runs bounded space}, Journal of the {ACM}, 67 (2020), pp.~2:1--2:54.

\bibitem{sdsl}
{\sc S.~Gog, T.~Beller, A.~Moffat, and M.~Petri}, {\em From theory to practice: Plug and play with succinct data structures}, in 13th International Symposium on Experimental Algorithms {SEA}, 2014, pp.~326--337.

\bibitem{word_ram}
{\sc T.~Hagerup}, {\em Sorting and searching on the word {RAM}}, in 15th Annual Symposium on Theoretical Aspects of Computer Science {STACS}, 1998, pp.~366--398.

\bibitem{suffix_filter}
{\sc J.~K{\"{a}}rkk{\"{a}}inen and J.~C. Na}, {\em Faster filters for approximate string matching}, in 9th Workshop on Algorithm Engineering and Experiments {ALENEX}, 2007, pp.~84--90.

\bibitem{top_k_colors}
{\sc M.~Karpinski and Y.~Nekrich}, {\em Top-$k$ color queries for document retrieval}, in 22nd Annual {ACM-SIAM} Symposium on Discrete Algorithms {SODA}, 2011, pp.~401--411.

\bibitem{two_bwt}
{\sc T.~W. Lam, R.~Li, A.~Tam, S.~Wong, E.~Wu, and S.~Yiu}, {\em High throughput short read alignment via bi-directional {BWT}}, in {IEEE} International Conference on Bioinformatics and Biomedicine {BIBM}, 2009, pp.~31--36.

\bibitem{suffix_array}
{\sc U.~Manber and G.~Myers}, {\em Suffix arrays: {A} new method for on-line string searches}, in 1st Annual {ACM-SIAM} Symposium on Discrete Algorithms {SODA}, 1990, pp.~319--327.

\bibitem{document_listing}
{\sc S.~Muthukrishnan}, {\em Efficient algorithms for document retrieval problems}, in 13th Annual {ACM-SIAM} Symposium on Discrete Algorithms {SODA}, 2002, pp.~657--666.

\bibitem{repetitiveness_measures}
{\sc G.~Navarro}, {\em Indexing highly repetitive string collections, part {I:} repetitiveness measures}, {ACM} Computing Surveys, 54 (2022), pp.~29:1--29:31.

\bibitem{move_data_structure}
{\sc T.~Nishimoto and Y.~Tabei}, {\em Optimal-time queries on {BWT}-runs compressed indexes}, in 48th International Colloquium on Automata, Languages, and Programming {ICALP}, 2021, pp.~101:1--101:15.

\bibitem{sd_array}
{\sc D.~Okanohara and K.~Sadakane}, {\em Practical entropy-compressed rank/select dictionary}, in 9th Workshop on Algorithm Engineering and Experiments {ALENEX}, 2007, pp.~60--70.

\bibitem{succinct_indexable_dictionary}
{\sc R.~Raman, V.~Raman, and S.~R. Satti}, {\em Succinct indexable dictionaries with applications to encoding $k$-ary trees, prefix sums and multisets}, {ACM} Transactions on Algorithms, 3 (2007), pp.~43:1--43:25.

\bibitem{columba_in_text_verification}
{\sc L.~Renders, L.~Depuydt, and J.~Fostier}, {\em Approximate pattern matching using search schemes and in-text verification}, in 9th International Work-Conference on Bioinformatics and Biomedical Engineering {IWBBIO}, 2022, pp.~419--435.

\bibitem{min_u}
{\sc L.~Renders, L.~Depuydt, S.~Rahmann, and J.~Fostier}, {\em Lossless approximate pattern matching: Automated design of efficient search schemes}, Journal of Computational Biology, 31 (2024), pp.~975--989.

\bibitem{columba_dynamic_partitioning}
{\sc L.~Renders, K.~Marchal, and J.~Fostier}, {\em Dynamic partitioning of search patterns for approximate pattern matching using search schemes}, iScience, 24 (2021), p.~102687.

\bibitem{bidirectional_fm_index}
{\sc T.~Schnattinger, E.~Ohlebusch, and S.~Gog}, {\em Bidirectional search in a string with wavelet trees and bidirectional matching statistics}, Information and Computation, 213 (2012), pp.~13--22.

\end{thebibliography}

\appendix

\section{Optimized Approximate Pattern Matching Algorithms}\label{sct:apm_algorithms}
In this section, we describe our optimized APM algorithms.

The br-index's own APM algorithm is hardcoded to the pigeonhole principle and supports Hamming distance only, while b-move and Move-rb support arbitrary search schemes and both distance metrics.
Our APM algorithms, however, can drive the br-index's data structures for both distance metrics as well; we use this in \Cref{sct:raw_index_amp_performance} to compare all indexes on a common algorithm.

During the execution of a search scheme it is possible that an $\SA$-interval is reported multiple times.
In contrast to b-move, the APM algorithms in br-index and Move-rb maintain the reported $\SA$-intervals in a hash set during the matching phase.
br-index uses the hash set in the \texttt{C++} standard library.
Move-rb uses the space-efficient hash set implementation \texttt{sparse-map} (see our GitHub repository).

\subsection{Hamming Distance}

In br-index's and b-move's APM algorithms, information about the currently matched position in $P$ is recomputed at each character extension.
Our algorithm computes this information only once for each position and search, thus reducing computational overhead.
Both br-index's and b-move's Hamming distance APM algorithms implement the search procedure recursively.
Our algorithm is implemented completely iteratively.
This reduces overhead from function calls and improves cache locality.

For Hamming distance, the CIGAR string of each occurrence is \texttt{``mM''}, indicating $m$ matches/mismatches, so CIGARs are of no use here.
In contrast to b-move and columba, which only support CIGAR output, our algorithm can also compute MD tags, which encode mismatches and their replacement characters.
We compute MD tags in the following way:
During the matching phase, we maintain a stack storing the currently mismatched characters and their positions in the pattern in the order in which they were matched by the algorithm.
Whenever we emit an $\SA$-interval, we construct an MD tag from the stack.

\paragraph{Results.}
\Cref{tab:mem} shows that deduplicating the reported $\SA$-intervals using a hash set also reduces peak memory usage.
Throughout, $\peak$ denotes the index plus the working memory of the search (\Cref{tab:mem}).
For Hamming distance, the peak memory relative to the index size, $\peak/\idx$, has geometric mean $1.03$ for \texttt{move-rb} and $1.35$ for \texttt{b-move} over the two texts \texttt{b-move} supports, with maxima $1.25$ and $6.22$, both at $k = 13$ and $m = 40$ on chr19.
\texttt{move-rb}'s working memory thus stays small throughout (at most $25\%$ of its index size), whereas \texttt{b-move}'s can grow to several times its index: \texttt{b-move} enumerates duplicate $\SA$-intervals, which inflates its intermediate occurrence array, whereas \texttt{move-rb} deduplicates them with a hash set during the matching phase.
Together with \texttt{b-move}'s $1.3\times$ larger index on chr19 (see \Cref{tab:texts}), this makes its peak memory $1.5\times$ that of \texttt{move-rb} in the geometric mean, and $6.5\times$ at most.
Without CIGAR output, \texttt{br-index}'s peak memory is $0.49\times$ that of \texttt{move-rb} in the geometric mean, close to the ratio of their index sizes (see \Cref{tab:texts}).

\subsection{Edit Distance}
The matching phase of b-move's (and columba's) edit distance APM algorithm is sophisticated.
It uses optimizations from the papers \cite{columba_dynamic_partitioning, min_u, columba_in_text_verification}.
According to \cite{columba_dynamic_partitioning}, it reports only a subset of all possible occurrences $\Occ_\edit(T, P, k) = \{(p, l, e) \mid d_\edit(P, T[p, p + l - 1]) = e \leq k\}$.
We adopt the definition of \emph{redundant} occurrences from \cite{columba_dynamic_partitioning}: An occurrence is considered redundant when another match of equal or smaller edit distance lies within $2k+1$ positions.
Thus, we want to compute a set $\Occ$ with the following \emph{coverage} property: for each $(p, l, e) \in \Occ_\edit(T, P, k)$, there is a $(p', l', e') \in \Occ$ with $|p' - p| \leq 2k+1$ and $e' \leq e$.
The matching phase of \cite{columba_dynamic_partitioning} already reports occurrences that satisfy coverage (see \Cref{lem:matching_phase}), except for those clipped at the boundaries of $T$. A subsequent filtering step, however, removes some of these reported occurrences in a way that breaks coverage.

Our algorithm builds upon b-move's algorithm.
It runs essentially the same search as b-move, but filters the occurrences in a coverage-preserving manner (see \Cref{sct:coverage}), and it finds occurrences clipped at the boundaries of $T$ (see \Cref{sct:boundary_occurrences}).
On request, it reports a CIGAR alignment together with the exact edit distance of each occurrence (see \Cref{sct:cigars}); to do so without access to the full text $T$, it keeps each context's matched string compressed in $\Oh(k)$ space (see \Cref{sct:matched_strings}) rather than $\Oh(m)$ space like b-move.
Our algorithm further improves on b-move by employing several minor optimizations, such as simpler code, fewer arithmetic operations, fewer function calls and fewer function parameters.
However, the main performance optimization of our algorithm is that we avoid enumerating duplicate and nested reported $\SA$-intervals (see \Cref{sct:duplicate_nested_sa_intervals}).
We first recall search schemes and b-move's algorithm before detailing these improvements.

\subsubsection{Search schemes}\label{sct:search_schemes}
Both b-move and our algorithm match a pattern approximately within the framework of search schemes~\cite{search_schemes}, which we recall briefly.

A search scheme $\mathcal{S} = \{S_1, \ldots, S_N\}$ consists of a set of searches.
Each search $S_j = (\pi, L, U)$ splits the pattern into the same number $t = |\pi| = |L| = |U|$ of parts $P_1\cdots P_t$, and specifies that in the $i$-th iteration,
part $P_{\pi[i]}$ is matched, and only those contexts with at least $L[i]$ and at most $U[i]$ errors are continued in the next iteration.
Let $O = O_1\cdots O_t$ be an approximate occurrence (string) of $P$ in $T$, where $O_i$ denotes the part aligned to $P_i$.
Then we call the array $e[1..t]$ the error configuration of $O$ w.r.t.\ $P$, where $e[i]$ is the cumulative number of errors up to the $i$-th part (in matching order), i.e.,\ $e[i] = \sum_{j = 1}^id(O_{\pi[j]}, P_{\pi[j]})$.
We call a search valid iff it satisfies the connectivity property, which ensures that in each iteration, the part $P_{\pi[i]}$ to match is directly adjacent to the already matched parts.
More formally, it holds $\forall i \in [2, t] : \pi[i] \in \{\min \pi[1, i - 1] - 1, \max \pi[1, i - 1] + 1\}$.
Finally, we call a search scheme lossless iff it covers every possible error configuration.
We assume a lossless search scheme throughout.
See \cite{search_schemes} for an overview of state-of-the-art search schemes.

\subsubsection{b-move's algorithm}\label{sct:bmove_algorithm}
b-move executes the searches of the scheme, then enumerates all reported $\SA$-intervals and filters the resulting occurrences.

\paragraph{Executing a search.}
b-move extends the currently matched substring character by character in the part's direction, branching over the up to $\sigma$ possible extensions and maintaining, per part, a banded edit-distance matrix that bounds the accumulated errors and prunes an extension once the bound is exceeded.
The $\SA$-intervals of the substrings that match a part within the error bound are continued in the next part, and those surviving after the last part report their $\SA$-intervals.
More precisely, a surviving context reports the occurrences of the \emph{cluster centers} of the final column of its banded matrix, i.e.,\ those cells that are local minima of the column's values~\cite{columba_dynamic_partitioning}; every other cell is dominated by a neighboring cell and is not reported.
Our algorithm executes the search in the same way, so we do not repeat its details; the differences described below concern only how the reported occurrences are represented, enumerated and filtered, and how occurrences at the text boundaries are recovered.

\paragraph{Occurrence enumeration and filtering.}
Let $I = \{[p_1, q_1], \ldots, [p_N, q_N]\}$ be the multiset of $\SA$-intervals reported across all searches of the scheme.
b-move enumerates the $\SA$-values of every interval in $I$ and filters the resulting occurrences in 2 stages:

Stage 1: Occurrences with the same starting position are reduced to one.
For each position, the occurrence with minimum edit distance has priority, and ties are broken in favor of the shortest occurrence.

Stage 2: The occurrences are processed in ascending order of their starting position while a separate array of kept occurrences is maintained.
The current occurrence is appended to the output array, unless it lies within $2k$ positions of the most recently kept occurrence, in which case only the better of the two is kept (the one with the smaller edit distance, or with equal edit distance and shorter length).

During the filtering, Stage 1 alone does not violate coverage.
However, Stage 2 does.
Since it compares the current occurrence only with the most recently kept one, the kept occurrence can drift arbitrarily far to the right: each discard is justified only against its immediate predecessor.
For example, let $k \geq 2$ and consider three reported occurrences $o_1 = (1, l, 2)$, $o_2 = (2k + 1, l, 1)$ and $o_3 = (4k + 1, l, 0)$, in the notation $(p, l, e)$ of $\Occ_\edit$.
Processing them left to right, $o_2$ replaces $o_1$ (they are $2k$ apart and $o_2$ has the smaller edit distance), and then $o_3$ replaces $o_2$ for the same reason, so only $o_3$ remains.
But now no kept occurrence lies within $2k+1$ positions of $o_1$, since $|(4k + 1) - 1| = 4k > 2k+1$, so coverage is violated.

\subsubsection{Ensuring the coverage condition}\label{sct:coverage}
Our algorithm performs Stage 1 implicitly through the sequence $E$ of \Cref{sct:duplicate_nested_sa_intervals}, and replaces b-move's incorrect Stage 2 by a corrected, two-sided variant, where each discarded occurrence is compared against a kept occurrence on \emph{each} side, rather than only against its predecessor.
Let $o_1$ and $o_2$ be the last two kept occurrences, with positions $p_1 < p_2$ and edit distances $d_1$ and $d_2$, and let $o_3$ (position $p_3$, edit distance $d_3$) be the occurrence currently processed.
We discard the middle occurrence $o_2$ whenever it is dominated on both sides, i.e.,\ $\dom(o_1, o_2, o_3) = (d_1 \leq d_2) \land (d_3 \leq d_2) \land (p_3 - p_1 \leq 2k+3)$ holds; this test is repeated, so several consecutive occurrences may be discarded before $o_3$ is appended.

The matching phase reports only the cluster centers of each context (see \Cref{sct:bmove_algorithm}), hence not an occurrence at \emph{every} position of $\Occ_\edit(T, P, k)$. It does, however, still cover $\Occ_\edit(T, P, k)$, and with the small radius $k$ that we rely on below.

\begin{lemma}[{\cite[Lemma~1]{columba_dynamic_partitioning}}]\label{lem:matching_phase}
    For each $(p, l, e) \in \Occ_\edit(T, P, k)$, the matching phase reports an occurrence $(p', l', e')$ with $e' \leq e$ and $|p' - p| \leq e - e' \leq k$.
\end{lemma}

\begin{theorem}\label{thm:coverage}
    The occurrences reported by our algorithm satisfy the coverage condition.
\end{theorem}
\begin{proof}
    Fix $(p, l, e) \in \Occ_\edit(T, P, k)$.
    By \Cref{lem:matching_phase}, Stage 2 receives an occurrence $o = (p', l', e')$ with $e' \leq e$ and $|p' - p| \leq e - e' \leq k$.
    We keep an occurrence, say at position $p''$, within $k+1$ of $p'$ and of edit distance at most $e'$; the coverage condition then follows from $|p'' - p| \leq (e - e') + (k+1) \leq 2k+1$.

    If $o$ is kept, we are done.
    Otherwise, $o$ is discarded by some $\dom$-test, with the last kept occurrence $a$ on its left and the currently processed occurrence $b$ on its right (see \Cref{alg:coverage_filter}).
    The test gives $a.p < p' < b.p$, $b.p - a.p \leq 2k+3$ and $a.d, b.d \leq e'$.
    Each later discard of $a$ or $b$ is again a $\dom$-test; it replaces the discarded occurrence by its neighbor on the same side, which is farther from $p'$ and of edit distance at most $e'$, and keeps $b.p - a.p \leq 2k+3$.
    When the sweep ends, $a$ and $b$ are both kept, so the nearer of the two is within $\floor{(2k+3)/2} = k+1$ of $p'$ and of edit distance at most $e'$.
\end{proof}

\Cref{alg:coverage_filter} realizes this rule in a single left-to-right sweep.
It receives the reported occurrences in an array $O[1..M]$, sorted by ascending starting position and already reduced by Stage 1 (see \Cref{sct:duplicate_nested_sa_intervals}); each occurrence $o = \langle p, l, d \rangle$ exposes its starting position $o.p$, length $o.l$ and edit distance $o.d$.
The algorithm keeps the surviving occurrences as a prefix $O[1..w]$: before appending the current occurrence $O[i]$, it repeatedly discards the last kept occurrence $O[w]$ as long as it is dominated by its two neighbors $O[w-1]$ and $O[i]$.
Since $w \leq i$ holds whenever we write to $O[w]$, each such write overwrites an already-processed entry of $O$; hence the algorithm runs in-place, i.e.,\ no separate output array is created.
This reduces running time and memory usage, especially when locating short patterns with many occurrences, where the output occurrence array accounts for most of the peak memory.

\begin{algorithm}[!h]
\caption{$\mathsf{Filter\text{-}In\text{-}Place}(O[1..M], k)$}\label{alg:coverage_filter}
\begin{algorithmic}[1]
\STATE $w \gets 0$;
\FOR{$i \gets 1$ \textbf{to} $M$}
    \WHILE{$w \geq 2 \land \dom(O[w-1], O[w], O[i])$}
        \STATE $w \gets w - 1$;
    \ENDWHILE
    \STATE $w \gets w + 1$;
    \STATE $O[w] \gets O[i]$;
\ENDFOR
\RETURN $O[1..w]$;
\end{algorithmic}
\end{algorithm}

\subsubsection{Recovering occurrences clipped at text boundaries}\label{sct:boundary_occurrences}
b-move's algorithm misses occurrences of $P$ whose alignment deletes a prefix of $P$ before position $1$ of $T$ or a suffix of $P$ after position $n - 1$.
We recover these occurrences directly.
Using the bidirectional index, we decode $T[1, w]$ using the backward index and $T[n - w, n - 1]$, where $w = \min(n - 1, m + k)$, using the forward index, applying the standard BWT inversion algorithm with $\LF$.
We then align $P$ against $T[1, w]$ and, symmetrically, $\bwd{P}$ against $\bwd{T}[1, w] = \bwd{T[n - w, n - 1]}$, using a banded edit-distance matrix, and report the best occurrence starting at position $1$ and the best one ending at position $n - 1$.
Since all occurrences clipped at the start begin at position $1$, and all clipped at the end lie within $2k+1$ positions of each other, the two reported occurrences suffice for coverage.

\subsubsection{Maintaining the matched strings in compressed space}\label{sct:matched_strings}
In both the b-move algorithm and our algorithm, an occurrence's length, its edit distance and its CIGAR alignment are all derived from the \emph{matched string} $S$: the string that the search contexts represent.
The error count accumulated during the search is the sum of the per-part alignment costs; it upper-bounds, but need not equal, $d_\edit(P, S)$, and $S$ may even be longer than the optimal alignment, i.e.,\ a prefix of $S$ may yield an alignment with a smaller edit distance.
Before enumerating the occurrences of a reported context, we therefore align $P$ against the prefix $S[1, l]$ of its matched string $S$ that minimizes $d_\edit(P, S[1, l])$, using a banded edit-distance matrix.
This yields the length $l$ and the exact edit distance $d = d_\edit(P, S[1, l])$; we call $(d, l)$ the \emph{distance-length tuple} of the context.

\paragraph{RLZ-compressed matched strings.}
b-move stores each context's matched string as a plain array of characters, copied whenever a part is entered and appended to at every extension, requiring $\Oh(m)$ time and space and a heap allocation per extension.
We instead store the matched strings compressed using relative Lempel--Ziv (RLZ) with respect to $P$.
Since each occurrence deviates from $P$ in at most $k$ positions, its matched string decomposes into $\leq k+1$ RLZ phrases w.r.t.\ $P$.
A context $A$ extending some context $B$ extends $B$'s phrases (to the left or right) using a reference pointer to $B$'s most recently added phrase.
Thereby, each extension adds only $\Oh(1)$ extra space.

\paragraph{Discarding unused phrases.}
Since a context references the phrases of the context it extends, a phrase is shared by all contexts that descend from the one which created it.
As the search creates and prunes contexts constantly, we must free a phrase as soon as it is not used by any context; else the space would grow with the total number of contexts explored, instead of with the number of contexts kept simultaneously.
However, we cannot free a phrase when the context that created it is pruned, because its descendants may still reference it.
Reference counting detects the right moment to discard a phrase:
Once a phrase is no longer referenced, it recursively decrements the reference counters of all phrases it consists of by following their chain of reference pointers.

\paragraph{Storing the phrases in a pool.}
By storing all phrases in a pool with a free-list, we have two advantages:
First, we can implement reference pointers as indices into the pool, thus saving space.
Second, no phrase incurs a heap allocation: freeing a phrase only marks its slot in the pool as unused, and every new phrase reuses such a slot, both in $\Oh(1)$ time.

\paragraph{Resulting space.}
The overall space occupied by the phrases is $\Oh(c_{\mathsf{max}} k)$, where $c_{\mathsf{max}}$ is the maximum number of simultaneously kept contexts, which is inherent to the algorithm and already optimized to be low.
In b-move's algorithm, the space occupied by the matched strings is $\Oh(c_{\mathsf{max}} m)$, and each extension incurs a heap allocation.

\subsubsection{Avoiding enumeration of duplicate and nested SA-intervals}\label{sct:duplicate_nested_sa_intervals}
b-move enumerates the $\SA$-values of every interval in $I$ separately before filtering the resulting occurrences.
Since the intervals in $I$ may coincide or be nested within one another, the same $\SA$-value is enumerated several times, and the subsequent filtering re-examines each copy.
We eliminate both sources of redundant work with the following two optimizations.

\paragraph{Optimization 1: Avoiding redundant $\SA$-interval enumerations.}
It suffices to enumerate the \emph{fundamental} intervals, i.e.,\ the distinct intervals in $I$ that are not properly contained in any other interval in $I$.

\paragraph{Optimization 2: Avoiding redundant occurrence filtering.}
Stage 1 of the occurrence filtering does not have to be performed explicitly:
Let $S_1, \ldots, S_N$ be the matched strings corresponding to the intervals in $I$, and let $(d_i, l_i)$ be the distance-length tuple of $S_i$ (see \Cref{sct:matched_strings}).
Let $f$ be a function that maps each position in an interval in $I$ to its minimum distance-length tuple, i.e.,\ $f(i) = \min \{(d_j, l_j) \mid j \in [1, N] \land i \in [p_j, q_j]\}$, where $i \in \cup \ I$.
Consider the sequence $E = \langle p'_1, q'_1, d'_1, l'_1 \rangle, \ldots, \langle p'_M, q'_M, d'_M, l'_M \rangle$, where each interval $[p'_i, q'_i]$ is maximal s.t.\ $\forall j \in [p'_i, q'_i] : f(j) = (d'_i, l'_i)$ and $p'_1 < \cdots < p'_M$.
Intuitively, $E$ stores the maximal regions in $I$, where the occurrences have the same (minimum) distance-length tuples.
If we have $E$, then we can implicitly perform filtering Stage 1 while enumerating the fundamental intervals, sorted by their starting positions:
While enumerating a fundamental $\SA$-interval at position $i$, we keep track of the interval $[p'_j, q'_j] \ni i$ in $E$ and report $\SA[i]$ with distance $d'_j$ and length $l'_j$.
$E$ can be computed with a sweep-line algorithm in $\Oh(N \log N)$ time and $\Oh(N)$ space.
The sweep line is the current position $j$ in $\SA$, and the state consists of $f(j)$ and a stack containing all intervals in $I$ that contain $j$.

\subsubsection{Computing CIGAR strings and MD tags}\label{sct:cigars}
For each reported context, its CIGAR is obtained by tracing back the very banded matrix (see \Cref{sct:matched_strings}) that computes its distance-length tuple $(d, l)$, so a single dynamic program yields both the exact distance and the alignment, a sequence of match, mismatch, insertion and deletion operations.
Since all occurrences reported by one context share the same matched string, the alignment is computed once per reported $\SA$-interval rather than once per occurrence.
We store it the same way: our algorithm returns one CIGAR per reported $\SA$-interval in a single array and keeps with each occurrence only an index into this array, so a CIGAR occupies space once per interval rather than once per occurrence.
columba and b-move instead attach a full copy of the CIGAR to every occurrence, duplicating it across all occurrences of the same interval; the more occurrences a pattern has, the more this inflates their output (which is included in $\pk{(\cdot)}$ in \Cref{tab:mem}).
While columba and b-move are not capable of computing the SAM MD tag, we store with each mismatch and deletion the corresponding pattern/text character, so that the aligned string and its SAM MD tag can be derived.

\paragraph{Results.}
Recall that \texttt{b-move} stores the matched strings uncompressed in $\Oh(c_{\mathsf{max}} m)$ space (see \Cref{sct:matched_strings}) and stores and enumerates duplicate and nested $\SA$-intervals (see \Cref{sct:duplicate_nested_sa_intervals}), whereas our algorithm needs $\Oh(c_{\mathsf{max}} k)$ space for the matched strings, enumerates only the fundamental intervals and filters the occurrences in place (\Cref{alg:coverage_filter}).
The edit-distance columns of \Cref{tab:mem} show the effect.

\texttt{move-rb}'s peak memory usage stays close to its index size: over the two texts \texttt{b-move} supports, the ratio $\peak/\idx$ has geometric mean $1.10$ and reaches at most $1.97$, i.e.,\ the working memory is typically $10\%$ of the index; on dewiki it is typically below $1\%$ ($1.003$ in the geometric mean, at most $1.04$).
For \texttt{b-move}, the same ratio is $2.37$ in the geometric mean and reaches $8.24$.
Together with \texttt{b-move}'s $1.3\times$ larger index on chr19 (see \Cref{tab:texts}), this makes its peak memory $2.5\times$ that of \texttt{move-rb} in the geometric mean, at least $1.3\times$, and up to $7.4\times$ on sars2 for $k = 10$ and $m = 640$.

Two trends are visible.
First, the factor grows with $k$ (from $2.1\times$ at $k = 4$ to $2.6$--$2.8\times$ for $k \geq 7$), because a larger error bound keeps more contexts alive simultaneously, and \texttt{b-move} pays $\Oh(m)$ space per context where we pay $\Oh(k)$.
Second, it is much larger on sars2 ($3.6\times$ in the geometric mean) than on chr19 ($1.7\times$), because sars2 has 3 OoM more occurrences per pattern (see \Cref{tab:patterns_locate_apm_edit}), so \texttt{b-move}'s repeated enumeration of duplicate and nested $\SA$-intervals inflates its intermediate occurrence array the most there.
The maximum of $7.4\times$ combines both effects, as it occurs at a large $k$ and a large $m$ on sars2.

\section{Raw Index Performance}\label{sct:raw_index_performance}
\Cref{fig:raw_performance} shows the raw extension and $\SA$-interval enumeration throughput.
\texttt{br-index} is the slowest index on every text, at every $m$ and for both workloads.

\paragraph{Extension.}
Going from $m = 10$ to $m = 1280$, \texttt{move-rb}'s throughput stays flat on sars2 (it varies by a factor of $1.3$ over all $m$), but grows by $2.6\times$ on chr19 and by $11\times$ on dewiki.
\texttt{move-rb}'s per-extension cost does not depend on how many BWT sub-runs the $\SA$-interval spans: the scans in \Cref{sct:practical_extensions} are confined to the at most two blocks of $L'$ containing $\mu$ and $\nu$, whose boundary answers are precomputed in $L'_\pred$ and $L'_\suc$.
\texttt{move-rb} does, however, benefit from a small interval: as $b$ and $e$, and hence $\mu$ and $\nu$, move closer together, the two $\MLF$-$\move$-queries and the lookups on $L'$ access nearby memory, and the two blocks eventually coincide.
On sars2, this effect is absent because the $\SA$-intervals stay large: even at $m = 1280$, a pattern still has $2.7 \cdot 10^5$ occurrences on average, whereas on chr19 the intervals have collapsed to a few hundred (see \Cref{tab:patterns_count}).

\texttt{b-move} follows the same trend, but more pronounced: on chr19, its throughput grows by $6\times$ over the same range.
The reason is that \texttt{b-move} answers $\pred_c$- and $\suc_c$-queries by scanning the runs covered by the $\SA$-interval ($\Oh(r)$ time).
Once the interval collapses, this scan becomes cheap, and \texttt{move-rb}'s lead shrinks from 2.9$\times$ ($m = 10$) to 1.2$\times$ ($m = 1280$).
On sars2, where the intervals stay large, the scan stays expensive and the lead remains at 1.7--2.1$\times$.

\texttt{columba} cannot profit from a shrinking interval at all, as an FM-index extension costs the same regardless of the interval size; its throughput even decreases with $m$ (by $5\times$ on sars2 and $2.5\times$ on chr19).
It is therefore the fastest index for short patterns ($m \leq 320$ on sars2, $m \leq 20$ on chr19) and is overtaken by \texttt{move-rb} beyond that.

\paragraph{$\SA$-interval enumeration.}
Here, the throughput is governed by the number $occ$ of reported occurrences relative to the matching work $m N$.
\texttt{move-rb-rlzsa}'s advantage over \texttt{move-rb} is largest when $occ \gg m N$ (short patterns with many occurrences), because it decodes contiguous ranges of $\SA$ by copying from its reference, whereas $\MPhi$ and $\MPhiinv$ perform one $\move$-query per occurrence.
The advantage shrinks as $m$ grows and $occ$ falls: on chr19 and dewiki, where the occurrences per pattern drop by more than 3 OoM between $m = 10$ and $m = 1280$ (see \Cref{tab:patterns_locate}), it decays to 1.2$\times$ and 1.1$\times$, resp.; on sars2, where the intervals stay large, it stays at 16--30$\times$ throughout.

\texttt{columba} enumerates 3.6--4.6$\times$ faster than \texttt{move-rb} on sars2 at every $m$, and up to 4.2$\times$ faster on chr19 for $m \leq 160$; at $m = 320$ the two are tied, and for $m \geq 640$ \texttt{move-rb} is faster (by 2.0$\times$ at $m = 1280$).
\texttt{columba}'s direct $\SA$ access pays off exactly when there are many occurrences to report, and stops paying off once the extension phase accounts for most of the running time.
\texttt{columba}'s index is 15--28$\times$ larger than \texttt{move-rb-rlzsa}'s (see \Cref{tab:texts}), which nevertheless enumerates faster at every measured $m$.

\section{Raw Index APM Performance}\label{sct:raw_index_amp_performance}
\begin{figure*}[t]
\centering
\input{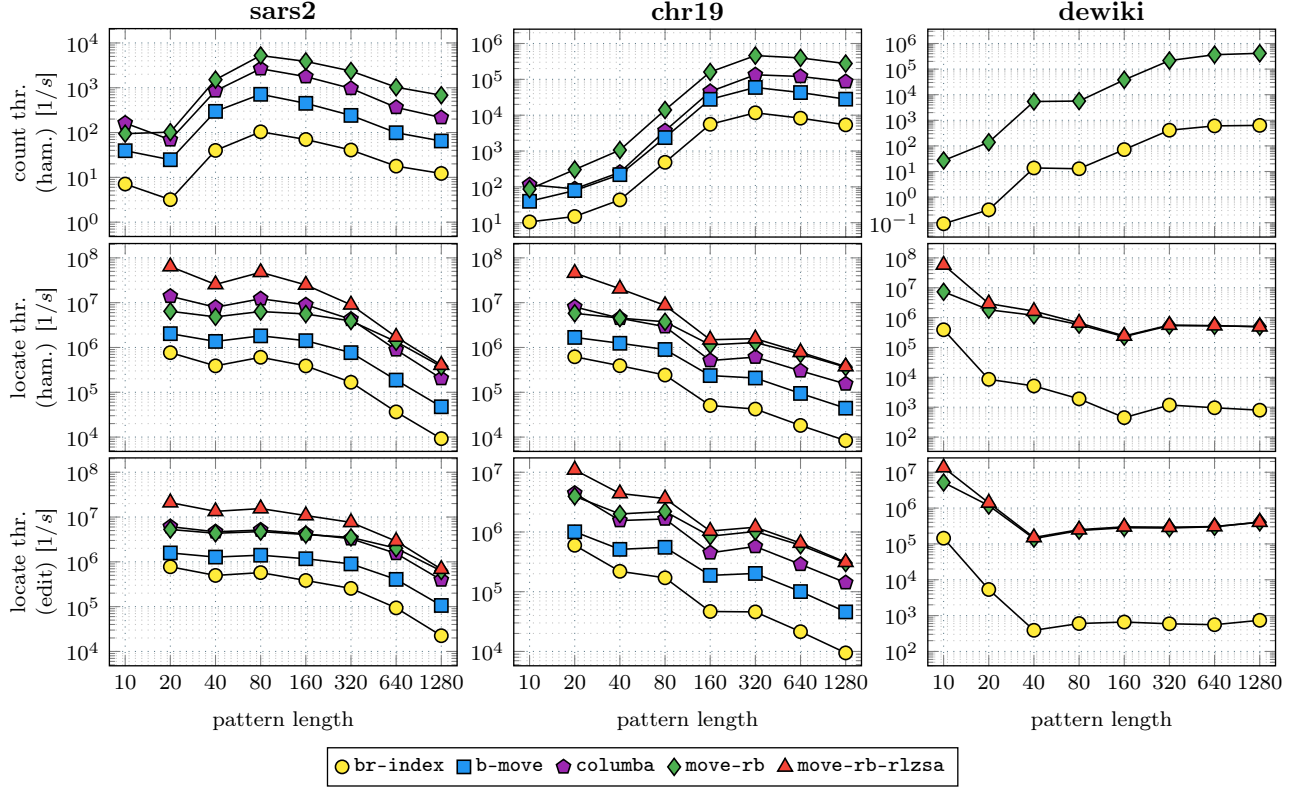}
\caption{APM query throughput as a geometric mean over the results for $k \in \{4, 7, 10, 13\}$. All indexes use our optimized APM algorithms.}
\label{fig:apm_same_alg}
\end{figure*}
\Cref{fig:apm_same_alg} shows the APM query throughput when all indexes run \texttt{move-rb}'s APM algorithms, with each curve a geometric mean over $k \in \{4, 7, 10, 13\}$.
The algorithm is now fixed, so the remaining differences are due to the index data structures alone.

For counting w.r.t.\ Hamming distance, \texttt{move-rb} leads \texttt{br-index} by 13--58$\times$ (typ.\ 41$\times$) on sars2, 8--51$\times$ (typ.\ 28$\times$) on chr19 and 2.5--2.8 (typ.\ 2.7) OoM on dewiki, and it leads \texttt{b-move} by 2.2--10.5$\times$ (typ.\ 6.1$\times$).
For locating, it leads \texttt{br-index} by 6.6--44$\times$ (typ.\ 17$\times$) on the DNA texts and by 1.3--2.8 (typ.\ 2.4) OoM on dewiki, and \texttt{b-move} by 3.2--8.2$\times$ (typ.\ 4.6$\times$).
The leads over \texttt{br-index} and over \texttt{b-move} both grow with $m$ and with $k$: over \texttt{b-move}, from typ.\ 5.5$\times$ at $k = 4$ to typ.\ 8.1$\times$ at $k = 13$ for counting, and from typ.\ 4.3$\times$ to 5.8$\times$ for locating.
The reason is that a larger $k$ or $m$ means more extensions per pattern, so the per-extension cost of the index carries more weight; that cost is exactly what this experiment isolates.
Only the uncompressed \texttt{columba} is ever faster than \texttt{move-rb}: for counting only at $m = 10$ (by up to 1.9$\times$), and for locating only for $m \leq 320$ (by up to 2.2$\times$).
For all other $m$, \texttt{move-rb} leads \texttt{columba} by 1.5--4.2$\times$ (counting) and 1.3--2.4$\times$ (locating), and its lead over \texttt{columba} also grows with $k$ (typ.\ 2.1$\times$ at $k = 4$ to 3.1$\times$ at $k = 13$ for counting).
Thus, \texttt{move-rb} remains the fastest compressed index even when every index runs our APM algorithms.

\section{Per-k APM Performance}\label{sct:apm_per_k}
\Cref{fig:apm_native_k4,fig:apm_native_k7,fig:apm_native_k10,fig:apm_native_k13} show the native APM query throughput individually for each error bound $k \in \{4, 7, 10, 13\}$, and \Cref{fig:apm_same_alg_k4,fig:apm_same_alg_k7,fig:apm_same_alg_k10,fig:apm_same_alg_k13} show the same measurements with all indexes running our optimized APM algorithms.
At every individual $k$, \texttt{move-rb} and \texttt{move-rb-rlzsa} outperform \texttt{b-move} and \texttt{br-index}, and \texttt{b-move} outperforms \texttt{br-index}, so the geometric means of \Cref{fig:apm,fig:apm_same_alg} are representative.
Only \texttt{columba}'s rank relative to \texttt{move-rb} varies with $k$: on sars2 at $m = 640$, for instance, \texttt{move-rb} leads at $k \in \{4, 10\}$ and \texttt{columba} at $k \in \{7, 13\}$.
A larger $k$ enlarges the search space and the number of reported occurrences, which generally lowers the throughput of every index.

\paragraph{Omitted measurements.}
Some of the shortest $m$ are missing, for all indexes alike, and the more, the larger $k$ is (see \Cref{tab:patterns_count_apm,tab:patterns_locate_apm_ham,tab:patterns_locate_apm_edit}): at $k = 13$ on sars2 and chr19, for instance, only $m \geq 40$ is measured w.r.t.\ Hamming distance, and only $m \geq 80$ w.r.t.\ edit distance (on sars2, also $m = 1280$ is missing).
We omitted these measurements because the number of occurrences per pattern, and hence the memory needed to locate them, explodes there.
Already at the smallest $m$ we still measure, a pattern has up to $2.2 \cdot 10^7$ occurrences (sars2, Hamming distance, $k = 7$, $m = 20$; see \Cref{tab:patterns_locate_apm_ham}), and halving $m$ multiplies this number by up to almost 3 OoM.

How the speedups change with $k$ differs per index.
Against \texttt{br-index}, \texttt{move-rb}'s native lead grows with $k$: on the DNA texts, from typ.\ 1.7 OoM ($k = 4$) to typ.\ 2.4 OoM ($k = 13$) for counting w.r.t.\ Hamming distance, and from typ.\ 1.3 to 2.2 OoM for locating.
The reason is that \texttt{br-index} is hardcoded to the pigeonhole principle, whose search space grows much faster with $k$ than that of minU.
Against \texttt{b-move}, the native lead is roughly constant in $k$ w.r.t.\ Hamming distance (typ.\ 3.5--3.9$\times$), and also w.r.t.\ edit distance on sars2 (typ.\ 6.1--7.8$\times$).
Only on chr19 does it fall with $k$ w.r.t.\ edit distance (typ.\ 3.4$\times$ to 2.1$\times$), following the number of occurrences per pattern on which \texttt{move-rb}'s $\SA$-interval deduplication can save.
With the algorithms held fixed, all leads grow with $k$: over \texttt{b-move}, from typ.\ 5.5$\times$ ($k = 4$) to 8.1$\times$ ($k = 13$) for counting and from 4.3$\times$ to 5.8$\times$ for locating; over \texttt{columba}, from 2.1$\times$ to 3.1$\times$ for counting.
A larger $k$ means more extensions per pattern, so the per-extension cost of the index, which this experiment isolates, carries more weight.
Finally, with the native algorithms, \texttt{columba} overtakes \texttt{move-rb} in 9 of 52 Hamming and 5 of 47 edit measurements, by at most 2.5$\times$, resp.\ 2.1$\times$ (both at $k = 7$ and $m = 1280$ on sars2).
All but two of them are at $m \geq 640$, so they largely average out in \Cref{fig:apm}.
\section{Query Pattern Statistics}\label{sct:patt_tables}
For every text and pattern length $m$, $k \in \{4, 7, 10, 13\}$ and for the three APM pattern sets \Cref{tab:patterns_count,tab:patterns_locate,tab:patterns_count_apm,tab:patterns_locate_apm_ham,tab:patterns_locate_apm_edit} report the number $N$ of generated patterns and the average number $\floor{occ / N}$ of occurrences per pattern.

\paragraph{Pattern generation and measurement.}
We generate the extension pattern sets (\Cref{tab:patterns_count,tab:patterns_locate}) with our tool \texttt{move-rb-gen-ext-queries} and the APM pattern sets (\Cref{tab:patterns_count_apm,tab:patterns_locate_apm_ham,tab:patterns_locate_apm_edit}) with \texttt{move-rb-gen-apm-queries}, and we run the corresponding measurements with \texttt{move-rb-bench-ext-queries} and \texttt{move-rb-bench-apm-queries}.
Each pattern is a substring of $T$ starting at a uniformly random position, and all indexes are run on the same pattern sets.
Rather than fixing $N$, each generator calibrates it per pattern set with \texttt{move-rb-rlzsa}, our fastest index: it keeps adding patterns until \texttt{move-rb-rlzsa}'s query time on the set reaches $X$ seconds (the \texttt{-{}-time} parameter), but always includes at least $Y$ patterns (the \texttt{-{}-min} parameter), so that even the fastest index runs long enough to be timed reliably while no set is reduced to a single pattern.
For the APM sets, we use $X = 5$ for \texttt{sars2} and \texttt{chr19}, and $X = 1$ for \texttt{dewiki}, where the gap between \texttt{br-index} and \texttt{move-rb} is so large that a larger $X$ would make \texttt{br-index}'s measurements prohibitively slow; for the extension sets, $X = 1$ throughout.
Consequently, $N$ is small exactly where a single query is expensive, i.e.,\ for a large $k$ and a small $m$: there, a handful of patterns already reaches $X$ seconds, and each of them contributes millions of occurrences (see \Cref{tab:patterns_count_apm,tab:patterns_locate_apm_ham,tab:patterns_locate_apm_edit}).
The floor $Y$ binds only in these cases; we set $Y = 4$ for \texttt{sars2} and \texttt{chr19}, and $Y = 2$ for \texttt{dewiki}, whereas for the extension sets $N$ always stays far above the floor.
Each benchmark then repeats the whole pattern set, per index, until at least $10$ seconds have elapsed (its own \texttt{-{}-time} parameter), and reports the average time per run; this times every index accurately while keeping the total measurement time bounded.

\setcounter{section}{4}
\setcounter{figure}{0}

\begin{figure*}[p]
\centering
\vspace*{\fill}
\begin{subfigure}[b]{.31\textwidth}
\hspace*{-0.7cm}
\begin{tikzpicture}[marks]
\begin{axis}[ymin=1e-5,ymax=1e-1,
    querystyle_kfig,
    title={sars2},
    ylabel={count thr.\\ (ham.) [1/$\mu$s]},
    y label style={at={(0.0,0.5)},align=center},
    xmin=10,
    xmax=1280,
    xticklabels={},
    xmode=log,
    ymode=log,
    xlabel={\vphantom{pattern length}}
]

\addplot[brindex] coordinates { (10,2.19709e-05) (20,2.89395e-05) (40,0.00020914) (80,0.000198459) (160,0.000134097) (320,7.80153e-05) (640,4.11102e-05) (1280,2.05946e-05) };
\addlegendentry{algo=count\_br\_index\_native};
\addplot[moverb] coordinates { (10,0.000459527) (20,0.00352825) (40,0.0154566) (80,0.0131948) (160,0.00920761) (320,0.00553387) (640,0.00295338) (1280,0.00143675) };
\addlegendentry{algo=count\_move\_rb\_move};

\legend{};
\end{axis}
\end{tikzpicture}
\end{subfigure}
\begin{subfigure}[b]{.31\textwidth}
\begin{tikzpicture}[marks]
\begin{axis}[ymin=1e-5,ymax=1e0,querystyle_kfig,
    title={chr19},
    ylabel={\phantom{p}},
    xmin=10,
    xmax=1280,
    xticklabels={},
    xmode=log,
    ymode=log,
    xlabel={\vphantom{pattern length}}
]

\addplot[brindex] coordinates { (10,3.28163e-05) (20,2.86382e-05) (40,0.00034288) (80,0.00270969) (160,0.0169777) (320,0.0177846) (640,0.0121985) (1280,0.00852923) };
\addlegendentry{algo=count\_br\_index\_native};
\addplot[moverb] coordinates { (10,0.000424701) (20,0.00204283) (40,0.0250445) (80,0.186175) (160,0.696222) (320,0.73817) (640,0.605838) (1280,0.460457) };
\addlegendentry{algo=count\_move\_rb\_move};

\legend{};
\end{axis}
\end{tikzpicture}
\end{subfigure}
\begin{subfigure}[b]{.31\textwidth}
\begin{tikzpicture}[marks]
\begin{axis}[ymin=1e-6,ymax=1e0,querystyle_kfig,
    title={dewiki},
    ylabel={\phantom{p}},
    xmin=10,
    xmax=1280,
    xticklabels={},
    xmode=log,
    ymode=log,
    xlabel={\vphantom{pattern length}}
]

\addplot[brindex] coordinates { (10,9.30524e-07) (20,6.03164e-06) (40,6.97353e-05) (80,0.000187695) (160,0.000648553) (320,0.000826678) (640,0.000820943) (1280,0.000812535) };
\addlegendentry{algo=count\_br\_index\_native};
\addplot[moverb] coordinates { (10,0.000362932) (20,0.00722399) (40,0.0355053) (80,0.115473) (160,0.328826) (320,0.479355) (640,0.560841) (1280,0.611009) };
\addlegendentry{algo=count\_move\_rb\_move};

\legend{};
\end{axis}
\end{tikzpicture}
\end{subfigure}

\vspace{-1.35cm}
\begin{subfigure}[b]{.31\textwidth}
\hspace*{-0.7cm}
\begin{tikzpicture}[marks]
\begin{axis}[ymin=1e-2,ymax=1e2,
    querystyle_kfig,
    ylabel={locate thr.\\ (ham.) [1/$\mu$s]},
    y label style={at={(0.0,0.5)},align=center},
    xmin=10,
    xmax=1280,
    xticklabels={},
    xmode=log,
    ymode=log,
    title={\vphantom{Ig}},
    xlabel={\vphantom{pattern length}}
]

\addplot[brindex] coordinates { (20,0.675933) (40,0.849427) (80,0.759015) (160,0.556007) (320,0.272855) (640,0.0640838) (1280,0.0228983) };
\addlegendentry{algo=locate\_br\_index\_native};
\addplot[columba] coordinates { (20,1.9373) (40,1.9639) (80,1.99574) (160,1.99421) (320,1.9984) (640,1.80008) (1280,1.31244) };
\addlegendentry{algo=locate\_columba};
\addplot[bmove] coordinates { (20,0.988408) (40,0.991161) (80,1.1352) (160,1.12303) (320,1.1238) (640,0.784414) (1280,0.389176) };
\addlegendentry{algo=locate\_columba\_rlc};
\addplot[moverb] coordinates { (20,6.80191) (40,6.86713) (80,6.70023) (160,6.32287) (320,5.10282) (640,2.32239) (1280,0.924499) };
\addlegendentry{algo=locate\_move\_rb\_move};
\addplot[moverbrlzsa] coordinates { (20,80.6563) (40,99.0588) (80,76.0882) (160,45.5184) (320,17.5174) (640,3.50105) (1280,1.08221) };
\addlegendentry{algo=locate\_move\_rb\_rlzsa};

\legend{};
\end{axis}
\end{tikzpicture}
\end{subfigure}
\begin{subfigure}[b]{.31\textwidth}
\begin{tikzpicture}[marks]
\begin{axis}[ymin=1e-2,ymax=1e2,querystyle_kfig,
    ylabel={\phantom{p}},
    xmin=10,
    xmax=1280,
    xticklabels={},
    xmode=log,
    ymode=log,
    title={\vphantom{Ig}},
    xlabel={\vphantom{pattern length}}
]

\addplot[brindex] coordinates { (20,0.69268) (40,0.614168) (80,0.122316) (160,0.103758) (320,0.0615708) (640,0.0270059) (1280,0.0131952) };
\addlegendentry{algo=locate\_br\_index\_native};
\addplot[columba] coordinates { (20,1.1605) (40,1.40851) (80,1.53934) (160,1.6042) (320,1.0457) (640,0.751617) (1280,0.529152) };
\addlegendentry{algo=locate\_columba};
\addplot[bmove] coordinates { (20,0.898605) (40,1.02896) (80,0.926489) (160,0.962157) (320,0.824057) (640,0.560872) (1280,0.355151) };
\addlegendentry{algo=locate\_columba\_rlc};
\addplot[moverb] coordinates { (20,6.37893) (40,6.26065) (80,3.68791) (160,2.6794) (320,1.83046) (640,1.04502) (1280,0.594867) };
\addlegendentry{algo=locate\_move\_rb\_move};
\addplot[moverbrlzsa] coordinates { (20,84.7002) (40,67.4534) (80,8.5387) (160,4.3445) (320,2.42243) (640,1.1874) (1280,0.624085) };
\addlegendentry{algo=locate\_move\_rb\_rlzsa};

\legend{};
\end{axis}
\end{tikzpicture}
\end{subfigure}
\begin{subfigure}[b]{.31\textwidth}
\begin{tikzpicture}[marks]
\begin{axis}[ymin=1e-3,ymax=1e2,querystyle_kfig,
    ylabel={\phantom{p}},
    xmin=10,
    xmax=1280,
    xticklabels={},
    xmode=log,
    ymode=log,
    title={\vphantom{Ig}},
    xlabel={\vphantom{pattern length}}
]

\addplot[brindex] coordinates { (10,0.318378) (20,0.123975) (40,0.038455) (80,0.0052307) (160,0.0035344) (320,0.00228963) (640,0.0012547) (1280,0.00140827) };
\addlegendentry{algo=locate\_br\_index\_native};
\addplot[moverb] coordinates { (10,7.31258) (20,8.32097) (40,3.89776) (80,1.6513) (160,1.31769) (320,1.07876) (640,0.789537) (1280,1.0322) };
\addlegendentry{algo=locate\_move\_rb\_move};
\addplot[moverbrlzsa] coordinates { (10,58.1153) (20,29.59) (40,7.59809) (80,2.06854) (160,1.53158) (320,1.17793) (640,0.810177) (1280,1.02473) };
\addlegendentry{algo=locate\_move\_rb\_rlzsa};

\legend{};
\end{axis}
\end{tikzpicture}
\end{subfigure}

\vspace{-1.35cm}
\begin{subfigure}[b]{.31\textwidth}
\hspace*{-0.7cm}
\begin{tikzpicture}[marks]
\begin{axis}[ymin=1e-1,ymax=1e2,
    querystyle_kfig,
    ylabel={locate thr.\\ (edit) [1/$\mu$s]},
    y label style={at={(0.0,0.5)},align=center},
    xmin=10,
    xmax=1280,
    xlabel={pattern length},
    xmode=log,
    ymode=log,
    title={\vphantom{Ig}}
]

\addplot[columba] coordinates { (20,1.33696) (40,1.40358) (80,1.66197) (160,1.46907) (320,1.65369) (640,1.63496) (1280,0.900762) };
\addlegendentry{algo=locate\_columba};
\addplot[bmove] coordinates { (20,0.523839) (40,0.527205) (80,0.521015) (160,0.533955) (320,0.525206) (640,0.467673) (1280,0.302525) };
\addlegendentry{algo=locate\_columba\_rlc};
\addplot[moverb] coordinates { (20,5.27563) (40,5.32147) (80,5.3756) (160,5.19526) (320,4.70726) (640,2.932) (1280,0.943634) };
\addlegendentry{algo=locate\_move\_rb\_move};
\addplot[moverbrlzsa] coordinates { (20,20.9949) (40,22.8852) (80,22.6716) (160,18.865) (320,14.2163) (640,5.00375) (1280,1.08542) };
\addlegendentry{algo=locate\_move\_rb\_rlzsa};

\legend{};
\end{axis}
\end{tikzpicture}
\end{subfigure}
\begin{subfigure}[b]{.31\textwidth}
\begin{tikzpicture}[marks]
\begin{axis}[ymin=1e-1,ymax=1e2,querystyle_kfig,
    ylabel={\phantom{p}},
    xmin=10,
    xmax=1280,
    xlabel={pattern length},
    xmode=log,
    ymode=log,
    title={\vphantom{Ig}}
]

\addplot[columba] coordinates { (20,0.714669) (40,0.7716) (80,1.38412) (160,1.50846) (320,1.26261) (640,0.668534) (1280,0.473091) };
\addlegendentry{algo=locate\_columba};
\addplot[bmove] coordinates { (20,0.503202) (40,0.537056) (80,0.754164) (160,0.856875) (320,0.727268) (640,0.499424) (1280,0.294791) };
\addlegendentry{algo=locate\_columba\_rlc};
\addplot[moverb] coordinates { (20,3.96068) (40,4.60671) (80,2.93385) (160,2.26569) (320,1.64984) (640,0.978064) (1280,0.547116) };
\addlegendentry{algo=locate\_move\_rb\_move};
\addplot[moverbrlzsa] coordinates { (20,10.9534) (40,16.8865) (80,5.33006) (160,3.42276) (320,2.11036) (640,1.1013) (1280,0.575407) };
\addlegendentry{algo=locate\_move\_rb\_rlzsa};

\legend{};
\end{axis}
\end{tikzpicture}
\end{subfigure}
\begin{subfigure}[b]{.31\textwidth}
\begin{tikzpicture}[marks]
\begin{axis}[ymin=1e-1,ymax=1e2,querystyle_kfig,
    ylabel={\phantom{p}},
    xmin=10,
    xmax=1280,
    xlabel={pattern length},
    xmode=log,
    ymode=log,
    title={\vphantom{Ig}}
]

\addplot[moverb] coordinates { (10,5.16966) (20,7.1933) (40,2.21285) (80,0.902516) (160,1.19696) (320,1.10304) (640,0.8404) (1280,0.678855) };
\addlegendentry{algo=locate\_move\_rb\_move};
\addplot[moverbrlzsa] coordinates { (10,13.4878) (20,9.75258) (40,2.72512) (80,1.00716) (160,1.39104) (320,1.24501) (640,0.871189) (1280,0.68895) };
\addlegendentry{algo=locate\_move\_rb\_rlzsa};

\legend{};
\end{axis}
\end{tikzpicture}
\end{subfigure}

\vspace*{-0.35cm}

\centering
\begin{tikzpicture}[legendmarks]
\begin{axis}[legend,legend columns=5]
\addplot[brindex] coordinates { (0,0) };
\addlegendentry{\texttt{br-index}};
\addplot[bmove] coordinates { (0,0) };
\addlegendentry{\texttt{b-move}};
\addplot[columba] coordinates { (0,0) };
\addlegendentry{\texttt{columba}};
\addplot[moverb] coordinates { (0,0) };
\addlegendentry{\texttt{move-rb}};
\addplot[moverbrlzsa] coordinates { (0,0) };
\addlegendentry{\texttt{move-rb-rlzsa}};
\end{axis}
\end{tikzpicture}

\vspace*{-0.1cm}
\caption{APM query throughput for $k = 4$. All indexes use their native APM algorithms.}
\label{fig:apm_native_k4}
\vspace{0.2cm}
\begin{subfigure}[b]{.31\textwidth}
\hspace*{-0.7cm}
\begin{tikzpicture}[marks]
\begin{axis}[ymin=1e-7,ymax=1e-2,
    querystyle_kfig,
    title={sars2},
    ylabel={count thr.\\ (ham.) [1/$\mu$s]},
    y label style={at={(0.0,0.5)},align=center},
    xmin=10,
    xmax=1280,
    xticklabels={},
    xmode=log,
    ymode=log,
    xlabel={\vphantom{pattern length}}
]

\addplot[brindex] coordinates { (10,1.04223e-06) (20,4.09062e-07) (40,7.77287e-06) (80,5.70538e-05) (160,4.46482e-05) (320,2.79597e-05) (640,1.04701e-05) (1280,7.66606e-06) };
\addlegendentry{algo=count\_br\_index\_native};
\addplot[moverb] coordinates { (10,1.94214e-05) (20,9.07842e-05) (40,0.00562784) (80,0.00648198) (160,0.00472656) (320,0.00263037) (640,0.000871668) (1280,0.000692836) };
\addlegendentry{algo=count\_move\_rb\_move};

\legend{};
\end{axis}
\end{tikzpicture}
\end{subfigure}
\begin{subfigure}[b]{.31\textwidth}
\begin{tikzpicture}[marks]
\begin{axis}[ymin=1e-7,ymax=1e0,querystyle_kfig,
    title={chr19},
    ylabel={\phantom{p}},
    xmin=10,
    xmax=1280,
    xticklabels={},
    xmode=log,
    ymode=log,
    xlabel={\vphantom{pattern length}}
]

\addplot[brindex] coordinates { (10,1.54035e-06) (20,3.49585e-07) (40,5.93273e-06) (80,0.00016867) (160,0.00172435) (320,0.00874931) (640,0.00632812) (1280,0.00403415) };
\addlegendentry{algo=count\_br\_index\_native};
\addplot[moverb] coordinates { (10,1.77357e-05) (20,4.50262e-05) (40,0.00339664) (80,0.0274337) (160,0.322002) (320,0.53869) (640,0.426646) (1280,0.296565) };
\addlegendentry{algo=count\_move\_rb\_move};

\legend{};
\end{axis}
\end{tikzpicture}
\end{subfigure}
\begin{subfigure}[b]{.31\textwidth}
\begin{tikzpicture}[marks]
\begin{axis}[ymin=1e-8,ymax=1e0,querystyle_kfig,
    title={dewiki},
    ylabel={\phantom{p}},
    xmin=10,
    xmax=1280,
    xticklabels={},
    xmode=log,
    ymode=log,
    xlabel={\vphantom{pattern length}}
]

\addplot[brindex] coordinates { (10,4.04742e-09) (20,2.13699e-07) (40,5.83513e-07) (80,6.79829e-06) (160,5.53013e-05) (320,0.000430937) (640,0.000497252) (1280,0.000466629) };
\addlegendentry{algo=count\_br\_index\_native};
\addplot[moverb] coordinates { (10,1.99686e-06) (20,0.000563964) (40,0.006993) (80,0.0099148) (160,0.105362) (320,0.337911) (640,0.412754) (1280,0.456294) };
\addlegendentry{algo=count\_move\_rb\_move};

\legend{};
\end{axis}
\end{tikzpicture}
\end{subfigure}

\vspace{-1.35cm}
\begin{subfigure}[b]{.31\textwidth}
\hspace*{-0.7cm}
\begin{tikzpicture}[marks]
\begin{axis}[ymin=1e-3,ymax=1e2,
    querystyle_kfig,
    ylabel={locate thr.\\ (ham.) [1/$\mu$s]},
    y label style={at={(0.0,0.5)},align=center},
    xmin=10,
    xmax=1280,
    xticklabels={},
    xmode=log,
    ymode=log,
    title={\vphantom{Ig}},
    xlabel={\vphantom{pattern length}}
]

\addplot[brindex] coordinates { (20,0.302776) (40,0.235299) (80,0.528356) (160,0.295258) (320,0.126916) (640,0.0208552) (1280,0.00376055) };
\addlegendentry{algo=locate\_br\_index\_native};
\addplot[columba] coordinates { (20,1.81635) (40,1.95772) (80,1.96443) (160,2.01085) (320,1.93255) (640,1.49767) (1280,0.631024) };
\addlegendentry{algo=locate\_columba};
\addplot[bmove] coordinates { (20,0.836707) (40,1.22659) (80,1.21813) (160,1.2044) (320,1.04134) (640,0.518687) (1280,0.139173) };
\addlegendentry{algo=locate\_columba\_rlc};
\addplot[moverb] coordinates { (20,6.06974) (40,6.71857) (80,6.52241) (160,5.84843) (320,4.23066) (640,1.15997) (1280,0.254308) };
\addlegendentry{algo=locate\_move\_rb\_move};
\addplot[moverbrlzsa] coordinates { (20,50.7574) (40,72.6055) (80,56.0231) (160,28.2323) (320,10.3053) (640,1.40594) (1280,0.265735) };
\addlegendentry{algo=locate\_move\_rb\_rlzsa};

\legend{};
\end{axis}
\end{tikzpicture}
\end{subfigure}
\begin{subfigure}[b]{.31\textwidth}
\begin{tikzpicture}[marks]
\begin{axis}[ymin=1e-2,ymax=1e2,querystyle_kfig,
    ylabel={\phantom{p}},
    xmin=10,
    xmax=1280,
    xticklabels={},
    xmode=log,
    ymode=log,
    title={\vphantom{Ig}},
    xlabel={\vphantom{pattern length}}
]

\addplot[brindex] coordinates { (20,0.200439) (40,0.179801) (80,0.062434) (160,0.0139013) (320,0.0319544) (640,0.0143688) (1280,0.00642364) };
\addlegendentry{algo=locate\_br\_index\_native};
\addplot[columba] coordinates { (20,1.03952) (40,1.20989) (80,1.2279) (160,0.862847) (320,0.808049) (640,0.661456) (1280,0.524872) };
\addlegendentry{algo=locate\_columba};
\addplot[bmove] coordinates { (20,0.801917) (40,0.834989) (80,0.817109) (160,0.661366) (320,0.699043) (640,0.470339) (1280,0.264848) };
\addlegendentry{algo=locate\_columba\_rlc};
\addplot[moverb] coordinates { (20,5.25705) (40,6.0081) (80,3.7806) (160,1.7323) (320,1.43843) (640,0.771241) (1280,0.389456) };
\addlegendentry{algo=locate\_move\_rb\_move};
\addplot[moverbrlzsa] coordinates { (20,24.7814) (40,50.3216) (80,8.98017) (160,2.39821) (320,1.78808) (640,0.845708) (1280,0.405428) };
\addlegendentry{algo=locate\_move\_rb\_rlzsa};

\legend{};
\end{axis}
\end{tikzpicture}
\end{subfigure}
\begin{subfigure}[b]{.31\textwidth}
\begin{tikzpicture}[marks]
\begin{axis}[ymin=1e-4,ymax=1e1,querystyle_kfig,
    ylabel={\phantom{p}},
    xmin=10,
    xmax=1280,
    xticklabels={},
    xmode=log,
    ymode=log,
    title={\vphantom{Ig}},
    xlabel={\vphantom{pattern length}}
]

\addplot[brindex] coordinates { (20,0.00444148) (40,0.00571822) (80,0.000346262) (160,0.000119767) (320,0.0014295) (640,0.000729952) (1280,0.00054556) };
\addlegendentry{algo=locate\_br\_index\_native};
\addplot[moverb] coordinates { (20,2.76106) (40,3.25883) (80,0.462798) (160,0.258071) (320,0.832719) (640,0.579994) (1280,0.493812) };
\addlegendentry{algo=locate\_move\_rb\_move};
\addplot[moverbrlzsa] coordinates { (20,3.84945) (40,5.81706) (80,0.503909) (160,0.266206) (320,0.872488) (640,0.588171) (1280,0.48813) };
\addlegendentry{algo=locate\_move\_rb\_rlzsa};

\legend{};
\end{axis}
\end{tikzpicture}
\end{subfigure}

\vspace{-1.35cm}
\begin{subfigure}[b]{.31\textwidth}
\hspace*{-0.7cm}
\begin{tikzpicture}[marks]
\begin{axis}[ymin=1e-1,ymax=1e2,
    querystyle_kfig,
    ylabel={locate thr.\\ (edit) [1/$\mu$s]},
    y label style={at={(0.0,0.5)},align=center},
    xmin=10,
    xmax=1280,
    xlabel={pattern length},
    xmode=log,
    ymode=log,
    title={\vphantom{Ig}}
]

\addplot[columba] coordinates { (40,1.65947) (80,1.50991) (160,1.36394) (320,1.75924) (640,1.5482) (1280,1.32866) };
\addlegendentry{algo=locate\_columba};
\addplot[bmove] coordinates { (40,0.604917) (80,0.537818) (160,0.490426) (320,0.564832) (640,0.449794) (1280,0.343778) };
\addlegendentry{algo=locate\_columba\_rlc};
\addplot[moverb] coordinates { (40,5.21088) (80,5.1212) (160,4.83643) (320,3.94729) (640,2.13388) (1280,0.634582) };
\addlegendentry{algo=locate\_move\_rb\_move};
\addplot[moverbrlzsa] coordinates { (40,20.1697) (80,19.14) (160,16.045) (320,8.92404) (640,3.07553) (1280,0.691719) };
\addlegendentry{algo=locate\_move\_rb\_rlzsa};

\legend{};
\end{axis}
\end{tikzpicture}
\end{subfigure}
\begin{subfigure}[b]{.31\textwidth}
\begin{tikzpicture}[marks]
\begin{axis}[ymin=1e-1,ymax=1e1,querystyle_kfig,
    ylabel={\phantom{p}},
    xmin=10,
    xmax=1280,
    xlabel={pattern length},
    xmode=log,
    ymode=log,
    title={\vphantom{Ig}}
]

\addplot[columba] coordinates { (40,0.885679) (80,0.886445) (160,0.780817) (320,1.09911) (640,0.600897) (1280,0.349545) };
\addlegendentry{algo=locate\_columba};
\addplot[bmove] coordinates { (40,0.645896) (80,0.609665) (160,0.605401) (320,0.626056) (640,0.409261) (1280,0.217572) };
\addlegendentry{algo=locate\_columba\_rlc};
\addplot[moverb] coordinates { (40,3.93365) (80,3.3696) (160,1.40968) (320,1.34144) (640,0.79442) (1280,0.404656) };
\addlegendentry{algo=locate\_move\_rb\_move};
\addplot[moverbrlzsa] coordinates { (40,11.3094) (80,7.11781) (160,1.79001) (320,1.64661) (640,0.879297) (1280,0.420379) };
\addlegendentry{algo=locate\_move\_rb\_rlzsa};

\legend{};
\end{axis}
\end{tikzpicture}
\end{subfigure}
\begin{subfigure}[b]{.31\textwidth}
\begin{tikzpicture}[marks]
\begin{axis}[ymin=1e-2,ymax=1e0,querystyle_kfig,
    ylabel={\phantom{p}},
    xmin=10,
    xmax=1280,
    xlabel={pattern length},
    xmode=log,
    ymode=log,
    title={\vphantom{Ig}}
]

\addplot[moverb] coordinates { (20,0.192962) (40,0.969372) (80,0.616502) (160,0.614152) (320,0.886057) (640,0.0956561) (1280,0.535502) };
\addlegendentry{algo=locate\_move\_rb\_move};
\addplot[moverbrlzsa] coordinates { (20,0.199614) (40,1.06017) (80,0.659404) (160,0.654607) (320,0.971912) (640,0.0948591) (1280,0.542032) };
\addlegendentry{algo=locate\_move\_rb\_rlzsa};

\legend{};
\end{axis}
\end{tikzpicture}
\end{subfigure}

\vspace*{-0.35cm}

\centering
\begin{tikzpicture}[legendmarks]
\begin{axis}[legend,legend columns=5]
\addplot[brindex] coordinates { (0,0) };
\addlegendentry{\texttt{br-index}};
\addplot[bmove] coordinates { (0,0) };
\addlegendentry{\texttt{b-move}};
\addplot[columba] coordinates { (0,0) };
\addlegendentry{\texttt{columba}};
\addplot[moverb] coordinates { (0,0) };
\addlegendentry{\texttt{move-rb}};
\addplot[moverbrlzsa] coordinates { (0,0) };
\addlegendentry{\texttt{move-rb-rlzsa}};
\end{axis}
\end{tikzpicture}

\vspace*{-0.1cm}
\caption{APM query throughput for $k = 7$. All indexes use their native APM algorithms.}
\label{fig:apm_native_k7}
\vspace*{\fill}
\end{figure*}

\begin{figure*}[p]
\centering
\vspace*{\fill}
\begin{subfigure}[b]{.31\textwidth}
\hspace*{-0.7cm}
\begin{tikzpicture}[marks]
\begin{axis}[ymin=1e-8,ymax=1e-2,
    querystyle_kfig,
    title={sars2},
    ylabel={count thr.\\ (ham.) [1/$\mu$s]},
    y label style={at={(0.0,0.5)},align=center},
    xmin=20,
    xmax=1280,
    xticklabels={},
    xmode=log,
    ymode=log,
    xlabel={\vphantom{pattern length}}
]

\addplot[brindex] coordinates { (20,3.9282e-08) (40,4.36174e-07) (80,1.72959e-05) (160,2.14994e-05) (320,1.36997e-05) (640,6.30627e-06) (1280,4.06648e-06) };
\addlegendentry{algo=count\_br\_index\_native};
\addplot[moverb] coordinates { (20,3.34099e-06) (40,0.000580258) (80,0.00355255) (160,0.00288859) (320,0.00159914) (640,0.000625298) (1280,0.000388159) };
\addlegendentry{algo=count\_move\_rb\_move};

\legend{};
\end{axis}
\end{tikzpicture}
\end{subfigure}
\begin{subfigure}[b]{.31\textwidth}
\begin{tikzpicture}[marks]
\begin{axis}[ymin=1e-7,ymax=1e0,querystyle_kfig,
    title={chr19},
    ylabel={\phantom{p}},
    xmin=40,
    xmax=1280,
    xticklabels={},
    xmode=log,
    ymode=log,
    xlabel={\vphantom{pattern length}}
]

\addplot[brindex] coordinates { (40,3.22904e-07) (80,1.87922e-05) (160,0.000297538) (320,0.00348318) (640,0.00423287) (1280,0.00263624) };
\addlegendentry{algo=count\_br\_index\_native};
\addplot[moverb] coordinates { (40,0.000279537) (80,0.00441657) (160,0.0901944) (320,0.395952) (640,0.338796) (1280,0.226748) };
\addlegendentry{algo=count\_move\_rb\_move};

\legend{};
\end{axis}
\end{tikzpicture}
\end{subfigure}
\begin{subfigure}[b]{.31\textwidth}
\begin{tikzpicture}[marks]
\begin{axis}[ymin=1e-8,ymax=1e0,querystyle_kfig,
    title={dewiki},
    ylabel={\phantom{p}},
    xmin=20,
    xmax=1280,
    xticklabels={},
    xmode=log,
    ymode=log,
    xlabel={\vphantom{pattern length}}
]

\addplot[brindex] coordinates { (20,1.12694e-08) (40,2.21203e-07) (80,6.58151e-07) (160,4.89634e-06) (320,0.000148799) (640,0.000311909) (1280,0.000320432) };
\addlegendentry{algo=count\_br\_index\_native};
\addplot[moverb] coordinates { (20,1.49871e-05) (40,0.00339733) (80,0.00166026) (160,0.0165442) (320,0.179842) (640,0.318118) (1280,0.353568) };
\addlegendentry{algo=count\_move\_rb\_move};

\legend{};
\end{axis}
\end{tikzpicture}
\end{subfigure}

\vspace{-1.35cm}
\begin{subfigure}[b]{.31\textwidth}
\hspace*{-0.7cm}
\begin{tikzpicture}[marks]
\begin{axis}[ymin=1e-3,ymax=1e2,
    querystyle_kfig,
    ylabel={locate thr.\\ (ham.) [1/$\mu$s]},
    y label style={at={(0.0,0.5)},align=center},
    xmin=20,
    xmax=1280,
    xticklabels={},
    xmode=log,
    ymode=log,
    title={\vphantom{Ig}},
    xlabel={\vphantom{pattern length}}
]

\addplot[brindex] coordinates { (40,0.0176609) (80,0.255442) (160,0.164193) (320,0.0651244) (640,0.0199513) (1280,0.0036871) };
\addlegendentry{algo=locate\_br\_index\_native};
\addplot[columba] coordinates { (40,1.86594) (80,1.97593) (160,1.93237) (320,1.7609) (640,1.32846) (1280,0.535421) };
\addlegendentry{algo=locate\_columba};
\addplot[bmove] coordinates { (40,1.05253) (80,1.07137) (160,0.96696) (320,0.72016) (640,0.374772) (1280,0.105138) };
\addlegendentry{algo=locate\_columba\_rlc};
\addplot[moverb] coordinates { (40,5.23257) (80,6.26768) (160,5.30052) (320,3.54306) (640,1.45614) (1280,0.291598) };
\addlegendentry{algo=locate\_move\_rb\_move};
\addplot[moverbrlzsa] coordinates { (40,18.1147) (80,40.1152) (160,19.311) (320,6.97607) (640,1.84622) (1280,0.305525) };
\addlegendentry{algo=locate\_move\_rb\_rlzsa};

\legend{};
\end{axis}
\end{tikzpicture}
\end{subfigure}
\begin{subfigure}[b]{.31\textwidth}
\begin{tikzpicture}[marks]
\begin{axis}[ymin=1e-3,ymax=1e1,querystyle_kfig,
    ylabel={\phantom{p}},
    xmin=40,
    xmax=1280,
    xticklabels={},
    xmode=log,
    ymode=log,
    title={\vphantom{Ig}},
    xlabel={\vphantom{pattern length}}
]

\addplot[brindex] coordinates { (40,0.0167319) (80,0.0356544) (160,0.00280786) (320,0.0129662) (640,0.00952765) (1280,0.00420386) };
\addlegendentry{algo=locate\_br\_index\_native};
\addplot[columba] coordinates { (40,0.984464) (80,1.08782) (160,0.450633) (320,0.753748) (640,0.595655) (1280,0.319833) };
\addlegendentry{algo=locate\_columba};
\addplot[bmove] coordinates { (40,0.531881) (80,0.636151) (160,0.334268) (320,0.454772) (640,0.26719) (1280,0.135645) };
\addlegendentry{algo=locate\_columba\_rlc};
\addplot[moverb] coordinates { (40,4.27792) (80,3.81889) (160,0.73656) (320,1.15449) (640,0.634654) (1280,0.306543) };
\addlegendentry{algo=locate\_move\_rb\_move};
\addplot[moverbrlzsa] coordinates { (40,12.0898) (80,9.12527) (160,0.854152) (320,1.37334) (640,0.678502) (1280,0.316464) };
\addlegendentry{algo=locate\_move\_rb\_rlzsa};

\legend{};
\end{axis}
\end{tikzpicture}
\end{subfigure}
\begin{subfigure}[b]{.31\textwidth}
\begin{tikzpicture}[marks]
\begin{axis}[ymin=1e-5,ymax=1e0,querystyle_kfig,
    ylabel={\phantom{p}},
    xmin=20,
    xmax=1280,
    xticklabels={},
    xmode=log,
    ymode=log,
    title={\vphantom{Ig}},
    xlabel={\vphantom{pattern length}}
]

\addplot[brindex] coordinates { (20,0.000339988) (40,3.48749e-05) (80,0.00187509) (160,0.000183504) (320,0.000342171) (640,0.000486447) (1280,0.000348487) };
\addlegendentry{algo=locate\_br\_index\_native};
\addplot[moverb] coordinates { (20,0.248619) (40,0.298718) (80,1.0412) (160,0.344034) (320,0.399221) (640,0.474182) (1280,0.382171) };
\addlegendentry{algo=locate\_move\_rb\_move};
\addplot[moverbrlzsa] coordinates { (20,0.256416) (40,0.302126) (80,1.30111) (160,0.354304) (320,0.409233) (640,0.478483) (1280,0.376131) };
\addlegendentry{algo=locate\_move\_rb\_rlzsa};

\legend{};
\end{axis}
\end{tikzpicture}
\end{subfigure}

\vspace{-1.35cm}
\begin{subfigure}[b]{.31\textwidth}
\hspace*{-0.7cm}
\begin{tikzpicture}[marks]
\begin{axis}[ymin=1e-1,ymax=1e2,
    querystyle_kfig,
    ylabel={locate thr.\\ (edit) [1/$\mu$s]},
    y label style={at={(0.0,0.5)},align=center},
    xmin=20,
    xmax=1280,
    xlabel={pattern length},
    xmode=log,
    ymode=log,
    title={\vphantom{Ig}}
]

\addplot[columba] coordinates { (40,1.25326) (80,1.0524) (160,1.15416) (320,1.64898) (640,1.4145) (1280,0.474975) };
\addlegendentry{algo=locate\_columba};
\addplot[bmove] coordinates { (40,0.367832) (80,0.394789) (160,0.41259) (320,0.442095) (640,0.351073) (1280,0.116522) };
\addlegendentry{algo=locate\_columba\_rlc};
\addplot[moverb] coordinates { (40,2.95917) (80,4.5783) (160,3.48106) (320,3.77155) (640,2.06697) (1280,0.409647) };
\addlegendentry{algo=locate\_move\_rb\_move};
\addplot[moverbrlzsa] coordinates { (40,5.20788) (80,14.0342) (160,7.06898) (320,8.56592) (640,2.92883) (1280,0.430563) };
\addlegendentry{algo=locate\_move\_rb\_rlzsa};

\legend{};
\end{axis}
\end{tikzpicture}
\end{subfigure}
\begin{subfigure}[b]{.31\textwidth}
\begin{tikzpicture}[marks]
\begin{axis}[ymin=1e-1,ymax=1e1,querystyle_kfig,
    ylabel={\phantom{p}},
    xmin=40,
    xmax=1280,
    xlabel={pattern length},
    xmode=log,
    ymode=log,
    title={\vphantom{Ig}}
]

\addplot[columba] coordinates { (40,0.547637) (80,0.969464) (160,0.562212) (320,0.657707) (640,0.43567) (1280,0.209867) };
\addlegendentry{algo=locate\_columba};
\addplot[bmove] coordinates { (40,0.341412) (80,0.52247) (160,0.367487) (320,0.400582) (640,0.236124) (1280,0.107207) };
\addlegendentry{algo=locate\_columba\_rlc};
\addplot[moverb] coordinates { (40,0.443023) (80,2.02159) (160,0.634703) (320,0.854874) (640,0.494984) (1280,0.231697) };
\addlegendentry{algo=locate\_move\_rb\_move};
\addplot[moverbrlzsa] coordinates { (40,0.449861) (80,3.15647) (160,0.704847) (320,0.962266) (640,0.524716) (1280,0.236339) };
\addlegendentry{algo=locate\_move\_rb\_rlzsa};

\legend{};
\end{axis}
\end{tikzpicture}
\end{subfigure}
\begin{subfigure}[b]{.31\textwidth}
\begin{tikzpicture}[marks]
\begin{axis}[ymin=1e-2,ymax=1e0,querystyle_kfig,
    ylabel={\phantom{p}},
    xmin=20,
    xmax=1280,
    xlabel={pattern length},
    xmode=log,
    ymode=log,
    title={\vphantom{Ig}}
]

\addplot[moverb] coordinates { (40,0.0199655) (80,0.702952) (160,0.268216) (320,0.124233) (640,0.3645) (1280,0.306207) };
\addlegendentry{algo=locate\_move\_rb\_move};
\addplot[moverbrlzsa] coordinates { (40,0.0200625) (80,0.742497) (160,0.273206) (320,0.125299) (640,0.367614) (1280,0.302963) };
\addlegendentry{algo=locate\_move\_rb\_rlzsa};

\legend{};
\end{axis}
\end{tikzpicture}
\end{subfigure}

\vspace*{-0.35cm}

\centering
\begin{tikzpicture}[legendmarks]
\begin{axis}[legend,legend columns=5]
\addplot[brindex] coordinates { (0,0) };
\addlegendentry{\texttt{br-index}};
\addplot[bmove] coordinates { (0,0) };
\addlegendentry{\texttt{b-move}};
\addplot[columba] coordinates { (0,0) };
\addlegendentry{\texttt{columba}};
\addplot[moverb] coordinates { (0,0) };
\addlegendentry{\texttt{move-rb}};
\addplot[moverbrlzsa] coordinates { (0,0) };
\addlegendentry{\texttt{move-rb-rlzsa}};
\end{axis}
\end{tikzpicture}

\vspace*{-0.1cm}
\caption{APM query throughput for $k = 10$. All indexes use their native APM algorithms.}
\label{fig:apm_native_k10}
\vspace{0.2cm}
\begin{subfigure}[b]{.31\textwidth}
\hspace*{-0.7cm}
\begin{tikzpicture}[marks]
\begin{axis}[ymin=1e-7,ymax=1e-2,
    querystyle_kfig,
    title={sars2},
    ylabel={count thr.\\ (ham.) [1/$\mu$s]},
    y label style={at={(0.0,0.5)},align=center},
    xmin=40,
    xmax=1280,
    xticklabels={},
    xmode=log,
    ymode=log,
    xlabel={\vphantom{pattern length}}
]

\addplot[brindex] coordinates { (40,7.11839e-08) (80,3.15957e-06) (160,1.14548e-05) (320,9.69001e-06) (640,5.12552e-06) (1280,5.23545e-06) };
\addlegendentry{algo=count\_br\_index\_native};
\addplot[moverb] coordinates { (40,0.000104822) (80,0.00246949) (160,0.0017936) (320,0.00134203) (640,0.000694971) (1280,0.0005557) };
\addlegendentry{algo=count\_move\_rb\_move};

\legend{};
\end{axis}
\end{tikzpicture}
\end{subfigure}
\begin{subfigure}[b]{.31\textwidth}
\begin{tikzpicture}[marks]
\begin{axis}[ymin=1e-7,ymax=1e0,querystyle_kfig,
    title={chr19},
    ylabel={\phantom{p}},
    xmin=40,
    xmax=1280,
    xticklabels={},
    xmode=log,
    ymode=log,
    xlabel={\vphantom{pattern length}}
]

\addplot[brindex] coordinates { (40,5.50401e-08) (80,2.42002e-06) (160,8.25112e-05) (320,0.00110058) (640,0.00309531) (1280,0.00195412) };
\addlegendentry{algo=count\_br\_index\_native};
\addplot[moverb] coordinates { (40,5.36505e-05) (80,0.00169147) (160,0.032853) (320,0.284802) (640,0.28112) (1280,0.185269) };
\addlegendentry{algo=count\_move\_rb\_move};

\legend{};
\end{axis}
\end{tikzpicture}
\end{subfigure}
\begin{subfigure}[b]{.31\textwidth}
\begin{tikzpicture}[marks]
\begin{axis}[ymin=1e-9,ymax=1e0,querystyle_kfig,
    title={dewiki},
    ylabel={\phantom{p}},
    xmin=20,
    xmax=1280,
    xticklabels={},
    xmode=log,
    ymode=log,
    xlabel={\vphantom{pattern length}}
]

\addplot[brindex] coordinates { (20,2.91553e-09) (40,5.69506e-08) (80,4.55701e-08) (160,1.02169e-06) (320,3.09946e-05) (640,0.000182517) (1280,0.000238253) };
\addlegendentry{algo=count\_br\_index\_native};
\addplot[moverb] coordinates { (20,6.11283e-06) (40,0.00102496) (80,0.000536438) (160,0.00363698) (320,0.0734331) (640,0.241534) (1280,0.300186) };
\addlegendentry{algo=count\_move\_rb\_move};

\legend{};
\end{axis}
\end{tikzpicture}
\end{subfigure}

\vspace{-1.35cm}
\begin{subfigure}[b]{.31\textwidth}
\hspace*{-0.7cm}
\begin{tikzpicture}[marks]
\begin{axis}[ymin=1e-3,ymax=1e2,
    querystyle_kfig,
    ylabel={locate thr.\\ (ham.) [1/$\mu$s]},
    y label style={at={(0.0,0.5)},align=center},
    xmin=40,
    xmax=1280,
    xticklabels={},
    xmode=log,
    ymode=log,
    title={\vphantom{Ig}},
    xlabel={\vphantom{pattern length}}
]

\addplot[brindex] coordinates { (40,0.00299716) (80,0.0635543) (160,0.120143) (320,0.042365) (640,0.00975935) (1280,0.00325846) };
\addlegendentry{algo=locate\_br\_index\_native};
\addplot[columba] coordinates { (40,1.18443) (80,1.92213) (160,1.83359) (320,1.65768) (640,0.992298) (1280,0.485613) };
\addlegendentry{algo=locate\_columba};
\addplot[bmove] coordinates { (40,0.641843) (80,0.999995) (160,0.840671) (320,0.671913) (640,0.255542) (1280,0.0962243) };
\addlegendentry{algo=locate\_columba\_rlc};
\addplot[moverb] coordinates { (40,2.24112) (80,5.91218) (160,4.99299) (320,3.02568) (640,0.8341) (1280,0.286296) };
\addlegendentry{algo=locate\_move\_rb\_move};
\addplot[moverbrlzsa] coordinates { (40,3.2531) (80,30.0737) (160,15.5113) (320,5.15145) (640,0.957103) (1280,0.302932) };
\addlegendentry{algo=locate\_move\_rb\_rlzsa};

\legend{};
\end{axis}
\end{tikzpicture}
\end{subfigure}
\begin{subfigure}[b]{.31\textwidth}
\begin{tikzpicture}[marks]
\begin{axis}[ymin=1e-3,ymax=1e1,querystyle_kfig,
    ylabel={\phantom{p}},
    xmin=40,
    xmax=1280,
    xticklabels={},
    xmode=log,
    ymode=log,
    title={\vphantom{Ig}},
    xlabel={\vphantom{pattern length}}
]

\addplot[brindex] coordinates { (40,0.00535056) (80,0.0163835) (160,0.00151221) (320,0.00434519) (640,0.00680464) (1280,0.00302364) };
\addlegendentry{algo=locate\_br\_index\_native};
\addplot[columba] coordinates { (40,0.602977) (80,0.937427) (160,0.209246) (320,0.55656) (640,0.414764) (1280,0.280604) };
\addlegendentry{algo=locate\_columba};
\addplot[bmove] coordinates { (40,0.284646) (80,0.459505) (160,0.216879) (320,0.354408) (640,0.218375) (1280,0.10655) };
\addlegendentry{algo=locate\_columba\_rlc};
\addplot[moverb] coordinates { (40,2.60885) (80,3.68812) (160,0.500883) (320,0.897199) (640,0.513959) (1280,0.244964) };
\addlegendentry{algo=locate\_move\_rb\_move};
\addplot[moverbrlzsa] coordinates { (40,4.2958) (80,8.16438) (160,0.550042) (320,1.02767) (640,0.548954) (1280,0.248364) };
\addlegendentry{algo=locate\_move\_rb\_rlzsa};

\legend{};
\end{axis}
\end{tikzpicture}
\end{subfigure}
\begin{subfigure}[b]{.31\textwidth}
\begin{tikzpicture}[marks]
\begin{axis}[ymin=1e-5,ymax=1e1,querystyle_kfig,
    ylabel={\phantom{p}},
    xmin=20,
    xmax=1280,
    xticklabels={},
    xmode=log,
    ymode=log,
    title={\vphantom{Ig}},
    xlabel={\vphantom{pattern length}}
]

\addplot[brindex] coordinates { (20,0.0036665) (40,4.03737e-05) (80,3.07934e-05) (160,5.97933e-06) (320,0.000114506) (640,0.000348068) (1280,0.000265751) };
\addlegendentry{algo=locate\_br\_index\_native};
\addplot[moverb] coordinates { (20,1.97335) (40,0.532925) (80,0.143634) (160,0.0239243) (320,0.235141) (640,0.362579) (1280,0.320548) };
\addlegendentry{algo=locate\_move\_rb\_move};
\addplot[moverbrlzsa] coordinates { (20,2.59954) (40,0.546147) (80,0.143132) (160,0.0240665) (320,0.238605) (640,0.368924) (1280,0.320259) };
\addlegendentry{algo=locate\_move\_rb\_rlzsa};

\legend{};
\end{axis}
\end{tikzpicture}
\end{subfigure}

\vspace{-1.35cm}
\begin{subfigure}[b]{.31\textwidth}
\hspace*{-0.7cm}
\begin{tikzpicture}[marks]
\begin{axis}[ymin=1e-1,ymax=1e1,
    querystyle_kfig,
    ylabel={locate thr.\\ (edit) [1/$\mu$s]},
    y label style={at={(0.0,0.5)},align=center},
    xmin=40,
    xmax=1280,
    xlabel={pattern length},
    xmode=log,
    ymode=log,
    title={\vphantom{Ig}}
]

\addplot[columba] coordinates { (80,1.23529) (160,0.983727) (320,1.21824) (640,0.958616) };
\addlegendentry{algo=locate\_columba};
\addplot[bmove] coordinates { (80,0.443956) (160,0.27362) (320,0.347743) (640,0.301365) };
\addlegendentry{algo=locate\_columba\_rlc};
\addplot[moverb] coordinates { (80,3.97809) (160,3.13556) (320,2.1594) (640,1.28196) };
\addlegendentry{algo=locate\_move\_rb\_move};
\addplot[moverbrlzsa] coordinates { (80,9.35288) (160,6.36765) (320,3.09214) (640,1.54711) };
\addlegendentry{algo=locate\_move\_rb\_rlzsa};

\legend{};
\end{axis}
\end{tikzpicture}
\end{subfigure}
\begin{subfigure}[b]{.31\textwidth}
\begin{tikzpicture}[marks]
\begin{axis}[ymin=1e-1,ymax=1e1,querystyle_kfig,
    ylabel={\phantom{p}},
    xmin=40,
    xmax=1280,
    xlabel={pattern length},
    xmode=log,
    ymode=log,
    title={\vphantom{Ig}}
]

\addplot[columba] coordinates { (80,0.542944) (160,0.284218) (320,0.509807) (640,0.323458) (1280,0.18287) };
\addlegendentry{algo=locate\_columba};
\addplot[bmove] coordinates { (80,0.221362) (160,0.208435) (320,0.309792) (640,0.184588) (1280,0.0864012) };
\addlegendentry{algo=locate\_columba\_rlc};
\addplot[moverb] coordinates { (80,1.1774) (160,0.257234) (320,0.571404) (640,0.337823) (1280,0.16175) };
\addlegendentry{algo=locate\_move\_rb\_move};
\addplot[moverbrlzsa] coordinates { (80,1.44939) (160,0.268523) (320,0.62225) (640,0.356007) (1280,0.164789) };
\addlegendentry{algo=locate\_move\_rb\_rlzsa};

\legend{};
\end{axis}
\end{tikzpicture}
\end{subfigure}
\begin{subfigure}[b]{.31\textwidth}
\begin{tikzpicture}[marks]
\begin{axis}[ymin=1e-3,ymax=1e0,querystyle_kfig,
    ylabel={\phantom{p}},
    xmin=20,
    xmax=1280,
    xlabel={pattern length},
    xmode=log,
    ymode=log,
    title={\vphantom{Ig}}
]

\addplot[moverb] coordinates { (40,0.00912525) (80,0.00789297) (160,0.0324623) (320,0.0488842) (640,0.279134) (1280,0.237957) };
\addlegendentry{algo=locate\_move\_rb\_move};
\addplot[moverbrlzsa] coordinates { (40,0.00906354) (80,0.007874) (160,0.0324539) (320,0.04862) (640,0.284963) (1280,0.240371) };
\addlegendentry{algo=locate\_move\_rb\_rlzsa};

\legend{};
\end{axis}
\end{tikzpicture}
\end{subfigure}

\vspace*{-0.35cm}

\centering
\begin{tikzpicture}[legendmarks]
\begin{axis}[legend,legend columns=5]
\addplot[brindex] coordinates { (0,0) };
\addlegendentry{\texttt{br-index}};
\addplot[bmove] coordinates { (0,0) };
\addlegendentry{\texttt{b-move}};
\addplot[columba] coordinates { (0,0) };
\addlegendentry{\texttt{columba}};
\addplot[moverb] coordinates { (0,0) };
\addlegendentry{\texttt{move-rb}};
\addplot[moverbrlzsa] coordinates { (0,0) };
\addlegendentry{\texttt{move-rb-rlzsa}};
\end{axis}
\end{tikzpicture}

\vspace*{-0.1cm}
\caption{APM query throughput for $k = 13$. All indexes use their native APM algorithms.}
\label{fig:apm_native_k13}
\vspace*{\fill}
\end{figure*}

\begin{figure*}[p]
\centering
\vspace*{\fill}
\input{charts/apm_samealg_k4}
\caption{APM query throughput for $k = 4$. All indexes use our optimized APM algorithms.}
\label{fig:apm_same_alg_k4}
\vspace{0.2cm}
\input{charts/apm_samealg_k7}
\caption{APM query throughput for $k = 7$. All indexes use our optimized APM algorithms.}
\label{fig:apm_same_alg_k7}
\vspace*{\fill}
\end{figure*}

\begin{figure*}[p]
\centering
\vspace*{\fill}
\input{charts/apm_samealg_k10}
\caption{APM query throughput for $k = 10$. All indexes use our optimized APM algorithms.}
\label{fig:apm_same_alg_k10}
\vspace{0.2cm}
\begin{subfigure}[b]{.31\textwidth}
\hspace*{-0.7cm}
\begin{tikzpicture}[marks]
\begin{axis}[ymin=1e-6,ymax=1e-2,
    querystyle_kfig,
    title={sars2},
    ylabel={count thr.\\ (ham.) [1/$\mu$s]},
    y label style={at={(0.0,0.5)},align=center},
    xmin=40,
    xmax=1280,
    xticklabels={},
    xmode=log,
    ymode=log,
    xlabel={\vphantom{pattern length}}
]

\addplot[brindex] coordinates { (40,2.8649e-06) (80,4.50916e-05) (160,3.02442e-05) (320,2.20737e-05) (640,1.17833e-05) (1280,9.675e-06) };
\addlegendentry{algo=count\_br\_index};
\addplot[columba] coordinates { (40,5.498e-05) (80,0.00124726) (160,0.000827679) (320,0.000555198) (640,0.000247514) (1280,0.000180639) };
\addlegendentry{algo=count\_columba\_apm};
\addplot[bmove] coordinates { (40,2.16301e-05) (80,0.000311304) (160,0.000193998) (320,0.000132871) (640,6.88235e-05) (1280,5.21444e-05) };
\addlegendentry{algo=count\_columba\_rlc\_apm};
\addplot[moverb] coordinates { (40,0.000104822) (80,0.00246949) (160,0.0017936) (320,0.00134203) (640,0.000694971) (1280,0.0005557) };
\addlegendentry{algo=count\_move\_rb\_move};

\legend{};
\end{axis}
\end{tikzpicture}
\end{subfigure}
\begin{subfigure}[b]{.31\textwidth}
\begin{tikzpicture}[marks]
\begin{axis}[ymin=1e-6,ymax=1e0,querystyle_kfig,
    title={chr19},
    ylabel={\phantom{p}},
    xmin=40,
    xmax=1280,
    xticklabels={},
    xmode=log,
    ymode=log,
    xlabel={\vphantom{pattern length}}
]

\addplot[brindex] coordinates { (40,2.22676e-06) (80,4.74767e-05) (160,0.00104566) (320,0.00718457) (640,0.00546098) (1280,0.00335489) };
\addlegendentry{algo=count\_br\_index};
\addplot[columba] coordinates { (40,1.00167e-05) (80,0.000375327) (160,0.0080519) (320,0.0800403) (640,0.0838162) (1280,0.056429) };
\addlegendentry{algo=count\_columba\_apm};
\addplot[bmove] coordinates { (40,1.13694e-05) (80,0.000228279) (160,0.00509473) (320,0.035883) (640,0.0279944) (1280,0.0174085) };
\addlegendentry{algo=count\_columba\_rlc\_apm};
\addplot[moverb] coordinates { (40,5.36505e-05) (80,0.00169147) (160,0.032853) (320,0.284802) (640,0.28112) (1280,0.185269) };
\addlegendentry{algo=count\_move\_rb\_move};

\legend{};
\end{axis}
\end{tikzpicture}
\end{subfigure}
\begin{subfigure}[b]{.31\textwidth}
\begin{tikzpicture}[marks]
\begin{axis}[ymin=1e-8,ymax=1e0,querystyle_kfig,
    title={dewiki},
    ylabel={\phantom{p}},
    xmin=20,
    xmax=1280,
    xticklabels={},
    xmode=log,
    ymode=log,
    xlabel={\vphantom{pattern length}}
]

\addplot[brindex] coordinates { (20,1.15634e-08) (40,2.29805e-06) (80,7.62509e-07) (160,5.91002e-06) (320,0.000119523) (640,0.000367578) (1280,0.000424375) };
\addlegendentry{algo=count\_br\_index};
\addplot[moverb] coordinates { (20,6.11283e-06) (40,0.00102496) (80,0.000536438) (160,0.00363698) (320,0.0734331) (640,0.241534) (1280,0.300186) };
\addlegendentry{algo=count\_move\_rb\_move};

\legend{};
\end{axis}
\end{tikzpicture}
\end{subfigure}

\vspace{-1.35cm}
\begin{subfigure}[b]{.31\textwidth}
\hspace*{-0.7cm}
\begin{tikzpicture}[marks]
\begin{axis}[ymin=1e-3,ymax=1e2,
    querystyle_kfig,
    ylabel={locate thr.\\ (ham.) [1/$\mu$s]},
    y label style={at={(0.0,0.5)},align=center},
    xmin=40,
    xmax=1280,
    xticklabels={},
    xmode=log,
    ymode=log,
    title={\vphantom{Ig}},
    xlabel={\vphantom{pattern length}}
]

\addplot[brindex] coordinates { (40,0.0871145) (80,0.44529) (160,0.262238) (320,0.0978649) (640,0.019024) (1280,0.00617744) };
\addlegendentry{algo=locate\_br\_index};
\addplot[columba] coordinates { (40,2.17425) (80,10.0455) (160,6.67147) (320,2.6035) (640,0.472257) (1280,0.138657) };
\addlegendentry{algo=locate\_columba\_apm};
\addplot[bmove] coordinates { (40,0.534938) (80,1.5692) (160,1.12736) (320,0.511073) (640,0.100196) (1280,0.0323557) };
\addlegendentry{algo=locate\_columba\_rlc\_apm};
\addplot[moverb] coordinates { (40,2.24112) (80,5.91218) (160,4.99299) (320,3.02568) (640,0.8341) (1280,0.286296) };
\addlegendentry{algo=locate\_move\_rb\_move};
\addplot[moverbrlzsa] coordinates { (40,3.2531) (80,30.0737) (160,15.5113) (320,5.15145) (640,0.957103) (1280,0.302932) };
\addlegendentry{algo=locate\_move\_rb\_rlzsa};

\legend{};
\end{axis}
\end{tikzpicture}
\end{subfigure}
\begin{subfigure}[b]{.31\textwidth}
\begin{tikzpicture}[marks]
\begin{axis}[ymin=1e-3,ymax=1e1,querystyle_kfig,
    ylabel={\phantom{p}},
    xmin=40,
    xmax=1280,
    xticklabels={},
    xmode=log,
    ymode=log,
    title={\vphantom{Ig}},
    xlabel={\vphantom{pattern length}}
]

\addplot[brindex] coordinates { (40,0.160136) (80,0.206964) (160,0.0186088) (320,0.0278366) (640,0.0118718) (1280,0.00515454) };
\addlegendentry{algo=locate\_br\_index};
\addplot[columba] coordinates { (40,1.29437) (80,2.36322) (160,0.165677) (320,0.390945) (640,0.220656) (1280,0.105563) };
\addlegendentry{algo=locate\_columba\_apm};
\addplot[bmove] coordinates { (40,0.676451) (80,0.797907) (160,0.0899833) (320,0.138417) (640,0.0617602) (1280,0.0273111) };
\addlegendentry{algo=locate\_columba\_rlc\_apm};
\addplot[moverb] coordinates { (40,2.60885) (80,3.68812) (160,0.500883) (320,0.897199) (640,0.513959) (1280,0.244964) };
\addlegendentry{algo=locate\_move\_rb\_move};
\addplot[moverbrlzsa] coordinates { (40,4.2958) (80,8.16438) (160,0.550042) (320,1.02767) (640,0.548954) (1280,0.248364) };
\addlegendentry{algo=locate\_move\_rb\_rlzsa};

\legend{};
\end{axis}
\end{tikzpicture}
\end{subfigure}
\begin{subfigure}[b]{.31\textwidth}
\begin{tikzpicture}[marks]
\begin{axis}[ymin=1e-5,ymax=1e1,querystyle_kfig,
    ylabel={\phantom{p}},
    xmin=20,
    xmax=1280,
    xticklabels={},
    xmode=log,
    ymode=log,
    title={\vphantom{Ig}},
    xlabel={\vphantom{pattern length}}
]

\addplot[brindex] coordinates { (20,0.00615525) (40,0.00119078) (80,0.000289702) (160,3.09887e-05) (320,0.000440495) (640,0.00064189) (1280,0.000477753) };
\addlegendentry{algo=locate\_br\_index};
\addplot[moverb] coordinates { (20,1.97335) (40,0.532925) (80,0.143634) (160,0.0239243) (320,0.235141) (640,0.362579) (1280,0.320548) };
\addlegendentry{algo=locate\_move\_rb\_move};
\addplot[moverbrlzsa] coordinates { (20,2.59954) (40,0.546147) (80,0.143132) (160,0.0240665) (320,0.238605) (640,0.368924) (1280,0.320259) };
\addlegendentry{algo=locate\_move\_rb\_rlzsa};

\legend{};
\end{axis}
\end{tikzpicture}
\end{subfigure}

\vspace{-1.35cm}
\begin{subfigure}[b]{.31\textwidth}
\hspace*{-0.7cm}
\begin{tikzpicture}[marks]
\begin{axis}[ymin=1e-2,ymax=1e1,
    querystyle_kfig,
    ylabel={locate thr.\\ (edit) [1/$\mu$s]},
    y label style={at={(0.0,0.5)},align=center},
    xmin=40,
    xmax=1280,
    xlabel={pattern length},
    xmode=log,
    ymode=log,
    title={\vphantom{Ig}}
]

\addplot[brindex] coordinates { (80,0.363273) (160,0.221747) (320,0.10071) (640,0.0476403) };
\addlegendentry{algo=locate\_br\_index};
\addplot[columba] coordinates { (80,3.97225) (160,2.86509) (320,1.6217) (640,0.823549) };
\addlegendentry{algo=locate\_columba\_apm};
\addplot[bmove] coordinates { (80,1.15025) (160,0.82008) (320,0.444382) (640,0.219237) };
\addlegendentry{algo=locate\_columba\_rlc\_apm};
\addplot[moverb] coordinates { (80,3.97809) (160,3.13556) (320,2.1594) (640,1.28196) };
\addlegendentry{algo=locate\_move\_rb\_move};
\addplot[moverbrlzsa] coordinates { (80,9.35288) (160,6.36765) (320,3.09214) (640,1.54711) };
\addlegendentry{algo=locate\_move\_rb\_rlzsa};

\legend{};
\end{axis}
\end{tikzpicture}
\end{subfigure}
\begin{subfigure}[b]{.31\textwidth}
\begin{tikzpicture}[marks]
\begin{axis}[ymin=1e-3,ymax=1e1,querystyle_kfig,
    ylabel={\phantom{p}},
    xmin=40,
    xmax=1280,
    xlabel={pattern length},
    xmode=log,
    ymode=log,
    title={\vphantom{Ig}}
]

\addplot[brindex] coordinates { (80,0.0741463) (160,0.0121173) (320,0.0233091) (640,0.011104) (1280,0.00483814) };
\addlegendentry{algo=locate\_br\_index};
\addplot[columba] coordinates { (80,0.723083) (160,0.116811) (320,0.294642) (640,0.161061) (1280,0.0771256) };
\addlegendentry{algo=locate\_columba\_apm};
\addplot[bmove] coordinates { (80,0.256748) (160,0.048474) (320,0.0998988) (640,0.0510177) (1280,0.0231053) };
\addlegendentry{algo=locate\_columba\_rlc\_apm};
\addplot[moverb] coordinates { (80,1.1774) (160,0.257234) (320,0.571404) (640,0.337823) (1280,0.16175) };
\addlegendentry{algo=locate\_move\_rb\_move};
\addplot[moverbrlzsa] coordinates { (80,1.44939) (160,0.268523) (320,0.62225) (640,0.356007) (1280,0.164789) };
\addlegendentry{algo=locate\_move\_rb\_rlzsa};

\legend{};
\end{axis}
\end{tikzpicture}
\end{subfigure}
\begin{subfigure}[b]{.31\textwidth}
\begin{tikzpicture}[marks]
\begin{axis}[ymin=1e-5,ymax=1e0,querystyle_kfig,
    ylabel={\phantom{p}},
    xmin=20,
    xmax=1280,
    xlabel={pattern length},
    xmode=log,
    ymode=log,
    title={\vphantom{Ig}}
]

\addplot[brindex] coordinates { (40,2.05423e-05) (80,1.93334e-05) (160,4.80932e-05) (320,7.87265e-05) (640,0.000485353) (1280,0.000380516) };
\addlegendentry{algo=locate\_br\_index};
\addplot[moverb] coordinates { (40,0.00912525) (80,0.00789297) (160,0.0324623) (320,0.0488842) (640,0.279134) (1280,0.237957) };
\addlegendentry{algo=locate\_move\_rb\_move};
\addplot[moverbrlzsa] coordinates { (40,0.00906354) (80,0.007874) (160,0.0324539) (320,0.04862) (640,0.284963) (1280,0.240371) };
\addlegendentry{algo=locate\_move\_rb\_rlzsa};

\legend{};
\end{axis}
\end{tikzpicture}
\end{subfigure}

\vspace*{-0.35cm}

\centering
\begin{tikzpicture}[legendmarks]
\begin{axis}[legend,legend columns=5]
\addplot[brindex] coordinates { (0,0) };
\addlegendentry{\texttt{br-index}};
\addplot[bmove] coordinates { (0,0) };
\addlegendentry{\texttt{b-move}};
\addplot[columba] coordinates { (0,0) };
\addlegendentry{\texttt{columba}};
\addplot[moverb] coordinates { (0,0) };
\addlegendentry{\texttt{move-rb}};
\addplot[moverbrlzsa] coordinates { (0,0) };
\addlegendentry{\texttt{move-rb-rlzsa}};
\end{axis}
\end{tikzpicture}

\vspace*{-0.1cm}
\caption{APM query throughput for $k = 13$. All indexes use our optimized APM algorithms.}
\label{fig:apm_same_alg_k13}
\vspace*{\fill}
\end{figure*}

\setcounter{section}{5}
\setcounter{table}{0}

\begin{table*}[p]
\centering

\vspace*{\fill}
\begin{minipage}[c]{0.49\textwidth}
\centering
\scriptsize
\setlength{\tabcolsep}{3pt}
\renewcommand{\arraystretch}{1.18}
\begin{tabular}{|r|rr|rr|rr|}
\hline
\multicolumn{1}{|c|}{text}
& \multicolumn{2}{c|}{sars2}
& \multicolumn{2}{c|}{chr19}
& \multicolumn{2}{c|}{dewiki} \\
$m$
& $N$ & $\floor{occ/N}$
& $N$ & $\floor{occ/N}$
& $N$ & $\floor{occ/N}$ \\
\hline
  10 & 124733 & 1765895 & 137833 & 1199869 & 28012 & 253393 \\
  20 &  48320 & 1599931 &  71840 &  100930 & 23252 &  18400 \\
  40 &  23838 & 1528209 &  47583 &    2769 & 19732 &   4078 \\
  80 &  12302 & 1408223 &  31893 &     886 & 15685 &   2702 \\
 160 &   6349 & 1221234 &  19652 &     793 & 11377 &   1394 \\
 320 &   3353 &  935224 &  11290 &     682 &  7516 &   1503 \\
 640 &   1819 &  598603 &   6102 &     542 &  4554 &    615 \\
1280 &   1030 &  270775 &   3199 &     393 &  2608 &    206 \\
\hline
\end{tabular}
\captionof{table}{Pattern sets used for extension measurements.}
\label{tab:patterns_count}

\vspace{0.5cm}
\begin{tabular}{|r|rr|rr|rr|}
\hline
\multicolumn{1}{|c|}{text}
& \multicolumn{2}{c|}{sars2}
& \multicolumn{2}{c|}{chr19}
& \multicolumn{2}{c|}{dewiki} \\
$m$
& $N$ & $\floor{occ/N}$
& $N$ & $\floor{occ/N}$
& $N$ & $\floor{occ/N}$ \\
\hline
  10 & 137 & 1770167 &   240 & 1113687 &   695 & 363573 \\
  20 & 152 & 1599232 &  2449 &  102251 &  8909 &  18249 \\
  40 & 160 & 1529478 & 25002 &    2719 & 13446 &   4594 \\
  80 & 171 & 1398672 & 23277 &     886 & 12176 &   2664 \\
 160 & 189 & 1235384 & 15409 &     794 &  9510 &   1526 \\
 320 & 222 &  997236 &  9051 &     683 &  6312 &   1695 \\
 640 & 310 &  612260 &  4978 &     546 &  3965 &    660 \\
1280 & 440 &  260110 &  2592 &     394 &  2246 &    211 \\
\hline
\end{tabular}
\captionof{table}{Pattern sets used for $\SA$-interval enumeration measurements.}
\label{tab:patterns_locate}
\end{minipage}
\hfill
\begin{minipage}[c]{0.49\textwidth}
\centering
\fontsize{6.5pt}{7.6pt}\selectfont
\setlength{\tabcolsep}{2pt}
\renewcommand{\arraystretch}{1.18}
\begin{tabular}{|rr|rr|rr|rr|}
\hline
\multicolumn{2}{|c|}{text}
& \multicolumn{2}{c|}{sars2}
& \multicolumn{2}{c|}{chr19}
& \multicolumn{2}{c|}{dewiki} \\
$k$ & $m$
& $N$ & $\floor{occ/N}$
& $N$ & $\floor{occ/N}$
& $N$ & $\floor{occ/N}$ \\
\hline
\num{4} &   \num{10} &  \num{216} & \num{1355134869} &   \num{199} & \num{1188337257} &   \num{36} & \num{7964645} \\
\num{4} &   \num{20} &  \num{845} &    \num{1750788} &   \num{490} &    \num{3343201} &  \num{351} &  \num{104757} \\
\num{4} &   \num{40} & \num{1868} &    \num{1655330} &  \num{2988} &     \num{327374} &  \num{900} &    \num{8298} \\
\num{4} &   \num{80} &  \num{800} &    \num{1617623} & \num{10923} &       \num{4227} & \num{1392} &    \num{1884} \\
\num{4} &  \num{160} &  \num{283} &    \num{1608933} & \num{20639} &        \num{983} & \num{1954} &     \num{802} \\
\num{4} &  \num{320} &   \num{86} &    \num{1395685} & \num{11079} &        \num{885} & \num{1415} &     \num{527} \\
\num{4} &  \num{640} &   \num{23} &    \num{1277406} &  \num{4549} &        \num{834} &  \num{843} &     \num{355} \\
\num{4} & \num{1280} &    \num{6} &    \num{1062504} &  \num{1778} &        \num{711} &  \num{464} &     \num{221} \\
\hline
\num{7} &   \num{10} &   \num{9} & \num{26027797178} &    \num{9} & \num{23643396101} &   \num{2} & \num{1716741013} \\
\num{7} &   \num{20} &  \num{22} &    \num{28989809} &   \num{11} &    \num{36580853} &  \num{28} &     \num{627181} \\
\num{7} &   \num{40} & \num{678} &     \num{1662128} &  \num{413} &     \num{1318297} & \num{246} &      \num{41841} \\
\num{7} &   \num{80} & \num{395} &     \num{1640185} & \num{1655} &       \num{33419} & \num{261} &       \num{1700} \\
\num{7} &  \num{160} & \num{146} &     \num{1534497} & \num{9506} &        \num{1162} & \num{650} &        \num{807} \\
\num{7} &  \num{320} &  \num{41} &     \num{1505817} & \num{8098} &         \num{889} & \num{999} &        \num{544} \\
\num{7} &  \num{640} &   \num{7} &     \num{1130424} & \num{3244} &         \num{838} & \num{624} &        \num{396} \\
\num{7} & \num{1280} &   \num{4} &      \num{614985} & \num{1137} &         \num{735} & \num{351} &        \num{214} \\
\hline
\num{10} &   \num{20} &   \num{4} & \num{799334386} & -- & -- &   \num{2} & \num{1302498} \\
\num{10} &   \num{40} &  \num{70} &   \num{1684849} &   \num{34} & \num{4464461} &  \num{84} &   \num{15850} \\
\num{10} &   \num{80} & \num{216} &   \num{1649381} &  \num{273} &  \num{202272} &  \num{25} &    \num{6015} \\
\num{10} &  \num{160} &  \num{89} &   \num{1547560} & \num{2711} &    \num{1657} & \num{167} &    \num{1091} \\
\num{10} &  \num{320} &  \num{25} &   \num{1566325} & \num{5936} &     \num{943} & \num{549} &     \num{599} \\
\num{10} &  \num{640} &   \num{5} &   \num{1490785} & \num{2588} &     \num{848} & \num{485} &     \num{461} \\
\num{10} & \num{1280} &   \num{4} &   \num{1231271} &  \num{870} &     \num{752} & \num{273} &     \num{237} \\
\hline
\num{13} &   \num{20} & -- & -- & -- & -- &   \num{2} & \num{31309391} \\
\num{13} &   \num{40} &  \num{13} & \num{1685845} &    \num{7} & \num{4619983} &  \num{25} &     \num{7555} \\
\num{13} &   \num{80} & \num{151} & \num{1645531} &  \num{105} &  \num{431466} &  \num{45} &     \num{1469} \\
\num{13} &  \num{160} &  \num{55} & \num{1517316} & \num{1026} &    \num{3472} &  \num{83} &      \num{755} \\
\num{13} &  \num{320} &  \num{21} & \num{1477573} & \num{4274} &     \num{969} & \num{439} &      \num{527} \\
\num{13} &  \num{640} &   \num{6} & \num{1123596} & \num{2156} &     \num{845} & \num{366} &      \num{336} \\
\num{13} & \num{1280} &   \num{4} &  \num{746125} &  \num{715} &     \num{756} & \num{226} &      \num{217} \\
\hline
\end{tabular}
\captionof{table}{APM count-query pattern sets; en dashes mark omitted measurements.}
\label{tab:patterns_count_apm}
\end{minipage}

\vspace*{\fill}
\begin{minipage}[c]{0.49\textwidth}
\centering
\fontsize{6.5pt}{7.6pt}\selectfont
\setlength{\tabcolsep}{2pt}
\renewcommand{\arraystretch}{1.18}
\begin{tabular}{|rr|rr|rr|rr|}
\hline
\multicolumn{2}{|c|}{text}
& \multicolumn{2}{c|}{sars2}
& \multicolumn{2}{c|}{chr19}
& \multicolumn{2}{c|}{dewiki} \\
$k$ & $m$
& $N$ & $\floor{occ/N}$
& $N$ & $\floor{occ/N}$
& $N$ & $\floor{occ/N}$ \\
\hline
\num{4} &   \num{10} & -- & -- & -- & -- &    \num{6} & \num{3194011} \\
\num{4} &   \num{20} & \num{62} & \num{1688831} &    \num{29} & \num{3864776} &  \num{147} &  \num{104548} \\
\num{4} &   \num{40} & \num{66} & \num{1675259} &   \num{253} &  \num{389521} &  \num{987} &    \num{6423} \\
\num{4} &   \num{80} & \num{62} & \num{1653685} &  \num{8537} &    \num{4316} & \num{1646} &    \num{1130} \\
\num{4} &  \num{160} & \num{55} & \num{1575676} & \num{17447} &     \num{973} & \num{1562} &     \num{787} \\
\num{4} &  \num{320} & \num{34} & \num{1543124} &  \num{9096} &     \num{887} & \num{1345} &     \num{512} \\
\num{4} &  \num{640} & \num{13} & \num{1303815} &  \num{3732} &     \num{815} &  \num{802} &     \num{331} \\
\num{4} & \num{1280} &  \num{4} & \num{1430145} &  \num{1446} &     \num{726} &  \num{433} &    \num{1008} \\
\hline
\num{7} &   \num{20} &  \num{4} & \num{22031076} &    \num{5} & \num{14591903} &  \num{37} & \num{97795} \\
\num{7} &   \num{40} & \num{62} &  \num{1649700} &   \num{56} &  \num{1637275} & \num{211} & \num{23327} \\
\num{7} &   \num{80} & \num{57} &  \num{1649584} & \num{1075} &    \num{31330} & \num{273} &  \num{3701} \\
\num{7} &  \num{160} & \num{45} &  \num{1552192} & \num{8422} &     \num{1133} & \num{675} &  \num{1084} \\
\num{7} &  \num{320} & \num{25} &  \num{1483367} & \num{6662} &      \num{893} & \num{853} &   \num{706} \\
\num{7} &  \num{640} &  \num{7} &  \num{1216400} & \num{2636} &      \num{835} & \num{580} &   \num{344} \\
\num{7} & \num{1280} &  \num{4} &   \num{387847} &  \num{909} &      \num{767} & \num{323} &   \num{213} \\
\hline
\num{10} &   \num{20} & -- & -- & -- & -- &   \num{2} & \num{1023488} \\
\num{10} &   \num{40} & \num{33} & \num{1667677} &   \num{20} & \num{2155878} &  \num{53} &   \num{10108} \\
\num{10} &   \num{80} & \num{51} & \num{1631092} &  \num{230} &  \num{155522} & \num{107} &   \num{11154} \\
\num{10} &  \num{160} & \num{38} & \num{1592654} & \num{2623} &    \num{1347} & \num{362} &     \num{819} \\
\num{10} &  \num{320} & \num{18} & \num{1615341} & \num{5016} &     \num{923} & \num{478} &     \num{508} \\
\num{10} &  \num{640} &  \num{6} & \num{1500072} & \num{2136} &     \num{823} & \num{476} &     \num{344} \\
\num{10} & \num{1280} &  \num{4} & \num{1396514} &  \num{697} &     \num{800} & \num{251} &     \num{212} \\
\hline
\num{13} &   \num{20} & -- & -- & -- & -- &   \num{2} & \num{16649636} \\
\num{13} &   \num{40} &  \num{9} & \num{1687772} &    \num{5} & \num{4033664} &  \num{22} &    \num{25691} \\
\num{13} &   \num{80} & \num{44} & \num{1669396} &   \num{57} &  \num{548352} &  \num{78} &     \num{1747} \\
\num{13} &  \num{160} & \num{30} & \num{1635451} &  \num{952} &    \num{2547} &  \num{78} &      \num{798} \\
\num{13} &  \num{320} & \num{16} & \num{1397344} & \num{3630} &     \num{985} & \num{284} &      \num{504} \\
\num{13} &  \num{640} &  \num{6} & \num{1077729} & \num{1713} &     \num{827} & \num{364} &      \num{354} \\
\num{13} & \num{1280} &  \num{4} & \num{1152011} &  \num{562} &     \num{762} & \num{205} &      \num{236} \\
\hline
\end{tabular}
\captionof{table}{APM Hamming-distance locate-query pattern sets; en dashes mark omitted measurements.}
\label{tab:patterns_locate_apm_ham}
\end{minipage}
\hfill
\begin{minipage}[c]{0.49\textwidth}
\centering
\fontsize{6.5pt}{7.6pt}\selectfont
\setlength{\tabcolsep}{2pt}
\renewcommand{\arraystretch}{1.18}
\begin{tabular}{|rr|rr|rr|rr|}
\hline
\multicolumn{2}{|c|}{text}
& \multicolumn{2}{c|}{sars2}
& \multicolumn{2}{c|}{chr19}
& \multicolumn{2}{c|}{dewiki} \\
$k$ & $m$
& $N$ & $\floor{occ/N}$
& $N$ & $\floor{occ/N}$
& $N$ & $\floor{occ/N}$ \\
\hline
\num{4} &   \num{10} & -- & -- & -- & -- &    \num{2} & \num{9667806} \\
\num{4} &   \num{20} & \num{7} & \num{8377239} &    \num{15} & \num{5822686} &   \num{41} &  \num{178538} \\
\num{4} &   \num{40} & \num{8} & \num{8142132} &    \num{92} &  \num{685916} &  \num{423} &    \num{5883} \\
\num{4} &   \num{80} & \num{8} & \num{8173130} &  \num{4671} &    \num{4978} &  \num{582} &    \num{2763} \\
\num{4} &  \num{160} & \num{8} & \num{7548369} & \num{13520} &    \num{1027} & \num{1264} &     \num{902} \\
\num{4} &  \num{320} & \num{7} & \num{6992365} &  \num{7921} &     \num{953} & \num{1243} &     \num{645} \\
\num{4} &  \num{640} & \num{4} & \num{6170921} &  \num{3469} &     \num{895} &  \num{782} &     \num{447} \\
\num{4} & \num{1280} & \num{4} & \num{1735373} &  \num{1349} &     \num{803} &  \num{437} &     \num{285} \\
\hline
\num{7} &   \num{20} & -- & -- & -- & -- &   \num{2} & \num{360879} \\
\num{7} &   \num{40} & \num{5} & \num{12316892} &    \num{4} & \num{14806481} & \num{146} &   \num{6986} \\
\num{7} &   \num{80} & \num{5} & \num{11254693} &  \num{261} &   \num{120229} & \num{358} &   \num{1758} \\
\num{7} &  \num{160} & \num{5} & \num{10235658} & \num{6062} &     \num{1246} & \num{440} &   \num{1220} \\
\num{7} &  \num{320} & \num{4} &  \num{8902425} & \num{6058} &      \num{986} & \num{858} &    \num{761} \\
\num{7} &  \num{640} & \num{4} &  \num{8929898} & \num{2692} &      \num{950} & \num{203} &    \num{691} \\
\num{7} & \num{1280} & \num{4} &  \num{4477059} &  \num{942} &      \num{901} & \num{336} &    \num{319} \\
\hline
\num{10} &   \num{40} & \num{4} & \num{13468801} &    \num{4} & \num{1460356} &   \num{3} & \num{8336} \\
\num{10} &   \num{80} & \num{5} & \num{11924783} &   \num{66} &  \num{326894} & \num{202} & \num{3646} \\
\num{10} &  \num{160} & \num{5} &  \num{7797436} & \num{1106} &    \num{2973} & \num{257} & \num{1134} \\
\num{10} &  \num{320} & \num{4} & \num{15762605} & \num{3456} &    \num{1053} & \num{112} &  \num{856} \\
\num{10} &  \num{640} & \num{4} & \num{12364702} & \num{1605} &     \num{977} & \num{312} &  \num{567} \\
\num{10} & \num{1280} & \num{4} &  \num{2904388} &  \num{532} &     \num{911} & \num{198} &  \num{287} \\
\hline
\num{13} &   \num{40} & -- & -- & -- & -- &   \num{4} & \num{3084} \\
\num{13} &   \num{80} & \num{4} & \num{14573149} &    \num{9} & \num{1375248} &  \num{39} & \num{1993} \\
\num{13} &  \num{160} & \num{4} & \num{10711647} &  \num{259} &    \num{5079} &  \num{44} & \num{2234} \\
\num{13} &  \num{320} & \num{4} & \num{10411006} & \num{2053} &    \num{1190} &  \num{52} & \num{1081} \\
\num{13} &  \num{640} & \num{4} &  \num{8829307} & \num{1100} &     \num{972} & \num{232} &  \num{568} \\
\num{13} & \num{1280} & -- & -- &  \num{374} &     \num{924} & \num{148} &  \num{332} \\
\hline
\end{tabular}
\captionof{table}{APM edit-distance locate-query pattern sets; en dashes mark omitted measurements.}
\label{tab:patterns_locate_apm_edit}
\end{minipage}
\vspace*{\fill}
\end{table*}

\end{document}